\documentclass{article}
\usepackage{graphicx} 
\usepackage{amsmath,amsthm,amsfonts,amssymb,amscd, stmaryrd,textcomp,mathscinet}
\usepackage{dsfont} 
\usepackage[T1]{fontenc}
\usepackage[utf8]{inputenc}

\usepackage{bbold}
\usepackage{bbm}
\newcommand{\1}{\mathbbm{1}}
\usepackage{lscape}

\usepackage[list=true, position=b]{subcaption} 

\usepackage{subfiles}
\usepackage[hidelinks]{hyperref} 
\usepackage{algorithm}
\usepackage{algorithmic}
\usepackage[mode=buildnew]{standalone} 
\usepackage{tikz}
\usetikzlibrary{hobby, decorations.markings, arrows.meta, automata, positioning, quotes, shapes.geometric,graphs, shapes, backgrounds, shapes.misc, calc, fit, decorations.pathreplacing, trees}
\usepackage{pgfplots}
\usetikzlibrary{decorations.pathreplacing,calligraphy}
\tikzset{anchorbase/.style={baseline={([yshift=-0.5ex]current bounding box.center)}}}
\usepackage{lscape}
\usepackage{comment}

\DeclareMathOperator*{\argmax}{arg\,max}
\DeclareMathOperator*{\argmin}{arg\,min}

\usepackage[]{xcolor}
\definecolor{aliceblue}{rgb}{0.94, 0.97, 1.0}
\definecolor{darkred}{rgb}{0.65,0.15,0.25}
\definecolor{asparagus}{rgb}{0.53, 0.66, 0.42}

\definecolor{grey}{HTML}{BBC3C8}
\definecolor{darkgrey}{HTML}{706F6F}
\definecolor{bleu_ql}{HTML}{026C80}
\definecolor{complementary}{HTML}{8C6881}
\definecolor{orange_punch}{HTML}{FFA500}
\definecolor{bleu_punch}{HTML}{0048cc}
\definecolor{pbm}{HTML}{E54648}
\definecolor{magenta}{HTML}{FF0080}
\definecolor{autre}{HTML}{3EBBBB}
\definecolor{chimio}{HTML}{4CA64C}
\definecolor{jump}{HTML}{825d67}
\definecolor{chrom1}{HTML}{0E172A} 
\definecolor{chrom2}{HTML}{EB811B} 
\definecolor{chrom3}{HTML}{3792A4} 
\definecolor{chrom4}{HTML}{2E518A} 
\definecolor{cutblack}{HTML}{333333}
\newtheorem{defin}{Definition}
\newtheorem{lem}{Lemma}
\newtheorem{theorem}{Theorem}
\def\E{\mathbb{E}}
\def\cd{\check{d}}

\title{Bridging Impulse Control of Piecewise Deterministic Markov Processes and Markov Decision Processes: Review and Open Challenges\footnote{This work was supported by the French Agence Nationale de la Recherche (ANR), grant ANR-21-CE40-005 (project HSMM-INCA) and the European Union’s Horizon 2020 research and innovation program (Marie Sklodowska-Curie grant agreement No 890462)}}
\author{Alice Cleynen$^{\ddag}$\footnote{John Curtin School of Medical Research, Australian National University}\and Benoîte de Saporta\footnote{IMAG, Univ. Montpellier, CNRS, Montpellier, France} \and 
Orlane Rossini$^{\ddag}$ \and
Régis Sabbadin\footnote{INRAE-MIAT, Univ Toulouse, Toulouse, France}\and
 Amélie Vernay$^{\ddag}$}
\date{}
\begin{document}
\maketitle     
\begin{abstract} 
Control theory plays a pivotal role in understanding and optimizing the behavior of complex dynamical systems across various scientific and engineering disciplines. Two key frameworks that have emerged for modeling and solving control problems in stochastic systems are Piecewise Deterministic Markov Processes (PDMPs) and Markov Decision Processes (MDPs). Each framework has its unique strengths, and their intersection offers promising opportunities for tackling a broad class of problems, particularly in the context of impulse control and sequential decision-making for complex systems. The relationship between PDMPs and MDPs is a natural subject of exploration, as embedding impulse control problems for PDMPs into the MDP framework could open new avenues for their analysis and resolution. Specifically, this integration would allow leveraging the computational and theoretical tools developed for MDPs to address the challenges inherent in PDMPs. On the other hand, PDMPs can offer a versatile and simple paradigm to model continuous time problems that are often described as discrete-time MDPs parametrized by complex transition kernels. This transformation has the potential to bridge the gap between the two frameworks, enabling solutions to previously intractable problems and expanding the scope of both fields. This paper presents a comprehensive review of two research domains, guided by a running medical example. The example is revisited and formalized within the framework of the various concepts and objects introduced.
\end{abstract}

\tableofcontents  
\newpage
%
\section{Introduction}
%
Control theory plays a pivotal role in understanding and optimizing the behavior of complex dynamical systems across various scientific and engineering disciplines. Two key frameworks that have emerged for modeling and solving control problems in stochastic systems are Piecewise Deterministic Markov Processes (PDMPs) and Markov Decision Processes (MDPs). Each framework has its unique strengths, and their intersection offers promising opportunities for tackling a broad class of problems, particularly in the context of impulse control and sequential decision-making for complex systems. 

PDMPs, introduced in the 80s by Davis\cite{davis_piecewise-deterministic_1984}, are a versatile class of continuous-time and hybrid (discrete and continuous) state space stochastic processes that combine deterministic dynamics with discrete stochastic events, allowing for a realistic representation of systems with hybrid behavior. Impulse control problems for PDMPs are of particular interest in applications ranging from resource allocation and maintenance planning to biological systems and energy management. These problems involve determining the optimal timing and magnitude of interventions, often modeled as impulses, to minimize a cost or maximize a reward. However, the complexity of PDMPs, particularly their reliance on continuous state spaces and non-trivial event-driven dynamics, poses significant challenges for both theoretical analysis and numerical resolution of control problems. 

On the other hand, MDPs, introduced in the 50s by Bellman \cite{bellman_dynamic_1958}, provide a well-established framework for sequential decision-making under uncertainty, originally leveraging the discrete state and action spaces to facilitate analytical and computational approaches. They have been extended to very general state and action spaces. Extensions of the classical MDP framework, such as semi-Markov Decision Processes, Partially Observed MDPs or Bayesian MDPs offer increased flexibility in modeling various real-world problems. 
MDPs have been instrumental in advancing the study of stochastic control, particularly through their compatibility with dynamic programming and approximate inference techniques. 

The relationship between PDMPs and MDPs is a natural subject of exploration, as embedding impulse control problems for PDMPs into the MDP framework could open new avenues for their analysis and resolution. Specifically, this integration would allow leveraging the computational and theoretical tools developed for MDPs to address the challenges inherent in PDMPs. On the other hand, PDMPs can offer a versatile and simple paradigm to model continuous time problems that are often described as discrete-time MDPs parametrized by complex transition kernels. This transformation has the potential to bridge the gap between the two frameworks, enabling solutions to previously intractable problems and expanding the scope of both fields.

This paper presents a comprehensive review of two research domains, guided by a running medical example and some related broad challenges. 
The paper is organized as follows. Section~\ref{sec:example} introduces the medical example and related decision-making, states the main challenges addressed throughout the paper and provides some  reading guidelines. Section~\ref{sec:PDMP} introduces PDMPs and their impulse control problems, highlighting the mathematical structure and the versatility of their hybrid nature. Section~\ref{sec:mdp} delves into the framework of MDPs and their extensions, introducing control policies and formulation of MDPs under partial observation or unknown models as well as state-of-the-art algorithms for their resolution. Section~\ref{sec:cpdmp} demonstrates how impulse control problems for PDMPs can be embedded into the MDP framework to leverage the best of both worlds. Section~\ref{ssec:open} concludes with the challenges solved and identifies a series of open questions that could be addressed through this integration, such as scalability, approximate solutions, and new modeling paradigms.
\color{black}

By bringing these two frameworks together, we aim to contribute to the growing field of research at the interface of stochastic control, decision theory, statistics, computational sciences and reinforcement learning,  offering insights into how embedding PDMPs into the MDP framework can advance the field and tackle outstanding challenges in control and optimization. 
%

\section{Motivating medical example} \label{sec:example}

We introduce cancer patient monitoring as the motivating application for this survey. Cancer is characterized by alternating phases of remission and relapse, and effective clinical management requires carefully balancing treatment intensity and frequency against the risk of adverse effects. This challenge calls for principled mathematical frameworks.
In this section, we describe the disease dynamics, the available monitoring and treatment options, and the key challenges arising in patient follow-up. We also provide a reading guide for the remainder of the paper. This medical case study serves as a running example throughout: each theoretical concept and methodology introduced is illustrated within this shared context, giving readers a consistent and concrete reference point. For the simpler scenarios, a companion Jupyter notebook is provided to support hands-on exploration of simulations and solution methods.
By anchoring the exposition in a single, coherent application, we aim to bridge the gap between mathematical formalism and clinical intuition, and to demonstrate the practical relevance of the proposed approaches for real-world decision-making under uncertainty.
%
\subsection{Disease dynamics and decisions}
\label{subsec:genericPDMPdyn}
%
We assume that disease progression can be tracked through the level of a one-dimensional \emph{biomarker} $\zeta$. When the biomarker is at its nominal value $\zeta_0$, the patient is considered to be in \emph{remission}. When it exceeds this nominal value, i.e., $\zeta > \zeta_0$, the patient is considered to be in \emph{relapse}, with higher levels corresponding to greater disease severity as a proxy of tumor volume.

Several therapeutic options are available: no treatment, surgery, or chemotherapy. In the absence of treatment, the biomarker remains at its baseline value $\zeta = \zeta_0$  during remission. After a random duration, a relapse may occur and the biomarker begins to rise. When surgery is performed, the tumor is partially or fully removed, resulting in an instantaneous decrease in the biomarker level. When chemotherapy is administered, the biomarker decreases continuously until it reaches the nominal value, where it then remains constant.

%
The controller makes sequential decisions of two kinds. The first concerns the choice of treatment among the available therapeutic options. The second concerns either the timing of the chosen treatment, typically relevant for surgery, or its duration, typically relevant for chemotherapy. Decisions are made based on the controller's general knowledge of the disease and on the patient's history, encompassing past and present observations.
 The objective is to reduce both disease and treatment burdens, so as to preserve a satisfactory quality of life for the patient. This is typically done by introducing an objective function interpreted as a cost that needs to be minimized over time, and that balances the costs of treatments and disease burdens. 




\subsection{Main challenges}
\label{subsec:genericPDMPchallenge}
%
Several challenges arise if one wants to solve this medical treatment optimization problem.\\

\textbf{Modeling} 
One needs a modeling framework that can take into account both deterministic and random evolutions. Typically, the marker level should grow exponentially fast under relapse and without treatment to match usual cell growth models \cite{sullivan1972kinetics,blagoev2021drug}, but the time spent in remission should be random.

Importantly, while marker measurements are typically not acquired in continuous time, the observation dates might not be regularly spaced or known a priori, hence the model should allow continuous time dynamics of the marker level with possible measurement errors at acquisition times. This should be done while maintaining low complexity models with a low number of interpretable parameters.

Yet several levels of complexity may be desired. For instance, one could either consider a single type of relapse, or several more or less aggressive types, or even a chemotherapy-resistant type. One could also consider the time spent in remission to be memoryless (exponential distribution), or alternatively with some memory, for instance with a higher risk of relapse as time goes by, or with a higher risk of resistant relapse if the patient has undergone multiple chemotherapies.
One may also want to take into account hard boundary effects. This would encode that hitting a high value threshold $\zeta=D$ for the marker level may render all treatments inefficient and lead to the death of the patient. \\

\textbf{Parameter calibration} 
Once a class of model is chosen, in order for the model to be useful the parameters should be calibrated based on available data or expert knowledge, for instance from available cohorts of patients from different treatment regimens. The model should therefore be easily fitted to data, with inference methods to calibrate parameters. 
Similarly, making optimal decisions requires designing and calibrating an appropriate cost function. For instance, one may penalize hospital visits to prevent them from occurring too frequently. Ongoing treatments such as chemotherapy could incur a cost proportional to their duration, while surgery would incur a single, substantial one-time cost.\\

\textbf{Decision limitations}
Decision-making is subject to at least two types of limitations: constraints and partial information. 
First, some decisions may be mandatory or unavailable depending on the patient's current state. As an example, a cycle of chemotherapy, once initiated, may be required to last a minimum duration before it can be discontinued.
Second, decisions must be based solely on available information, referred to here as the patient history. In an ideal setting, the biomarker level would be continuously and perfectly observed at all times. In practice, however, observations are made only at discrete time points — and potentially infrequently, if the measurement procedure is invasive — and may be corrupted by noise. Decisions must therefore depend on the patient history alone, rather than on the true underlying marker level.\\

\textbf{Model uncertainty} 
Next, once the model has been chosen and calibrated, together with the cost function and the decision constraints, the decision-maker is ready to tackle the challenging task of identifying an optimal treatment strategy. One may additionally wish to incorporate patient-specific flexibility, for instance by allowing certain parameters to vary across individuals rather than being fixed at the population level. In such cases, these parameters are unlikely to be known a priori, and the optimization framework must explicitly account for uncertainty in some or all of the model parameters. \\

\textbf{Resolution}
Even once the model, cost function, decision constraints, and any  parameter uncertainty have been fully specified, the practical  computation of an optimal or near-optimal strategy remains a  significant challenge in its own right. In the continuous-time, continuous-state setting of the medical example, computing the overall cost function becomes intractable, even with algorithmic techniques such as dynamic programming. When partial observation is further taken into account, 
resolution requires maintaining and propagating a belief distribution over the hidden states, a task that is computationally demanding even after discretization. Scalable resolution methods that can handle the full complexity of the problem, including continuous state spaces, partial observability, and model uncertainty jointly, are therefore essential.

\subsection{Reading guide}

Throughout this survey, the medical application introduced above serves as a running example to ground the concepts presented in each section. Section~\ref{sec:PDMP} introduces Piecewise Deterministic Markov Processes (PDMPs) as a broad class of continuous-time models, discusses the associated inference challenges, and presents impulse control as a natural framework for sequential decision making. 
Although this class of models is especially suitable to address the modeling challenge, the PDMP literature is very scarce on taking into account noisy observations and even scarcer on model uncertainty. Section~\ref{sec:mdp} presents a detailed introduction to a simpler class of discrete time and discrete state space controlled Markov chains called Markov Decision Processes (MDPs) and reviews how these models have been extended to take into account partial observations and model uncertainty, both in terms of optimization problem statement and their practical resolution. 
Then Section~\ref{sec:cpdmp} describes how impulse control of PDMPs can be embedded within the MDP framework, therefore allowing to address all the challenges listed above. 

 Readers already familiar with MDPs and primarily interested in continuous time modeling may start with Section~\ref{sec:PDMP}, skip Section~\ref{sec:mdp} and go directly to Section~\ref{sec:cpdmp}. 
%
%
Readers unfamiliar with stochastic control, or seeking a more foundational and progressive treatment of control problems, may want to start with the easier context of MDPs in Section~\ref{sec:mdp}, possibly skip the review on algorithms from the literature to solve complex MDPs and then turn to the more complex impulse problem for PDMPs in Section~\ref{sec:PDMP} before going on to Section~\ref{sec:cpdmp}, as highlighted on Fig. \ref{fig:read}.
\begin{figure}[tp]
\centering
\definecolor{chrom2}{HTML}{EB811B} 
\definecolor{chrom3}{HTML}{3792A4} 
\definecolor{chrom4}{HTML}{2E518A} 

\begin{tikzpicture}[
  node distance=1.5cm and 2.8cm,
  every node/.style={align=center, minimum width=3.5cm, minimum height=1.1cm},
  background/.style={rectangle, draw=none, fill=none},
  arrow/.style={->, thick}
]

\node (app) [inner sep=0.4cm, align = center] {\textbf{Section~\ref{sec:example}: Medical example }};

\node (mdp_1) [below right =1cm of app] {Models for partial observability \\ and model uncertainty \\ Algorithms for solving MDPs}; 
\node (mdp) [rounded corners, draw=black, fit=(mdp_1), inner sep=0.4cm, label=above:{\textbf{Section~\ref{sec:mdp}: MDP}}] {};



\node (pdmp_1) [below left =1cm of app, align = center] {Models for hybrid dynamics \\  Models in continuous time \\ Parameter calibration}; 
\node (pdmp) [rounded corners, draw=black, fit=(pdmp_1), inner sep=0.4cm, label=above:{\textbf{Section~\ref{sec:PDMP}: PDMP}}] {};

\node (cpdmp_1) [below=4cm of app, align = center] {Models for continuous time \\ and hybrid state space \\ and partial observability};
\node (cpdmp) [rounded corners, draw=black, fit=(cpdmp_1), inner sep=0.4cm, label=above:{\textbf{Section~\ref{sec:cpdmp}: PDMP as MDP}}] {};


\draw[arrow, chrom2] (app) -- (pdmp);
\draw[arrow, chrom2]  ([xshift=0.2cm]pdmp.south) |- ([yshift=0.2cm]cpdmp.west);

\draw[arrow, chrom4, dashed] (app) -- (mdp);
\draw[arrow, chrom4, dashed] (mdp.west) -- (pdmp.east);
\draw[arrow, chrom4, dashed]  ([xshift=-0.2cm]pdmp.south) |- ([yshift=-0.2cm]cpdmp.west);

\node[anchor=south east, align=left, inner sep=0.3cm, xshift=0.1cm]
at (current bounding box.south east) {
\begin{tikzpicture}[baseline]
  \draw[arrow, chrom2] (0,0) -- (1,0);
  \node[right] at (1,0) {\footnotesize Familiar with MDPs};
  
  \draw[arrow, chrom4, dashed, align=center] (0,-0.5) -- (1,-0.5);
  \node[right] at (1,-0.5) {\footnotesize Unfamiliar with control};
\end{tikzpicture}
};

\end{tikzpicture}
\caption{\label{fig:read} Alternative reading paths: blue dashed arrows show the recommended route for readers unfamiliar with stochastic control; orange arrows are for readers already familiar with MDPs. Otherwise, follow the standard progression of the paper.}  
\end{figure}
\color{black}
%
\section{Piecewise Deterministic Markov Processes} \label{sec:PDMP}
%
Piecewise Deterministic Markov Processes (PDMPs) are a general class of non-diffusion processes introduced by M. Davis in the 80's \cite{davis_piecewise-deterministic_1984} covering a wide range of applications, such as workshop optimization, queuing theory, internet networks, reliability, insurance and finance or biology for instance. These continuous-time processes  extend the class of point processes \cite{davis_markov_1993, jacobsen_point_2006,cocozza-thivent_markov_2021} and are characterized by deterministic motions punctuated by random jumps. They are especially suitable to model complex systems with interacting random phenomena.

 This presentation of PDMPs is divided into three main topics and mainly addresses the modeling challenge.
In Section \ref{subsec:genericPDMP} we give the definition of and formulate our motivating example as a PDMP, provide key properties for resolution such as simulation properties, and discuss their calibration. In Section \ref{subsubsec:examples}, we describe several applicative fields that rely on PDMP for the modeling, emphasizing their versatility, and explain why they appear as better choices than continuous-time Markov chains or stochastic differential equations in our context. In Section \ref{ssec:impulse}, we provide a survey on the impulse control problem for this class of processes, both from theoretical and practical resolution points of view. We conclude this part in Section \ref{sssec:transition} with our challenges answered by PDMPs and those remaining.
\color{black}
%
\subsection{PDMP, definition and properties}
\label{subsec:genericPDMP}
%

We start with the definition and generic properties of PDMPs.
In Section \ref{subsubsec:defPDMP}, we provide a formal definition of PDMPs and cast our running medical example as a PDMP in Section \ref{subsubsec:medicPDMP}.  Extensions to semi-Markov dynamics are explored in Section \ref{subsubsec:PDsMP}, with further application to the running example detailed in Section \ref{subsubsec:medicPDsMP}. 
We then delve into key questions and tools related to PDMPs for our challenges, beginning with simulation techniques in Section \ref{subsubsec:simulation}. The identification and analysis of specific sub-chains of interest are addressed in Section \ref{sssec:chains}, and a discussion about inference is presented in Section \ref{sssec:applications}.
\color{black}
%
\subsubsection{Generic definition} 
\label{subsubsec:defPDMP}
%
We start with a formal definition of a PDMP, its state space and its local characteristics.
\begin{defin}[PDMP]
    \label{def:pdmp}
    A Piecewise Deterministic Markov Process (PDMP) $X=(X_t)_{t\geq 0}$ is defined by a tuple $\langle E, \Phi, \lambda, Q\rangle$, where
    \begin{itemize}
        \item The \emph{state space} $E$ is a finite union of Borel subsets of $\mathbb{R}^d$, for some $d\geq 1$ and $\mathcal{B}(E)$ is its Borel $\sigma$-field. It is the set of all possible states of the process. Denote $\bar E$ the closure of $E$ and $\mathcal{B}(\bar{E})$ its Borel $\sigma$-field. 
        
        \item The \emph{flow} $\Phi$ is a continuous function from $E \times \mathbb{R}_+$ onto $\bar E$ satisfying a semi-group property, \textit{i.e.} $\Phi(\cdot, t+s)=\Phi(\Phi(\cdot, t), s)$ for all $s,t \in \mathbb{R}_+$. It prescribes the deterministic motion between jumps. It may be defined by an explicit function or described as the solution to a system of ordinary differential equations (ODEs).
        
        \item The jump \emph{intensity}, or hazard rate, or risk function $\lambda$ is a measurable function from $\bar E$ onto $\mathbb{R}_+$ that determines the occurrence of random jumps.
        
        \item The jump \emph{kernel} $Q$ is a Markov kernel on $(\mathcal{B}({E}),\bar{E})$ that selects the new location of the process after each jump.
    \end{itemize}
    The PDMP dynamics can be described informally as follows: starting from some initial point $x \in   E$, the motion of the process follows the deterministic ﬂow $t\mapsto\Phi(x,t)$ until a first jump time $T_1$, which may be deterministically triggered by the process reaching the boundary $\partial E$ of $E$, or may occur randomly with hazard rate $t\mapsto\lambda(\Phi(x,t))$. At $T_1$ the process jumps to a new point $x'={X}_{T_1}$ selected with probability ${Q}(dx'|\Phi(x,T_1))$ (conditional on $T_1$), and the motion restarts from this new point as before.
    The ﬂow $\Phi$, the jump rate $\lambda$, and the Markov kernel ${Q}$ are called the \emph{local characteristics} of the PDMP.
\end{defin}
%
\begin{figure}[tp]
  \begin{center}
  \scalebox{.9}{
\begin{tikzpicture}[use Hobby shortcut, anchorbase, decoration={markings, mark=at position 0.5 with {\arrow{>}}; }]
        \draw[closed, scale=1](-2,-1) .. (1,-2) .. (2,2) .. (0,1);
        \draw[closed, scale=1] (4,-2) .. (5,-1) .. (6,1.5) .. (6.5, 1) .. (7,0.5); 
        \node at (-1.5,-3.5) {$E$};
        \filldraw [gray, align=center] (-1.5,-0.5) circle (2pt) node (x0) [below=0.1]{$x_0$\\$t=0$};
        \draw[postaction={decorate}, align=center] (-1.5,-0.5) to [out=0,in=-180] (1,1) to [out=0, in =-180] (1.5,1.5);
        \node[align=center] at (2.0, 0.9) {$\Phi(x_0,T_1)$\\$t=T_1^-$} ;
        \draw (1.6,1.4) to [out=-180, in=-90] (1.5,1.5) to [out=90, in=-180] (1.6,1.6);
        \filldraw [gray] (6,1) circle (2pt) ;
        \node[gray, align=center] at (5.6, 0.8) {$X_{T_1}$} ;
        \filldraw [gray] (5.5,-1.8) circle (2pt);
        \node[gray, align=center] at (5.5, -1.4) {$X_{T_2}$} ;
        \draw (7.1,0.4) to [out=-180, in=-90] (7,0.5) to [out=90, in=-180] (7.1,0.6);
        \draw[dashed] (1.5,1.5) to [out=0,in=-180] (3,2) to [out=0, in =-180] (6,1) ;
        \draw[postaction={decorate}, align=center] (6,1) to [out=-75,in=-180] (6.5,-0.5) to [out=75, in =-180] (7,0.5) node [right=0.1] {$\Phi(X_{T_1}, T_2-T_1)$\\$t=T_2^-$};
        \draw[dashed] (7,0.5) to [out=-70,in=90] (7.2,-0.9) to [out=-90, in =20] (5.5,-1.8) ;
        \draw[postaction={decorate}] (5.5,-1.8) to [out=-145,in=90] (5,-3.3) ;
        \node at  (3.7,2.2) {\small{$Q(\cdot |\Phi(x_0,T_1)) $}};
        \node at  (8.7,-1.4) {\small{$Q(\cdot | \Phi(X_{T_1}, T_2-T_1)) $}};
      \draw[]  (-2.5, -4.3) rectangle (10.5, 2.7);
\end{tikzpicture}
      }
  \end{center}
      \vspace{-2cm}
      \caption{{\bf Trajectory a generic PDMP.} Starting from an initial value $x_0$ at time $t=0$, the process follows a deterministic trajectory until a jump occurs, either at a random time (as at $T_1$) or because the process reaches the state boundary (as at $T_2$). At jump times, the process jumps to a new location drawn from kernel $Q$. } \label{fig:PDMP}
\end{figure}
%
A generic representation of a PDMP is given in Figure~\ref{fig:PDMP}. It is particularly convenient, when dealing with PDMPs, to separate the state space $E$ and states $x$ into discrete and Euclidean components. In this case the state space is called \emph{hybrid}. Let $\mathsf{M}$ be a finite  set, and for all $\mathsf{m}\in \mathsf{M}$, let $\mathsf{E}_{\mathsf{m}}$ be some Borel subset of $\mathbb{R}^{d_\mathsf{m}}$, which dimension $d_\mathsf{m}$ is allowed to change with $\mathsf{m}$. The state space is thus
\begin{align*}
E = \bigcup_{\mathsf{m} \in \mathsf{M}}\{\mathsf{m}\}\times \mathsf{E_m}.
\end{align*}

Any state $x\in E$ can be written as $x=(\mathsf{m,x})$ with $\mathsf{m}\in \mathsf{M}$ and $\mathsf{x\in E_m}$. The discrete component $\mathsf{m}$ is called the \emph{mode} or \emph{regime} and $\mathsf{x}$ is called the \emph{Euclidean variable}. 
In between jumps, the mode $\mathsf{m}$ remains constant and only the Euclidean variable $\mathsf{x}$ evolves through the flow. At jump times, both the mode and the Euclidean variable may be allowed to change through the Markov kernel $Q$. 
In the hybrid setting, we set 
\begin{align*}
\Phi(x,t)=(\mathsf{m},\Phi_\mathsf{m}(\mathsf{x},t)),
\end{align*}
for all $x=(\mathsf{m,x})\in E$ and $t\in\mathbb{R}_+$. 
The jump intensity can be written as $\lambda(x)=\lambda_\mathsf{m}(\mathsf{x})$ for all $x=(\mathsf{m,x})\in E$ and $t\in\mathbb{R}_+$ and is such that for any $x$ in $E$, there exists $\epsilon >0$ such that
\begin{align}\label{eq:lambda-eps}
 \int_0^{\epsilon} \lambda(\Phi({x},t))dt < +\infty,
\end{align}
forbidding instantaneous jumps. 
Let $\partial E$ be the boundary of $E$: $\partial E= \cup_{\mathsf{m \in M}}\{\mathsf{m}\}\times \partial \mathsf{E_m}$, and $ t^*(x)$ be the deterministic time the flow takes to reach the boundary $\partial E$ of $E$ when it starts from $x$
\begin{align*}
t^*(x)=t^*_{\mathsf{m}}(\mathsf{x})=\inf\{t>0: \Phi({x},t)\in \partial E \}.
\end{align*}
Jumps may occur via two means: \textit{random jumps} occur from the realization of the random clock with intensity $\lambda$ (\textit{e.g.} time $T_1$ on Figure \ref{fig:PDMP}), while \textit{deterministic jumps} occur when the process reaches the boundary of the state space (\textit{e.g.} time $T_2$ on Figure \ref{fig:PDMP}). Thus, starting from $x$ at time $0$, the first jump time $T_1$ has the following distribution
\begin{align*}
\mathbb{P}_x(T_1>t)=\mathbb{P}(T_1>t| {X_0}=x)=
\exp \left(-\int_0^t \lambda(\Phi(x,s))ds\right)\1_{t<t^*(x)}, \quad t\geq 0.
\end{align*}
For convenience, we also define 
\begin{align*}
 \Lambda(x,t)=\Lambda_{\mathsf{m}}(\mathsf{x},t)=\int_0^{t} \lambda(\Phi({x},s))ds,
\end{align*}
the cumulative risk up to $t$ starting from $x$. 
Finally, the Markov kernel $Q$ on $(\mathcal{B}({E}),\bar{E})$ 
selects the new location after each jump. We also set 
\begin{align*}
Q_\mathsf{m}(\cdot|\mathsf{x})=Q(\cdot|x),
\end{align*}
for $x=(\mathsf{m,x})$ in accordance with other notations. 
It satisfies $Q(\{x\}|x)=0$ so that a jump has to change the state of the process (either the mode, or the Euclidean variable or both).

After the first jump time $T_1$, the process restarts from a new state selected by $Q$ and follows the flow until the next jump time $T_2$. The sojourn time $T_2-T_1$ until the next jump is drawn from the jump intensity and time to reach the boundary as before. At the jump time $T_2$, the post-jump location is selected by the Markov kernel $Q$ and so on iteratively, as described in Algorithm~\ref{algo:PDMPiter}. {This algorithm is fully explained in Section \ref{subsubsec:simulation}, especially line \ref{Alg1-l4} regarding the simulation of the inter-jumps times from the state dependent intensities}.
Note that Algorithm~\ref{algo:PDMPiter} can only simulate the process up to time $T_\infty=\lim_{n\to \infty} T_n$ that might be finite. This phenomenon is called \emph{explosion}. To avoid explosion, one usually also requires that for all $x\in E$, one has
\begin{align*}
\mathbb{E}_x[\lim_{n\to \infty} T_n]=+\infty.
\end{align*}
{Note that the condition given in Eq.~(\ref{eq:lambda-eps}) is not sufficient to prevent explosion. We refer to \cite[(24.4-6)]{davis_markov_1993} for a counterexample and a discussion on sufficient conditions.}
%
\begin{algorithm}[tp]
\caption{Iterative construction of a trajectory of a  PDMP starting from state $x_0=(\mathsf{m_0,x_0})$ at time $0$ and up to the $N$-th jump time}\label{algo:PDMPiter}
\begin{algorithmic}[1]
    \STATE $x\leftarrow x_0$
    \STATE $T\leftarrow 0$
    \FOR{$n$ $\gets 0$ to $N-1$}
        \STATE $S\sim \lambda({x})$ \label{Alg1-l4}
        \STATE $S \leftarrow \min\{S, {t}^*(x)\}$
        \STATE $X_{t} \gets \Phi(x,t-T)$ for $T\leq t <T+S$ \label{Alg1-l6}
        \STATE $x'\sim Q(\cdot|\Phi(x,S))$ \label{Alg1-l7}
        \STATE $x \gets x'$
        \STATE $T \gets T+S$
    \ENDFOR
\end{algorithmic}
\end{algorithm}

A process $X$ defined by a flow $\Phi$, an intensity function $\lambda$ and a kernel $Q$ verifying the above properties is well defined and is strong Markov, as is shown, for instance, in \cite{davis_markov_1993} or \cite{jacobsen_point_2006}.
%
\subsubsection{Running medical example as a PDMP}
\label{subsubsec:medicPDMP}
%
Let us now formulate the medical example of Section~\ref{sec:example} as a PDMP.
The evolution of the biomarker is modeled through a PDMP  $X_t=(m_t,\zeta_t)$, where $m$ will denote the condition of the patient and $\zeta$ the marker level.  At the start of monitoring, the patient is administered a treatment, the effect of which is to reduce the biomarker level exponentially fast. If the level reaches a certain fixed threshold $\zeta_0$, the patient is considered to be in remission and the marker remains constant. In the (random) event of a relapse, the biomarker level rises again exponentially fast. A sample trajectory is illustrated in Figure~\ref{fig:pdmp}. 
%
\begin{figure}[tp]
\centering
\includegraphics[width=.5\textwidth]{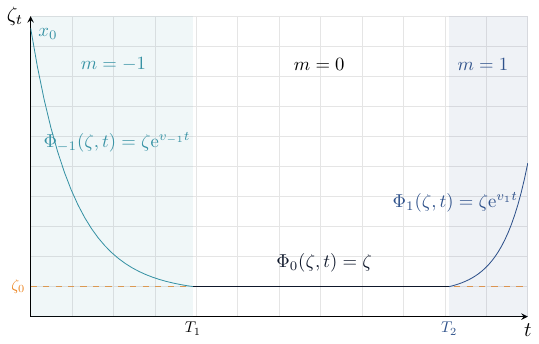}
\caption{\label{fig:pdmp} {\bf Trajectory of the running medical example as a PDMP.} The PDMP starts in mode $m=-1$ (patient under treatment) from an initial point $\mathtt{x}_0$ and follows a deterministic trajectory along its  exponential flow $\Phi_{-1}$ until the first jump occurs at time $T_1$ when reaching the boundary $\mathtt{x}=\zeta_0$ (remission). The mode switches to $m=0$ and the flow is constant equal to $\zeta_0$ until a new jump occurs at time $T_2$ with an exponential clock. The mode switches then to $m=1$ (relapse of the disease) and the trajectory rises exponentially along the flow $\Phi_1$.}    
\end{figure}

Thus, there are three possible modes for the patient corresponding to their condition. They can either be sick under treatment ($m=-1$), sick without treatment ($m=1$) or in remission ($m=0$). For simplicity, we assume that the process characteristics in remission mode are the same with or without treatment and do not differentiate the two cases. Set $M=\left \lbrace -1, 0, 1 \right \rbrace$ and ${E}_{-1} = (\zeta_0, +\infty)$, $E_0 = \left \lbrace \zeta_0 \right \rbrace$ and $E_{1} = \left[\zeta_0, +\infty\right)$. The state space of this PDMP is 
\begin{equation*}
    E = \bigcup_{m\in M}\{m\}\times E_{m}.
\end{equation*}
It is a hybrid state space with $\mathsf{M}=M$, $\mathsf{m}=m$, $\mathsf{x}=\zeta$ and $\mathsf{E_m}=E_m$ for all $m$.
Under treatment $m=-1$, the biomarker level $\zeta$ decreases exponentially with a slope $v_{-1}<0$. During relapse $m=1$, it increases exponentially with a slope $v_1>0$. For all $\zeta \in \mathbb{R}$ and for all $t\in\mathbb{R}_{+}$ we have
\begin{equation*} 
        \Phi_{-1}(\zeta, t) = \zeta\mathrm{e}^{v_{-1}t}, \qquad
        \Phi_{0}(\zeta,t)  = \zeta = \zeta_0, \qquad
        \Phi_{1}(\zeta, t)  = \zeta\mathrm{e}^{v_{1}t}.
\end{equation*}
%
In mode $m=-1$, a jump occurs when the subject reaches the fixed remission threshold $\zeta_0$. This is a deterministic jump at the boundary and $\lambda_{-1}(\zeta) = 0$. The jump time $t_{-1}^*(\zeta)$ is the solution of $\Phi_{-1}(\zeta,t)=\zeta_0$. That is, $t_{-1}^*(\zeta) = \frac{1}{v_{-1}}\mathrm{log}(\frac{\zeta_0}{\zeta})$. In mode $m=0$, jumps may only occur randomly, 
    $t_0^*(\zeta) = +\infty$ and the jump intensity $\lambda_0$ is a positive constant. Once the process reaches mode $m=1$, no more jump can occur and $\lambda_{1}(\zeta) = 0$, $t_1^*(\zeta) = +\infty$.
%
The Markov kernel is deterministic and given by
\begin{equation*} 
        Q((m^\prime, \zeta^\prime) \mid (-1,\zeta)) = \1_{\zeta=\zeta_0=\zeta^\prime}  \times \1_{m^\prime = 0}, \qquad
        Q((m^\prime, \zeta^\prime) \mid (0,\zeta))  = \1_{\zeta=\zeta_0=\zeta^\prime} \times \1_{m^\prime=1}.
\end{equation*}
%
\subsubsection{Semi-Markov dynamics} 
\label{subsubsec:PDsMP}
%
In the example above, when the patient is in mode $0$, the flow is constant. Hence, the jump intensity also has to be constant, \textit{i.e.} it corresponds to a memory-less exponential distribution. In this case, as in many other practical examples, see \textit{e.g.} \cite{barbu_semi-markov_2008}, the exponential distribution may not be realistic to model the underlying phenomenon. However allowing for more general distributions breaks the Markov property and falls into the class of \emph{semi-Markov} processes. 
It is especially easy to encompass semi-Markov dynamics into the PDMP framework by enlarging the state space. This new process is sometimes called the \emph{time-augmented process}, see \cite[Section 31.5]{davis_markov_1993} for a slightly different definition in the same spirit.
%
\begin{defin}[Time-augmented PDMP]
    \label{def:smdp}
    A time-augmented Piecewise Deterministic Markov Process is a PDMP defined by the tuple $\langle E, \Phi, \lambda, Q\rangle$, where
    \begin{itemize}
        \item The state space $E$ has the specific form 
\begin{align*}
E = \bigcup_{\mathsf{m} \in \mathsf{M}}\{\mathsf{m}\}\times \mathsf{E_m}\times \mathbb{R}_+,
\end{align*}
for some finite mode set $\mathsf{M}$, and where for all $\mathsf{m}\in \mathsf{M}$, $\mathsf{E}_{\mathsf{m}}$ is some Borel subset of $\mathbb{R}^{d_\mathsf{m}}$.
        
        \item The flow $\Phi$ has the specific form
\begin{align*}
\Phi(x,t)=(\mathsf{m},\Phi_\mathsf{m}(\mathsf{x}), u+t),
\end{align*}
for $x=(\mathsf{m},\mathsf{x},u)\in E$, where $\Phi_\mathsf{m}$ are flows from $E_\mathsf{m}\times \mathbb{R}_+$ onto $E_\mathsf{m}$ satisfying the semi-group property.
        \item The jump intensity just satisfies the conditions from definition \ref{def:pdmp} and is allowed to depend on all coordinates of $x\in E$.
        \item The jump kernel $Q$ has the specific form
\begin{align*}
Q(\{m'\}\times A \times B |x) = \1_{0\in B}Q_\mathsf{m}(A|\mathsf{x},u),
\end{align*}
for $x=(\mathsf{m},\mathsf{x},u)\in E$, $A$ any Borel subset of $\cup\{\mathsf m\}\times{\mathsf E_m}$, $B$ any Borel subset of $\mathbb{R}_+$, and where $Q_\mathsf{m}$ is a Markov kernel on $(\mathcal{B}(\cup\{\mathsf m\}\times{\mathsf E_m}),\overline{\mathsf{E}}_\mathsf{m})$.
    \end{itemize}
\end{defin}
The additional component $u$ can be interpreted as the time elapsed since the last jump. It moves along the flow at speed $1$ and is reset to $0$ at every jump.
One could define more sophisticated time-augmented processes where either $u$ is never reset to $0$ (it is thus the time since the beginning, as in \cite[section 31.5]{davis_markov_1993}),  resets less frequently than at each jump, or is multidimensional.
%
\subsubsection{Running medical example as a time-augmented PDMP}
\label{subsubsec:medicPDsMP}
%
We consider the same medical example as in Section~\ref{subsubsec:medicPDMP}, but now we want to allow the risk of relapse to depend on how long the patient has been in the remission phase. 
To this effect, we introduce the extra variable $u$ in mode $m=0$ representing the time elapsed since the entry into remission (\textit{i.e.} since the last jump to mode $m=0$). In mode $0$, the state space thus becomes $E_0=\{\zeta_0\}\times \mathbb{R}_+$, the flow is defined by $\Phi_{0}(\zeta,u,t)  = (\zeta_0,u+t)$, the intensity may now depend on $u$, for instance through a Weibull distribution \textit{i.e.} $\lambda_0(\zeta,u)=\beta u^\alpha$, and the Markov kernel becomes $Q(m^\prime, \zeta^\prime \mid (0,\zeta,0)) = \1_{\zeta=\zeta_0=\zeta^\prime} \times \1_{m^\prime=1}$.
%
In modes $m=-1$ and $m=1$, the state spaces, flows and jump intensities remain unchanged and the Markov kernel becomes $Q(m^\prime,\zeta^\prime,u^\prime  \mid (-1,\zeta)) = \1_{\zeta=\zeta_0=\zeta^\prime} \times \1_{u^\prime = 0}  \times \1_{m^\prime = 0}$. 
%
\subsubsection{Simulation}
\label{subsubsec:simulation}
%
A very interesting property of PDMPs is that they are particularly easy to simulate due to their intrinsic iterative construction as seen in Algorithm~\ref{algo:PDMPiter}. 
Computing the trajectory along the flow (Alg. \ref{algo:PDMPiter}, line \ref{Alg1-l6}) is straightforward if the flow is explicit. If it is only known as the solution to a system of ODEs, one just needs to use any suitable ODE solver. {Although this procedure is quite natural and very often used in practice, its impact on the simulation or Monte Carlo approximation errors is not addressed in the literature. To the authors' best knowledge, the only error study is in \cite{lemaire_thinning_2020} for a special case of PDMPs.}
Computing the post-jump location from the kernel $Q$ (Alg. \ref{algo:PDMPiter}, line \ref{Alg1-l7}) is also usually straightforward. In most applications, the kernel only changes the mode and leaves the Euclidean variable unchanged, or changes the Euclidean variable through some known mapping depending only on its value just before the jump, see examples in Section \ref{subsubsec:examples}. 

The only challenging step is then to simulate the jump dates from the non constant intensity (Alg.~\ref{algo:PDMPiter}, line~\ref{Alg1-l4}). We present here an exact simulation method called the \emph{Stochastic Simulation Algorithm} (SSA), or Gillespie algorithm \cite{gillespie_exact_1977} or \emph{thinning} algorithm that is especially suitable for PDMPs. The interested reader may consult \cite{cocozza-thivent_processus_1997,de_saporta_numerical_2016} for an alternative algorithm to simulate only sojourn-times and post-jump locations instead of full trajectories, and \cite{riedler_almost_2013,veltz_new_2015} and references therein for a broader discussion of alternative exact and approximate simulation methods. 
The SSA is based on the Markov property and the following property of the minimum of independent generalized exponential distributions. 
%
\begin{lem}\label{lem:min-expo-SSA}
Let $S_1$ and $S_2$ be two independent random variables with respective (non-constant) intensity functions $\lambda_1$ and $\lambda_2$. Set $S=\min\{S_1,S_2\}$ and $I=\argmin\{S_1,S_2\}$. Then the intensity of $S$ is $\lambda_1+\lambda_2$ and conditionally to $(S=t)$, the distribution of $I$ is $(\frac{\lambda_1(t)}{\lambda_1(t)+\lambda_1(t)},\frac{\lambda_2(t)}{\lambda_1(t)+\lambda_1(t)})$.
\end{lem}

\begin{proof}
First, notice that if $S_i$ has intensity $\lambda_i$, then its density function is $t\mapsto\lambda_i(t)e^{-\int_0^t \lambda_i(s)ds}$. Using the independence between $S_1$ and $S_2$, one obtains the joint distribution of $(S,I)$ as follows
\begin{align*}
    \mathbb{P}(S>t, I=1) 
    & = \mathbb{P}(S_1>t, S_2>S_1)\\
    & = \int_{t}^{\infty}\lambda_1(x)e^{-\int_0^x \lambda_1(s)ds}\int_{x}^{\infty}\lambda_2(y)e^{-\int_0^y \lambda_2(s)ds}dydx\\
    & = \int_{t}^{\infty}\lambda_1(x)e^{-\int_0^x \lambda_1(s)ds}e^{-\int_0^x \lambda_2(s)ds}dx\\   
    & = \int_{t}^{\infty}\frac{\lambda_1(x)}{\lambda_1(x)+\lambda_2(x)}(\lambda_1(x)+\lambda_2(x))e^{-\int_0^x (\lambda_1(s)+\lambda_2(s))ds}dx.
\end{align*}
Hence the density of $S$ is $(\lambda_1(x)+\lambda_2(x))e^{-\int_0^x (\lambda_1(s)+\lambda_2(u))ds}$, and 
\begin{align*}
    \mathbb{P}(I=1|S=t) = \frac{\lambda_1(x)}{\lambda_1(x)+\lambda_2(x)},
\end{align*}
and the result follows.
\end{proof}
%
The principle underlying the SSA algorithm is as follows. The aim is to sample from a random variable $S_1$ with (non-constant) intensity $\lambda_1$. Suppose that $\lambda_1$ is bounded from above by some constant $\overline{\lambda}$ and let $S$ be random variable with (regular) exponential distribution with parameter $\overline{\lambda}$, so that it is very easy to sample from $S$. In view of Lemma \ref{lem:min-expo-SSA}, $S$ can be interpreted as the minimum of $S_1$ and some virtual extra random variable $S_2$ with intensity $\lambda_2=\overline{\lambda}-\lambda_1$. Let $t$ be a sample value from $S$, then $t$ can be considered as a value sampled from $S_1$ with probability $\frac{\lambda_1}{\overline{\lambda}}$. This simple fact leads to a rejection algorithm where one repeatedly samples from $S$ until the result is accepted as a sample from $S_1$. In a Markov context, it is even more efficient as by construction $S_1\geq S$. Hence all samples are used, some are accepted as jump times (boundary jump: Alg.~\ref{algo:PDMP-SSA} line \ref{alg2:boundary} or random jump: Alg.~\ref{algo:PDMP-SSA} line \ref{alg2:random}), the others just move the process forward without creating a jump (Alg.~\ref{algo:PDMP-SSA} line \ref{alg2:no}). The procedure then restarts from the new location, as detailed in  Algorithm~\ref{algo:PDMP-SSA}.
%
\begin{algorithm}[tp]
\caption{Stochastic Simulation Algorithm for a PDMP starting from state $x_0=(\mathsf{m_0,x_0})$ at time $0$ and up to time horizon $H$}
\label{algo:PDMP-SSA}
\begin{algorithmic}[1]
    \STATE $x\leftarrow x_0$
    \STATE $t\leftarrow 0$
    \WHILE{$t<H$}
        \STATE $S\sim \text{Exp}(\overline{\lambda})$ 
        \STATE $X_{t+s} \gets \Phi(x,s)$ for $0\leq s < \min\{S,t^*(x)\}$
        \IF{$S>t^*(x)$} \label{alg2:boundary}
            \STATE $x'\sim Q(\cdot|\Phi(x,t^*(x)))$
        \ELSE
            \STATE $U\sim U[0,1]$
            \IF{$U\leq \frac{\lambda(S)}{\overline{\lambda}}$} \label{alg2:random}
                \STATE $x'\sim Q(\cdot|\Phi(x,S))$
            \ELSE \label{alg2:no}              
                \STATE $x' \gets \Phi(x,S)$               
            \ENDIF
        \ENDIF
        \STATE $x \gets x'$
        \STATE $t \gets t+\min\{S,t^*(x)\}$
    \ENDWHILE
\end{algorithmic}
\end{algorithm}

The choice of the upper bound $\overline\lambda$ can have a strong impact on the speed of the algorithm. If it is too high, it will keep proposing very short times and move forward extremely slowly. If at all possible, one should choose a tight upper bound for the intensity. {The interested reader is referred to \cite{lemaire2018exact} for an extensive discussion on the choice of global or locally optimized upper bounds for the intensity and their impact on simulation efficiency. Efficient generic strategies for unbounded intensities remain an open question.}
%
\subsubsection{Embedded chains}
\label{sssec:chains}
%
As for any continuous-time Markov process, one can define different discrete-time Markov chains embedded in a PDMP. Some of them are particularly interesting for our challenges and are specified below. \\

\textbf{Canonical chain.}
The canonical chain is the most natural Markov chain embedded in a PDMP. It corresponds to the post-jump-locations and inter-jump-times chains. Let us define $Z_0=X_{0}$ and $S_0=0$, and then iteratively, for all $n\in \mathbb{N}^*$, $Z_n=X_{T_n}$ the value of the PDMP at the $n^{th}$ jump time and $S_n=T_n-T_{n-1}$ the inter-jump time between the $n-1^{th}$ and $n^{th}$ jumps. Then $(Z_n,S_n)_{n\in\mathbb{N}}$ is a Markov chain called the \emph{canonical embedded chain} associated to the PDMP process $X$. Its transition kernel $P^{Z,S}$ is defined, for any Borel subsets $A$ of $E$ and $B$ of $\mathbb{R}_+$, by
\begin{align*}
    \lefteqn{P^{Z,S}(A\times B|x,t)=P^{Z,S}(A\times B|x)}\\
    &=\int_0^{t^*(x)}\lambda(\Phi(x,s))e^{-\Lambda(x,s)}Q(A|\Phi(x,s))\1_{B}(s)ds +e^{-\Lambda(x,t^*(x))}Q(A|\Phi(x,t^*(x)))\1_{B}(t^*(x)),
\end{align*}
where the first part of the equation corresponds to random jumps, while the second part corresponds to deterministic jumps. 
As the motion of the PDMP is deterministic between jumps, it is straightforward to reconstruct the full trajectory of the continuous-time PDMP from the knowledge of $(Z_n,S_n)_{n\in\mathbb{N}}$ as
\begin{align*}
X_t= \Phi\left(Z_{n},t-\sum_{i=0}^{n} S_i\right)  \quad \text{if } t \in \left[T_{n}=\sum_{i=0}^{n} S_i,T_{n+1}=\sum_{i=0}^{n+1} S_i \right),
\end{align*}
with the convention $T_0=S_0=0$. Thus, the canonical chain contains all the information of the process and is especially useful to obtain recursive formulations based on the Markov property. It is also at the heart of several numerical approximations for PDMPs, see \textit{e.g.} \cite{de_saporta_numerical_2016}.

For inference, this canonical chain can be extended with the pre-jump locations $\Phi(Z_{n-1},S_n)$, see Section~\ref{sssec:applications}. In some numerical procedures, it can also be extended with the jump times $T_n$. All these chains are especially interesting as they capture all the randomness of the process and have explicit transition kernels with a relatively simple form depending on the local characteristics of the PDMP. \\

\textbf{Skeleton chains.}
Other common embedded chains are the \emph{skeleton} chains defined by sampling the PDMP on a deterministic (but not necessarily homogeneous) grid $(t_1,\dots, t_n)$. We simply denote $X_n=X_{t_n}$, which once again is a (time inhomogeneous) Markov chain with kernels $P_{t_{n+1}-t_n}$.

While kernels $P_{t_{n+1}-t_n}$ might still have explicit forms, unless constraints are added on the PDMP there is no a priori limit in the number of jumps occurring between two time-points $t_n$ and $t_{n+1}$, which prevents a generic formulation of these kernels. Moreover, unless jump dates and locations are observed, it is not possible to reconstruct the original PDMP from the observation of $X_n$. However, in many application settings, $X_n$ might be the only available information from the process, see for instance the version of the running medical example developed in Section \ref{ssec:medic-pdmp-pomdp}. 
%
\subsubsection{Inference} 
\label{sssec:applications}

PDMPs are being studied by various communities in different frameworks.  Key directions in the literature include the study of long-time behavior and ergodic properties \cite{costa_stability_2008,cloez_probabilistic_2017}, large population asymptotics \cite{riedler_limit_2012,prodhomme_large_2024},  and their role as a unifying framework for continuous-time Monte Carlo methods \cite{fearnhead_piecewise_2018}. Here we briefly discuss statistical inference for PDMPs, which is one of the main challenges directly relevant to the motivating applications developed in 
this work.
\color{black}
Statistical estimation for PDMPs has been studied in the literature, but remains a challenging task in practice. A general overview of recent methods for statistical parameter estimation in PDMPs is given in \cite{azais_statistical_2018}. It is based on a wide range of application examples, in both parametric and non-parametric frameworks. While it is usually easy to estimate the flow of a PDMP, estimating the jump intensity or the kernel function can often be more difficult. We present here a few references dealing with the specific estimation of these characteristics in different contexts.

Examples of non parametric methods to estimate the jump mechanism under different settings and assumptions can be found in \cite{azais_nonparametric_2014,azais_new_2018,krell_statistical_2016}.
In \cite{azais_nonparametric_2014}, the authors propose an estimation procedure for the conditional distribution of the inter-arrival times of a PDMP. In \cite{krell_statistical_2016}, the author provide a method to estimate the jump rate of a process which state space is $(0, \infty)$ and where both the motion and the jump mechanism are deterministic. More recently, \cite{krell_nonparametric_2021} 
proposes an estimator of $\lambda$ 
without the assumption of a deterministic kernel. In all papers, a single trajectory of the process is fully observed in long time, and the PDMP is supposed to be ergodic.

In \cite{chiquet_method_2008} the authors present a method to estimate the transition kernel 
of a PDMP and compare it with the Monte-Carlo estimator on a reliability problem. \cite{azais_recursive_2014} gives a recursive estimation procedure to estimate the kernel of a PDMP from its perfect observations within a long time interval, without assumptions of stationarity. The works presented in \cite{azais_semi-parametric_2015} and \cite{azais_estimation_2019} focus on estimating PDMP-related functions rather than inferring intrinsic process characteristics. In all those works, the statistical approach varies according to the application domain, as each has its own specificity.

\subsection{PDMPs as versatile models}
\label{subsubsec:examples}
%
PDMPs are constructed using a small set of intuitive and easily interpretable local characteristics, making them a highly versatile and powerful modeling tool. They encompass and generalize many classes of Markov processes, including
discrete-time and space Markov chains (corresponding to a
    time-augmented PDMP with no Euclidean variables, null intensity, and jumps  triggered at a boundary at each time unit);
continuous-time, discrete-state space Markov processes (corresponding  to a PDMP with no Euclidean variables and constant intensity for each mode);
semi-Markov discrete-state space processes (corresponding to a time-augmented PDMP with no Euclidean variables and time-dependent
    intensity).

The key modeling strength of PDMPs lies in their capacity to jointly represent two qualitatively different types of dynamics: smooth, often physically or biologically motivated \emph{deterministic evolution} between events, and abrupt, \emph{stochastic transitions} that change the operating regime or instantaneously reset part of the state.  This hybrid structure positions PDMPs at a productive intermediate between several classical modeling frameworks.
Purely deterministic ODE models cannot account for the intrinsic randomness of discrete events such as failures, division, or infection.  Continuous-time Markov chains (CTMCs) handle stochastic discrete transitions but cannot represent continuous-state evolution without resorting to coarse discretization.
Stochastic differential equations (SDEs) introduce diffusion-type noise that is inappropriate when the randomness stems from rare, discrete, state-dependent events rather than from a large number of small independent perturbations.  The PDMP framework accommodates all of these features naturally, through the joint specification of the flow $\Phi$, the jump intensity $\lambda$, and the transition kernel $Q$. 
The following sections illustrate this point through five major application domains.  Figures~\ref{fig:SIR1}--\ref{fig:TCP} display sample trajectories for representative models in each domain.
%
\subsubsection{Epidemiology and population dynamics.}
%
Epidemic spread has long been modeled through compartmental frameworks such as the SIR model, where individuals belong to one of three classes: Susceptible, Infected, or Removed. In the large-population limit, the law of mass action yields a deterministic ODE system that describes the smooth evolution of the three compartment sizes.  However, when the population is small or moderate, stochastic effects become critical: a single infection event can ignite or extinguish an outbreak, and the timing of individual transitions matters. The standard stochastic SIR model is a CTMC in which each infection or recovery is treated as a discrete jump.

\begin{figure}[tp]
    \centering
    \includegraphics[width=0.80\linewidth]{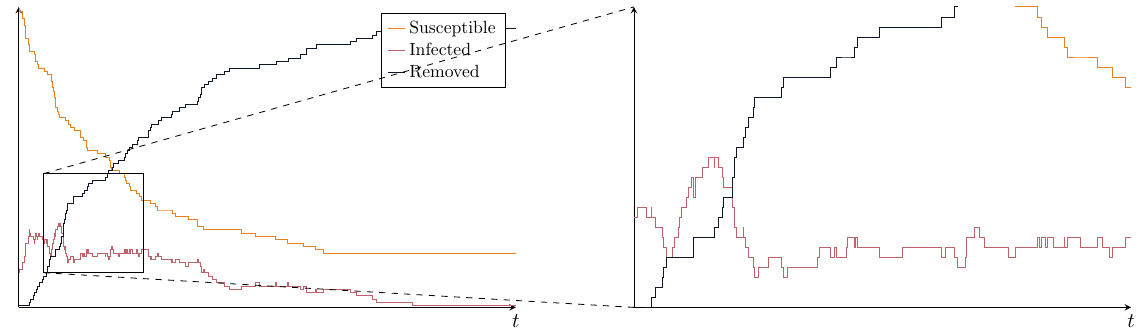}
    \caption{{\bf Sample trajectory of the stochastic SIR model.} The PDMP has no euclidean variables and 
    modes of the form $\mathsf{m}=(s,i)$ where $N$ is the total population size, $s$ the number of susceptible individuals, $i$ the number of infected individuals, and $r=N-s-i$ the number of removed individuals. 
    The right-hand-side figure is a detailed section of the main graph on the left emphasizing piecewise constant trajectories.  }
    \label{fig:SIR1}
\end{figure}

Viewed as a PDMP, the stochastic SIR model (Figure~\ref{fig:SIR1}) has no Euclidean variable and $(N+1)^{2}$ modes of the form $\mathsf{m} = (s,i)$, where $s$ and $i$ are the current numbers of susceptible and infected individuals, $N$ is the total population size, and the number of removed individuals is thus $N-s-i$.  The only possible transitions are infection, $\mathsf{m} \to (s-1, i+1)$, and removal, $\mathsf{m} \to (s, i-1)$, each occurring with a constant, mode-dependent intensity.  While this particular model has no Euclidean component, casting it in the PDMP framework immediately provides access to the full analytical and computational toolkit of PDMPs: the embedded canonical chain, the associated ergodic theory, and the impulse control and inference methods discussed in subsequent sections.  Furthermore, in many realistic epidemiological models, one wishes to track auxiliary continuous quantities (such as the concentration of a pathogen in the environment, the proportion of vaccinated individuals evolving under a vaccination campaign, or a biomarker level linked to disease severity) alongside the discrete compartmental state.  Such extensions fit naturally in the PDMP framework but are awkward to handle within a pure CTMC.
 
More generally, PDMPs arise as natural limit objects for individual-based population models when the population size is large but finite-size fluctuations remain relevant; see e.g.~\cite{doumic_statistical_2015,costa_piecewise_2016}.  In this regime, PDMPs capture stochastic effects that are entirely lost in the ODE limit. 
\subsubsection{Molecular biology and genomics.}
%
Many processes in molecular and cell biology combine continuous biochemical kinetics with discrete random events that switch the operating regime of the system.  PDMPs are ideally suited to this dual hybrid structure.\\

\textbf{Gene expression.}
Gene expression is a canonical example of a multi-scale hybrid dynamical system.
The transcriptional activity of a gene is regulated by its promoter, which switches stochastically between an active (\emph{on}) and an inactive (\emph{off}) state at rates that can depend on the concentrations of transcription factors.  When the promoter is active, mRNA is produced at a roughly constant rate; in both states, mRNA and the resulting protein evolve according to linear production-degradation ODEs whose parameters depend on the promoter state.  A PDMP model of this system (Figure~\ref{fig:geneexpr}) has two modes corresponding to the promoter state and a two-dimensional Euclidean variable tracking the mRNA and protein levels \cite{herbach_inferring_2017}.
The mode drives the flow of the Euclidean variable: the mRNA dynamics switch between a net production regime (promoter on) and a pure degradation regime (promoter off), producing the characteristic burst-like trajectories observed experimentally.  This bursty behavior in which mRNA is produced in short intense pulses separated by silent periods is a hallmark of stochastic gene expression that purely deterministic ODE models cannot reproduce and that CTMC models can only approximate by discretizing the mRNA and protein counts, which becomes computationally prohibitive for large copy numbers. The PDMP model 
retains continuous, analytically tractable dynamics for the copy numbers while preserving the discrete, stochastic nature of promoter switching.  This makes it well-suited for studying transcriptional noise, its propagation from mRNA to protein, and its dependence on promoter kinetics.
%
\begin{figure}[tp]
    \centering
    \includegraphics[width=0.50\linewidth]{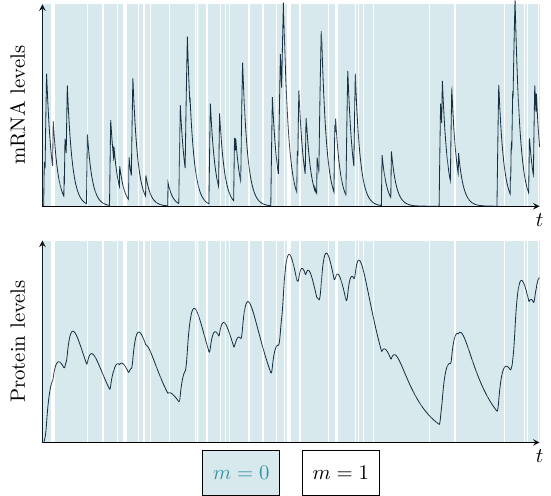}
    \caption{{\bf Sample trajectory of a gene expression model.} 
    The mode is the state of the promoter (0:off or 1:on), the first Euclidean variable is the quantity of transcripted messenger RNA (it increases when the promoter is on, and decreases otherwise), and the second Euclidean variable is the quantity of translated protein (that depends on the quantity of mRNA). The promoter switches mode with constant mode-dependent intensities. The model is from \cite{herbach_inferring_2017}.
    }
    \label{fig:geneexpr} 
\end{figure}
\\

\textbf{Cell growth and division.}
Individual cell growth and division follow a similar hybrid structure.  Between division events, the size of a cell grows deterministically according to an exponential law governed by the cell's metabolic rate.  Division occurs stochastically, with a rate that increases with cell size, reflecting the greater probability of initiating cytokinesis for a larger cell.  At each division event, the cell is replaced by two daughter cells, each receiving a random fraction of the mother cell's content, yielding an instantaneous downward jump in the tracked cell size.  The resulting sawtooth trajectory (Figure~\ref{fig:cellsizer}) is a PDMP with no mode, a one-dimensional Euclidean variable, and an exponential flow \cite{doumic_statistical_2015}.  The dependence of the jump rate on the current cell size is handled exactly by the PDMP intensity function $\lambda(x)$, whereas a pure CTMC model would require discretising the continuous size variable, introducing approximation error.
%
\begin{figure}[tp]
    \centering
    \includegraphics[width=0.50\linewidth]{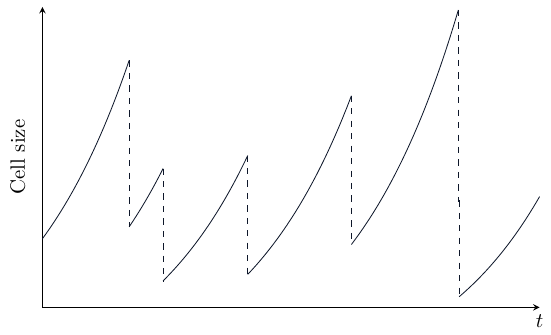}
    \caption{{\bf Sample trajectory of an individual cell growth and fragmentation model.} The PDMP has no mode, a 1-dimensional Euclidean variable with exponential flow corresponding to the size of a cell, and no boundary. The jump intensity increases with the cell size, and jumps correspond to divisions of the cell. The model is detailed in \cite{doumic_statistical_2015}.}
    \label{fig:cellsizer}
\end{figure}
\\

 \textbf{Bacterial chemotaxis.}
The zig-zag process (Figure~\ref{fig:zigzag}) provides a model of bacterial chemotaxis in a one-dimensional nutrient gradient.  A bacterium moves alternately (the two modes of the PDMP) towards and away from a nutrient source located at $0$, with its direction switching randomly at rates that are asymmetric: the switch rate is lower when the bacterium is moving towards the nutrient and higher when it is moving away, encoding the chemotactic bias that drives the cell towards more favorable regions \cite{fetique_long-time_2017}.  The continuous Euclidean variable tracks the distance to the nutrient source.  Pure diffusion models lose the directional persistence that characterizes run-and-tumble locomotion; pure CTMCs cannot represent the continuous spatial displacement between direction changes.
The zig-zag PDMP is thus the natural minimal model for this class of biased random walks, and has been the object of significant mathematical attention in the context of convergence to equilibrium and Monte Carlo methods \cite{fearnhead_piecewise_2018}.
%
\begin{figure}[tp]
    \centering
    \includegraphics[width=0.50\linewidth]{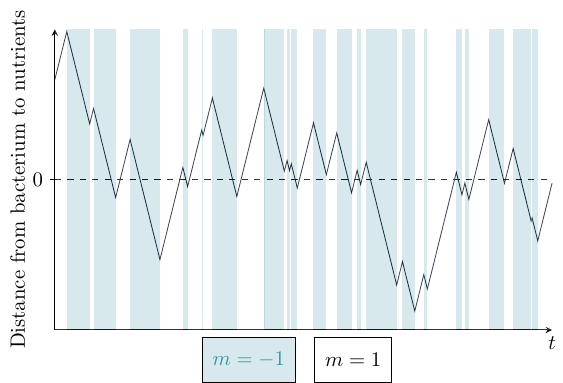}
    \caption{{\bf Sample trajectory of the Zig-zag model.} 
    The Euclidean variable corresponds to the distance between a bacterium and its nutrient, and the mode is its direction. It changes randomly with an intensity that depends whether the bacterium moves towards the nutrient (lower jump rate) or away from it (higher jump rate). The model is detailed in \cite{fetique_long-time_2017}.}
    \label{fig:zigzag} 
\end{figure}
%
\subsubsection{Neuroscience.}
%
Neuronal excitability is inherently a hybrid phenomenon. The membrane potential of a neuron evolves continuously according to a conductance-based ODE whose parameters depend on the state of ion channels embedded in the membrane.  These channels open and close stochastically, and it is their collective behavior that drives the generation of action potentials.  For a neuron with $N$ identical ion channels, the PDMP state consists of the mode $k \in \{0, \ldots, N\}$ (the number of currently open channels) and the membrane potential $V$, which is the Euclidean variable governed by an ODE that depends on $k$.  The Morris–Lecar model (Figure~\ref{fig:morrislecar}) is a prototypical instance of this construction \cite{riedler_limit_2012,lemaire_thinning_2020}. 
The PDMP formulation occupies a strategically important position in the modeling hierarchy.  At one extreme, treating each channel as an independent Markov chain and coupling all of them to the membrane potential yields an exact but computationally intractable model when $N$ is large.  At the other extreme, replacing the stochastic channel dynamics by their mean-field limit yields the classical deterministic Morris–Lecar or Hodgkin–Huxley equations, which lose all channel noise. The PDMP model retains the discreteness of channel gating — and thus captures variability in spike timing, stochastic resonance, and subthreshold membrane fluctuations that are biologically significant — while keeping the ODE structure for the potential, which makes the model both analytically and numerically tractable.  As $N$ goes to $\infty$, the PDMP converges to the deterministic ODE, providing a rigorous link between the stochastic and deterministic descriptions. 
\begin{figure}[tp]
    \centering
     \includegraphics[width=0.50\linewidth]{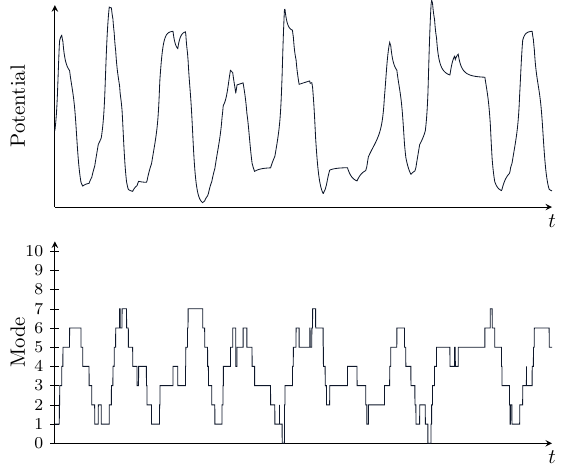} 
    \caption{\label{fig:morrislecar} {\bf Sample trajectory of the Morris Lecar model.} 
    The modes correspond to the number of open ionic channels along a neuron (for a total of $N=10$ channels), and the Euclidean variable is the electric potential of the neuron. The model is detailed in \cite{lemaire_thinning_2020}.}
\end{figure}
 %
\subsubsection{Reliability analysis and structural engineering.}
 %
Reliability analysis is one of the earliest and most established application domains of PDMPs, and was a central motivation in Davis's original theoretical work \cite{davis_markov_1993}.  Engineering systems degrade continuously over time due to physical processes (for instance due to fatigue, corrosion, etc) but also undergo discrete random events that change their degradation regime: environmental transitions, failures of protective mechanisms, or sudden material property changes.  Neither pure ODE models (which cannot account for discrete stochastic regime changes) nor pure CTMCs (which cannot represent the continuous physical degradation process accurately without coarse discretization) are fully adequate.  PDMPs capture both aspects jointly and parsimoniously.\\

\textbf{Crack propagation.}
In fracture mechanics, a crack in a structural component grows exponentially in time, with the growth rate governed by the material microstructure in the vicinity of the crack tip.  As the crack propagates, it may encounter regions of different material properties, causing the growth rate to change discontinuously; these transitions are modeled as random jumps between discrete modes.  The process terminates when the crack size exceeds a critical failure threshold — a natural \emph{boundary} in the PDMP state space.  The crack propagation model (Figure~\ref{fig:crack}) has three modes, an exponential flow, and a boundary at the failure threshold \cite{chiquet_piecewise_2009}, illustrating that boundaries are a first-class feature of the PDMP framework and arise naturally in reliability settings.  The ability to model failure as a boundary condition rather than as an additional absorbing state in a CTMC allows one to derive exact expressions for the distribution of the failure time.
%
\begin{figure}[tp]
    \centering
    \includegraphics[width=0.50\linewidth]{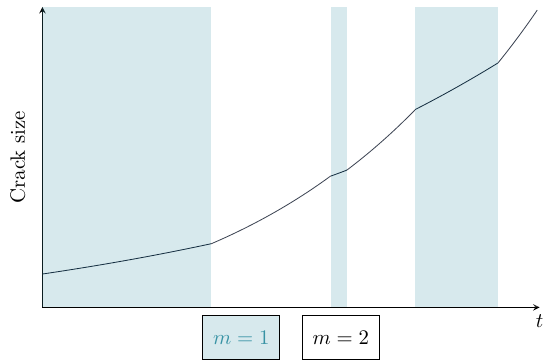}
    \caption{\label{fig:crack} {\bf Sample trajectory of a crack propagation model.} 
    The modes correspond to material characteristics influencing the propagation of cracks, and the Euclidean variable is the crack size. The process stops when the crack size reaches a failure threshold. The model is detailed in \cite{chiquet_piecewise_2009}.}
\end{figure}
\\
 
 \textbf{Corrosion.}
The corrosion model (Figure~\ref{fig:corrosion}) illustrates a more complex reliability application in which a structural component is protected by an anti-corrosion coating and is subject to three possible environmental regimes (e.g., dry, humid, saline) of varying aggressiveness \cite{de_saporta_optimal_2012}. The six modes of the PDMP encode the presence or absence of the protective coating and the current environmental state; the three-dimensional Euclidean variable tracks the cumulative thickness loss due to corrosion, the current aggressiveness of the environment, and the time spent under active protection. The succession of environmental regimes follows a deterministic cycle but with random durations, a feature that requires the joint use of both a continuous time-tracking variable and a discrete environment indicator.  This example highlights a general advantage of the PDMP approach in reliability: by augmenting the state space with continuous auxiliary variables (here, time in current environment), one can represent complex non-Markovian aging phenomena within a Markovian framework.
%
\begin{figure}[tp]
    \centering
    \includegraphics[width=0.50\linewidth]{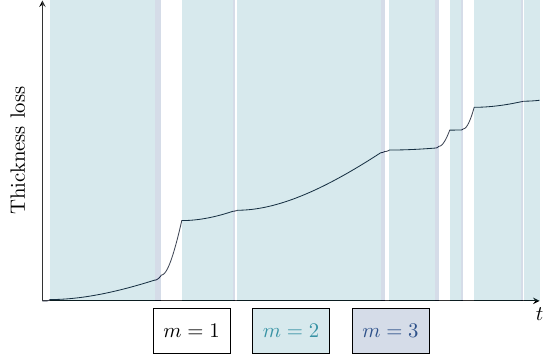}
    \caption{ {\bf Sample trajectory of a corrosion model.} 
    Modes $\mathsf{m}=(a,m)$ correspond to the presence ($a=1$)/absence ($a=0$) of an anti-corrosion protection (not shown) and three possible environments $m\in\{1,2,3\}$. The succession of environments follows a deterministic cycle with random durations. The Euclidean variables are the thickness loss due to corrosion, the aggressiveness of the environment (not shown) and the time spend with active anti-corrosion protection if $a=1$ (not shown). The model is detailed in \cite{de_saporta_optimal_2012}}
    \label{fig:corrosion}
\end{figure}
\subsubsection{Internet networks.}
%
The Transmission Control Protocol (TCP) governs data transmission rates over a network by maintaining a congestion window whose size $W(t)$ determines how many data packets can be in transit simultaneously.  In the absence of congestion, $W$ grows linearly in time as the protocol probes for available bandwidth.  When a packet is lost (signaling congestion) $W$ drops abruptly to a fraction of its current value.  The rate at which congestion events occur increases with $W$, since a larger window generates more traffic and is therefore more likely to saturate a link.  This dynamics defines a PDMP (Figure~\ref{fig:TCP}) with no mode, a one-dimensional Euclidean variable $W$ with linear flow, and jumps occurring with a state-dependent intensity \cite{chafai_long_2010}. 
The PDMP formulation of TCP has several advantages over alternative approaches. Purely deterministic fluid models capture the average window size evolution but cannot describe the variability of throughput across connections or the distribution of inter-loss times.  Fully discrete models require binning the window size, introducing both approximation error and a loss of the linear-growth structure.  The PDMP model, by contrast, preserves the exact linear drift and admits a closed-form stationary distribution for the window size, enabling exact computation of the mean throughput and higher-order statistics.  It has been used as a foundation for throughput optimization and active queue management problems
\cite{chafai_long_2010}.
%
\begin{figure}[tp]
    \centering
    \includegraphics[width=0.5\linewidth]{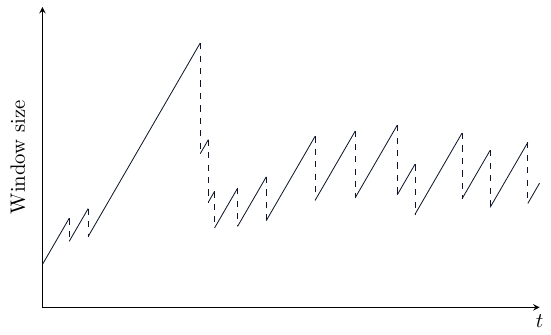}
    \caption{ {\bf Sample trajectory of the transmission control protocol (TCP) model.} The PDMP has no mode, a 1-dimensional Euclidean variable with linear flow corresponding to the available window size of a communications network, and no boundary. The jump intensity increases with the window size, and jumps correspond to random congestion resulting in a punctual decrease of the window size. The model is detailed in \cite{chafai_long_2010}.}
    \label{fig:TCP}
\end{figure}
%
\color{black}
\subsection{Impulse control for PDMPs}
\label{ssec:impulse}
%
In this section, we delve into the control of Piecewise Deterministic Markov Processes. Broadly speaking, a stochastic control problem involves influencing the dynamics of a process through control variables or decisions to either maximize a reward or minimize a cost. The primary questions of interest include the characterization and regularity properties of the \emph{value function} — representing the optimal reward or cost achievable across all possible decisions or decision sequences — and the existence and properties of optimal decision sequences, called \emph{strategies}. 
We begin by briefly outlining the two principal classes of control problems for PDMPs in Section~\ref{sssec:controlPDMP}. Our focus then shifts to impulse control problems, which we define in Section~\ref{ss:impulse}, illustrate with our running medical example in Section~\ref{ex:impulse_control_pdmp}, and address in terms of resolution strategies in Section~\ref{sssec:resolution_pdmp}. Finally, we conclude this part with a discussion on the special case of control problems involving hidden information in Section~\ref{sssec:POPDMP}.
%
\subsubsection{Continuous versus impulse control}
\label{sssec:controlPDMP}
%
There are two types of control for PDMPs: \emph{continuous control} and \emph{impulse control}, see \cite[Chapters 4 and 5]{davis_markov_1993}.
Continuous control describes situations in which the control variable is selected at \emph{all times} and acts on the process through the characteristics ($\Phi$, $\lambda$, and  $Q$) by influencing the deterministic motion and the jump mechanism. By opposition, impulse control intervenes \emph{punctually} on the process by moving it to a new point of the state space at some times specified by the controller. 

In continuous control problems, the local characteristics of the PDMP, the flow $\Phi$, the jump intensity $\lambda$ and the transition kernel $Q$, 
 depend on an additional control parameter that is selected by the controller at \emph{each instant of time}. 
Such problems have been studied in the literature mostly under the stationary condition by considering two strategies, the resolution of an integro-differential optimality inequality, the so called Hamilton-Jacobi-Bellman (HJB) equation (see for instance \cite{dempster_necessary_1992,dempster_generalized_1996}), or the translation into a discrete-time Markov Decision Process (MDP) on a functional space (see \cite{almudevar_dynamic_2001,bauerle_markov_2011,forwick_piecewise_2004,schal_piecewise_1998}). Continuous control problems where the process is not required to reach stationarity have also received some attention (see \textit{e.g.} \cite{costa_average_2010}) but remain seldom studied. 

In this paper we focus only on impulse control, where at \emph{discrete} controller-chosen times, the process is sent to a new controller-chosen point in the state space. The simplest form of impulse control is \emph{optimal stopping} where the  controller selects a single date at which the process is stopped.
Impulse control problems for PDMPs have been studied from the theoretical point of view: characterization of the value function and optimal policies through dynamic programming in \textit{e.g}. \cite{gugerui_optimal_1986,costa_approximations_1988,cloez_optimal_2020} for optimal stopping and \cite{costa_impulse_1989,de_saporta_optimal_2017} for the general impulse control problem. Numerical schemes to approximate the value function and compute explicit policies close to optimality have also been proposed in \cite{costa_impulse_1989,de_saporta_numerical_2016}.
%
\subsubsection{Definition of an impulse control problem}
\label{ss:impulse}
%
The formal probabilistic apparatus necessary to rigorously define the impulse control
problem for PDMPs is rather technical, and will not be used in the sequel, therefore, for the sake of simplicity, we only present an informal description of the problem. The interested reader is referred to \cite{costa_impulse_1989} or \cite{dufour_optimal_2016} for a proper definition.
%
\begin{defin}[Impulse strategy]
    \label{def:impulse_strat}
    A general impulse strategy $\mathcal{S}=(\tau_n,\chi_n)_{n\geq 1}$ is a sequence of non-anticipative $\mathbb{R}_+$-valued intervention times $\tau_n$ and $\mathbb{U}\subset E$-valued non-anticipative random variables $\chi_n$.
\end{defin}
%
    The trajectory of the PDMP controlled by strategy $\mathcal{S}$ is described in Algorithm \ref{algo:impulsePDMP}. Between intervention dates, the PDMP follows its natural trajectory. At the $n$-th impulsion (line \ref{alg3-impulse}), the process is instantaneously moved to the new starting point $\chi_n$ and restarts from there. 
    The set of admissible strategies is denoted by $\mathbb{S}$.
%
\begin{algorithm}[tp]
\caption{Simulation of a trajectory of a PDMP controlled by strategy $\mathcal{S}$ starting from state $x_0$ up to the $n$-th intervention}
\label{algo:impulsePDMP}
\begin{algorithmic}[1]
    \STATE $x\leftarrow x_0$
    \STATE $T\leftarrow 0$
    \FOR{$n$ $\gets 1$ to $N$}
        \WHILE{$T<\tau_n$}
            \STATE $S\sim \lambda({x})$ 
            \STATE $S \leftarrow \min\{S, {t}^*(x)\}$
            \IF{$S<\tau_n-T$} \label{alg3-natural}
                \STATE $X_{t} \gets \Phi(x,t-T)$ for $T\leq t <T+S$ 
                \STATE $x'\sim Q(\cdot|\Phi(x,S))$ 
                \STATE $x \gets x'$
                \STATE $T \gets T+S$
            \ELSE \label{alg3-impulse}
                \STATE $X_{t} \gets \Phi(x,t-T)$ for $T\leq t <\tau_n$ 
                \STATE $x' \gets \chi_n$
                \STATE $x \gets x'$
                \STATE $T \gets \tau_n$                
            \ENDIF
        \ENDWHILE
    \ENDFOR
\end{algorithmic}
\end{algorithm}
%
\begin{defin}[Cost of an impulse strategy]
    \label{def:cost_impulse_strat}
    The cost of strategy $\mathcal{S}$ for a PDMP starting at $x$ at time $0$ and up to infinity is defined by
    \begin{align*}\label{cout_strategie_PDMP}
    \mathcal{V}(\mathcal{S},x)=\E_x^{\mathcal{S}} \left[ \int_0^{+\infty} e^{-\gamma t}c_R(X_t)dt + \sum_{n=1}^{\infty}  e^{-\gamma \tau_n}c_I\left({X}_{\tau_n},{X}_{\tau_n^+}\right)\right],
    \end{align*}
    where $\gamma\geq 0$ is a discount factor, $c_R$ is the running cost and $c_I$ is the impulse cost.
\end{defin}
%
The specific assumptions that the cost functions must satisfy are omitted here, but can be found in \cite{costa_impulse_1989}. 
The discount factor $\gamma$ can either be interpreted as an economic inflation factor discounting long-term costs compared to short-term ones, or as a mathematical requirement to have finite expectations despite infinite horizon if running and impulse costs are bounded.
\color{black}
For simplicity, we present an infinite-horizon criterion as it is the most common in the literature. Alternative criteria are also possible, such as those with a fixed horizon $\mathsf{H}$ or a random horizon defined by a prescribed number of jumps or impulses, as discussed after Theorem~\ref{th:edynamic_programming_PDMP}. 
%
\begin{defin}[Value function]
    \label{def:value_PDMP}
    The value function for a PDMP starting at $x$ at time $0$ is defined by
    \begin{align*}
    \mathcal{V}^\star(x)=\inf_{\mathcal{S}\in\mathbb{S}}\mathcal{V}(\mathcal{S},x).
    \end{align*}
\end{defin}
%
In most impulse control problems for PDMPs, there is no optimal strategy as the infimum may not be reached. We define instead $\epsilon$-optimal strategies.
%
\begin{defin}[$\epsilon$-optimal strategy]
    \label{def:epsi-opti}
    An $\epsilon$-optimal strategy $\mathcal{S_\epsilon}$ is an admissible strategy satisfying, for all $x\in E$
    \begin{align*}
    \mathcal{V}(\mathcal{S_\epsilon},x)\leq \mathcal{V}^\star(x) + \epsilon.
    \end{align*}
\end{defin}
%
Under suitable assumptions (see again \cite{costa_impulse_1989} for details), the value function can be characterized as the unique solution of some dynamic programming equations.
%
\begin{theorem}[Dynamic programming]
\label{th:edynamic_programming_PDMP}
The optimal value function $\mathcal{V}^{\star}$ is limit of the following recursion. Set $\mathcal{V}^\star_0(x)=\mathcal{V}(\mathcal{S}_{\emptyset},x)$ and $\mathcal{V}^\star_{n+1}(x)=\mathcal{L}(\mathcal{V}^\star_{n})(x)$ for all $n \geq 0$, then for all $x$ in $E$, one has
\[\mathcal{V}^\star(x)=\lim_{n \rightarrow + \infty} \mathcal{V}^\star_{n}(x),\]
where $\mathcal{V}(\mathcal{S}_{\emptyset},x)$ is the cost associated to the no-impulse strategy $\mathcal{S}_\emptyset=(\tau_1=+\infty)$ starting from $x$
\begin{align*}
 \mathcal{V}(\mathcal{S}_{\emptyset},x)=\mathbb{E}_x^{\mathcal{S}_\emptyset}\left[ \int_0^{\infty} e^{-\gamma t} c_R(X_t)dt \right],
\end{align*}
and $\mathcal{L}$ is the single-jump-or-intervention operator defined as
\begin{align*}
\mathcal{L}V(x) = \inf_{t \in \mathbb{R}^+} J(MV,V)(x,t) \wedge KV(x),
\end{align*}
with
\begin{align*}
MV(x)&= \inf_{\chi\in \mathbb{U}} \{c_I(x,\chi)+V(\chi)\},\\
J(V,W)(x,t)&= \int_0^{t \wedge t^*(x)}\!\!\!\!\!\!\!\!\!\!\!\!\!\!e^{-\gamma s-\Lambda(x,s)} \left(c_R(\Phi(x,s))+\lambda(\Phi(x,s))QW(\Phi(x,s))\right) ds \\
&\quad + e^{-\gamma(t\wedge t^*(x))-\Lambda(x,t\wedge t^*(x))} V(\Phi(x,t\wedge t^*(x))), \\
KV(x)&= \int_0^{t^*(x)}e^{-\beta s-\Lambda(x,s)} \left(c_R(\Phi(x,s)) +\lambda(\Phi(x,s)) QV(\Phi(x,s))\right) ds \\
&+ e^{-\gamma t^*(x)-\Lambda(x,t^*(x))} QV(\Phi(x,t^*(x))).
\end{align*}
\end{theorem}
%
Operator $M$ selects the best starting point $\chi$ in the control set $\mathbb{U}$ after an intervention. Note that the control set could also depend on the position $x$ before the intervention. Operator $\inf _t J$ selects the best intervention date along the flow, while operator $K$ accounts for cases when it is best to wait for the next natural jump of the process before triggering an intervention. The iterations in the dynamic programming equation thus correspond either to natural jumps of the process or to interventions, hence the name \emph{single-jump-or-intervention} operator. Dynamic programming is based on the Markov property for the canonical embedded chain, and operators can be reformulated as expectations involving the said chain, see \cite{de_saporta_numerical_2012}. 
The iterate $\mathcal{V}^\ast_n$ can be interpreted as the optimal value function of the impulse control problem where at most $n$ jumps or interventions are allowed, and then strategy $\mathcal{S}_\emptyset$ is applied. 
Computing the cost of the no-impulse strategy on the whole space is demanding. Theorem \ref{th:edynamic_programming_PDMP} is also valid when replacing $\mathcal{V}(\mathcal{S}_{\emptyset},x)$ by any function $g$ such that $\mathcal{V}(\mathcal{S}_{\emptyset},x)\leq g(x)$ for all $x\in E$.

We have stated the results for an infinite horizon only for simplicity but also because other cases are embedded the infinite-horizon case, as seen above for the case where at most $n$ jumps or interventions are allowed. Similarly, if one is interested in controlling the process up to a finite time horizon $\mathsf{H}$, one just needs to kill the process, or send it to a cemetery state, when time $\mathsf{H}$ is reached, eventually pay a killing cost $c_T$, and define a null cost on this state to fall back to the infinite horizon framework and work with the equivalent infinite horizon problem. 
The discount factor $\gamma$ can be set to $0$ in this context as the interval of integration is finite. Hence, the cost of strategy $\mathcal{S}$ can equivalently be defined, in the fixed horizon $\mathsf{H}$, by
 \begin{align*}
    \mathcal{V}_{\mathsf{H}}(\mathcal{S},x)=\E_x^{\mathcal{S}} \left[ \int_0^{\mathsf{H}} c_R(X_t)dt + \sum_{\tau_n< \mathsf{H}}  c_I\left({X}_{\tau_n},{X}_{\tau_n^+}\right)+ c_T(X_H)\right],
 \end{align*}
where $c_T$ corresponds to some terminal cost function, and all the above results still hold.
\color{black}
%
\subsubsection{Running medical example as a controlled PDMP}
\label{ex:impulse_control_pdmp}
%
We revisit the medical example of Section~\ref{subsubsec:medicPDsMP}, now incorporating the additional option for the practitioner to perform surgery. This intervention entirely removes the tumor and restores the patient to a healthy state. The practitioner’s objective is to prevent the patient from reaching a critical state, which occurs when the tumor size exceeds a predefined threshold, leading to death. 
We add an extra mode $m=2$ to account for the death of the patient.
The state space becomes
\begin{equation*}
    E = \bigcup_{m\in M}\{m\}\times E_{m},
\end{equation*}
with $M=\left \lbrace -1, 0, 1, 2 \right \rbrace$ and ${E}_{-1} = (\zeta_0, D)$, $E_0 = \left \lbrace \zeta_0 \right \rbrace\times \mathbb{R}_+$, $E_{1} = \left[\zeta_0, D\right)$ and $E_2=\{D\}$, where $D$ is the death threshold. 
In this extended scenario, the biomarker dynamics are modeled by a PDMP as in Section~\ref{subsubsec:medicPDsMP} for modes $-1$, $0$ and $1$. In mode $m=2$, the flow is constant $\Phi_2(\zeta,t)=\zeta=D$, the jump intensity is null $\lambda_2=0$ so that no more jumps are allowed, \emph{i.e.} death is an absorbing state. At any time $\tau$, the practitioner may decide to intervene on the process by resetting the current state to $X_\tau=(0,\zeta_0,0)$, representing the post-surgery state. Impulse dates can thus be chosen freely as non-anticipative dates, and the control set is a single point $\mathbb{U}=\{(0,\zeta_0,0)\}$. 
The running cost is defined as the area under the marker curve, \textit{i.e.} $c_R(\zeta)=\zeta-\zeta_0$, while the impulse cost corresponds to a fixed surgery cost, $c_I(x,x')=C$, independent of the current marker value. Intuitively, this cost structure discourages the practitioner from performing surgery prematurely when the marker level is low. Instead, the decision is delayed until the marker level reaches a threshold where the cumulative cost of waiting outweighs the surgery cost. This setup suggests a threshold-based strategy, where surgery is performed when the marker exceeds a critical level $\zeta_T$. 
An example of a trajectory controlled by a threshold strategy 
is shown in Figure~\ref{fig:egimpulsecontrolledPDMP}. 
%
\begin{figure}[tp]
\centering
\includegraphics[width=.5\textwidth]{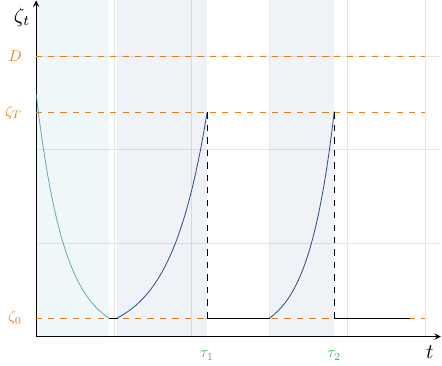}
\caption{{\bf Impulse control in the medical example.} The PDMP is controlled by a threshold strategy: a surgery intervention is performed whenever the marker level exceeds the threshold $\zeta_T$. Impulsion dates are denoted by $\tau_1$ and $\tau_2$ at which the tumor is completely removed, and the marker level is reset to its nominal value $\zeta_0$.}
\label{fig:egimpulsecontrolledPDMP}
\end{figure}
%
\subsubsection{Resolution of impulse control problems for PDMP}
\label{sssec:resolution_pdmp}
%
Optimal impulse control of PDMPs attracted attention almost as soon as the concept of PDMP was formalized by M. Davis. Contributions focusing only on the estimation of the value function are usually based on variational inequalities or value improvement and do not provide explicit strategies, see for instance \cite{lenhart_viscosity_1989, dempster_impulse_1995, gatarek_optimality_1992, dempster_optimal_1991, bandini_optimal_2018}. 

Explicit resolution strategies have also been proposed. In the simpler context of optimal stopping -- where the decision-maker can only decide to stop the process without restarting it from a new location -- the literature exhibits strategies based on either the discretization of the PDMP Markov jump kernel~$Q$ \cite{costa_approximations_1988}, or on the discretization of the Markov kernel of the canonical chain \cite{costa_stability_2008,de_saporta_numerical_2010} (see Section~\ref{sssec:chains}). After discretization, the value function can be estimated via dynamic programming \cite{gugerui_optimal_1986} using the computable discretized kernels.

These approaches have been extended to impulse control problems. \cite{costa_discretizations_1993} considers a discretization of the state space and time to obtain uniform convergence of the discretized problem to the original one. \cite{de_saporta_numerical_2012} consider optimal quantization of the kernel of the canonical chain to propose an $\epsilon$-optimal strategy. Other approaches, including iteration of a single-jump-or-intervention operator associated to the PDMP, have also successfully been proposed to construct explicit solutions in the infinite horizon context (see for instance \cite{costa_impulse_1989}). 

Theoretical results on the existence of $\epsilon$-optimal strategies not relying on the discretization of kernels  have also been derived in the case of infinite horizon impulse control, but they do not provide explicitly computable policies, see for instance \cite{costa_average_2010,costa_continuous_2013,de_saporta_optimal_2017}. 
All the approaches described above are numerically demanding and work for low dimensional PDMPs with limited complexity. They also require a perfect knowledge of the model and its parameters, together with a perfect observation of the process at all times (or equivalently at jumps times).
%
\subsubsection{Partially observed controlled PDMPs}
\label{sssec:POPDMP}
%
A very natural extension to the classical impulse control problem for PDMPs is to allow for hidden information. Suppose that the controller does not have access to the true value of the process, but only to some noisy observation of some of its components. Now decisions can only be taken based on the observations. 
The literature is much scarcer for this problem. 
Constructive resolution approaches have been proposed in the case where jumps are still perfectly observed, but the post-jump locations are only measured through noise. In \cite{brandejsky_numerical_2012}, the authors consider the equivalent fully observed process on the filter space and apply the quantization approach of \cite{de_saporta_numerical_2012} on the filtered process. \cite{bauerle_optimal_2018} reduce the problem to a discrete-time Markov Decision Process (MDP) and prove the existence of optimal policies, but provide no numerical approximation of the value function or optimal strategies. 
In the most challenging case where jump dates are unobserved, the authors have proposed a framework to construct an explicit strategy close to optimality in the optimal stopping context \cite{cleynen_change-point_2018} and for more general impulse problems \cite{cleynen_numerical_2021,de_saporta_monte-carlo_2024}. Those approaches are based on a conversion of the impulse problem into a (Partially Observed) MDP, as detailed in Section~\ref{subsec:POPDMP-POMDP}.
%

\subsection{Challenges addressed and challenges remaining}
\label{sssec:transition}
%
The PDMP framework introduced in this section directly addresses two of the challenges identified in Section~\ref{subsec:genericPDMPchallenge}. First, it provides a natural and parsimonious family of {\bf models} for hybrid phenomena evolving in continuous time: the deterministic flow captures the smooth dynamics of the biomarker between events, while the jump mechanism accounts for the discrete, stochastic transitions between disease conditions. Second, the growing literature on statistical {\bf inference} for PDMPs (Section~\ref{sssec:applications}) provides a principled basis for parameter calibration from observed patient trajectories, even when only skeleton observations are available.

However, the three remaining challenges {\bf decision under partial observation, model uncertainty, and resolution} are only partially addressed within the PDMP framework as it stands. While the theoretical formulation of impulse control for partially observed PDMPs is possible (Section~\ref{sssec:POPDMP}), its practical resolution is severely limited: when jump times are unobserved, one must resort to extensive discretizations of both the state space and the time horizon, which restricts tractable solutions to processes of very low complexity. Formal handling of model uncertainty, that is, optimal control when some parameters of the flow, jump intensity, or kernel are only partially known, has received little attention in the PDMP literature. More broadly, even in the fully observed setting with known parameters, computing an $\epsilon$-optimal strategy remains computationally demanding: existing exact methods rely on discretizations and scale poorly with state space complexity, while scalable approximate approaches remain scarce. 
To address these limitations, we turn in the following section to Markov Decision Processes, a framework that has been extensively developed for sequential decision-making under partial observability, model uncertainty, and computational constraints. MDPs operate in discrete time and, in their basic form, in discrete state spaces, but they come equipped with a rich algorithmic toolkit that is largely unavailable in the continuous-time PDMP setting. Section~\ref{sec:cpdmp} will then show how the two frameworks can be bridged, combining the modeling expressiveness of PDMPs with the resolution power of MDPs.
\color{black}
%
\section{Markov Decision Processes} \label{sec:mdp}
%
Markov Decision Processes (MDPs) provide a robust framework for modeling scenarios where sequential decisions are made in the presence of uncertainty. The objective is to determine an optimal sequence of decisions that maximizes rewards or minimizes costs, while accounting for the probabilistic evolution of the system. At each step, a decision-maker observes the current state, selects an action, and receives feedback based on that action, with the system transitioning probabilistically to the next state. 
MDPs belong to a class of stochastic control problems that originated in the 1950s \cite{bellman_dynamic_1958,howard_dynamic_1962}. Comprehensive treatments of the subject can be found in foundational texts such as \cite{boutilier_decision-theoretic_1999} and \cite{puterman_markov_1994}. Their ability to model complex, dynamic systems has made MDPs essential across a wide range of applications, from robotics and finance to healthcare and artificial intelligence. 
The versatility of MDPs in capturing stochastic dynamics and optimizing sequential decisions has led to their adoption across diverse fields, such as engineering, computer science, economics and other social sciences, making them a valuable tool for researchers and practitioners alike.
In robotics, MDPs have been employed for tasks such as navigation and control \cite{ponzoni_carvalho_chanel_multi-target_2013}.
Conversational agents and dialogue systems have leveraged MDPs to model and optimize human-machine interactions \cite{levin_using_1998,young_pomdp-based_2013}. 
The medical field has witnessed numerous applications of MDPs, including modeling disease progression and treatment planning for conditions like Alzheimer's disease \cite{hoey_automated_2010}, HIV \cite{keneally_markov_2016}, and diabetes \cite{roy_designing_2018}. 
MDPs have also been used in clinical decision-support systems for critical care \cite{schaefer_modeling_2004}.

In this section, we delve into the formal components of MDPs, including states, actions, and reward functions, while also introducing the concept of policies, which determine the action selection process (Section~\ref{subsec:genericMDP}) to address the modeling challenge. 
Additionally, we discuss various optimality criteria in MDPs. 
Unless otherwise stated, we consider finite spaces for both actions and states.
We then introduce several extensions of MDPs, such as the partially observed case (Section~\ref{sec:pomdp}) or cases where the model is partially unknown (Section~\ref{sec:bamdp}) to address the challenges related to decision limitations and model uncertainty respectively.
Finally, in Section \ref{sec:algo}, we present state-of-the-art algorithms for solving MDPs, exactly or approximately with regards   to the resolution challenge. We conclude this part in Section \ref{sec:cclMDP} with our challenges answered by MDPs and those remaining.
\color{black}
%
\subsection{Markov Decision Process definition and control}
\label{subsec:genericMDP}
%
We start with the formal definition of finite state and action spaces MDPs in Section~\ref{ssubsec:genericMDP}, then we turn to the definition of the different aspects of their optimization: policies (Section~\ref{subsubsec:policydefinition}), evaluation criteria (Section~\ref{subsubsec:evalcriteria}) and optimal policies (Section~\ref{sssec:optimalMDP}). We then illustrate these concepts with a new version of our running medical example in Section~\ref{ss:mini_mdp}, and briefly discuss extensions to more complex frameworks in Section~\ref{ss:extension_mdp}.
%
\subsubsection{Generic definition}
\label{ssubsec:genericMDP}
%
Informally, a Markov Decision Process (MDP) can be thought of as a Markov chain over a state space $S$, where transitions are influenced by actions chosen in a set $A$ by a decision-maker and evolve over an horizon of $H$ time steps. 
The decision process can span a finite or infinite number of time steps $H$, depending on the problem at hand. 
The stochastic transitions between states are governed by a transition probability matrix $P$ (or a sequence of matrices in the case of a time-inhomogeneous MDP). The associated costs of these transitions are captured in a (possibly time-dependent) cost function $c$, which quantifies the immediate cost or reward associated with each state-action pair.
%
\begin{defin}[MDP]
    \label{def:mdp}
    A Markov Decision Process is defined by a tuple $\langle S, A, H, P, c\rangle$ 
    \begin{itemize}
        \item $S$ is the set of all possible states in the environment, 
        $S$ may be a finite or countably infinite set. 
        
        \item $A$ is the set of actions (considered finite, unless otherwise precised) that the decision-maker can take. 
        
        \item $H\in\mathbb{N}\cup\{+\infty\}$ is the horizon of the process 
        counting the number of decision steps.
        
        \item $P$ describes the probabilities of transitioning, at any time step $0\leq t <H$ from a state $s_t \in S$ to another state $s' \in S$ after taking a specific action $a_t \in A$. Transition probabilities satisfy the Markov property
        $$P(s' | s_0, a_0, \cdots, s_t, a_t) = P(s' | s_t, a_t).$$ 
        In the case of finite state/action spaces, transition probabilities are typically 
        represented as transition matrices for all $a\in A$.
        When the MDP is non-stationary, $P=\left(P_t\right)_{t=0:H-1}$ is a family of time-indexed transition matrices.
        
        \item $c: S \times A \times S \rightarrow \mathbb{R}$ is the cost function that assigns a numerical cost to each state-action-state transition.
        Depending on the context, it may be worth considering a reward function instead of a cost function. 
        As for transitions, in the case of non-stationary MDPs, $c=\left(c_t\right)_{t=0:H-1}$ is a family of cost functions.
    \end{itemize}
\end{defin}
%
Let us consider a finite-horizon MDP $\langle S, A, H, P, c\rangle$ and an arbitrary initial state $s_0 \in S$. 
At any time $0\leq t<H$, the decision-maker chooses an action $a_t \in A$ and applies it in the current state $s_t$. This moves the process to the next state $s_{t+1}$, following probability $P(\cdot|s_t, a_t)$ and yields a cost $c_t(s_t, a_t, s_{t+1})$.
This process is iterated until termination when the horizon $H$ is reached. 
A terminal cost $C: S \rightarrow \mathbb{R}$ can be defined at the end of the process. This is illustrated on Figure \ref{fig:mdp}. 
%
\begin{figure}[tp]
    \centering
    \scalebox{.9}{
\begin{tikzpicture}[scale=0.75,
    node distance = 10mm and 15mm,
    state/.style={circle, draw, minimum size=1cm},
    decision/.style={draw, diamond, align=center, minimum size=1cm},
    reward/.style={draw, minimum size=1cm},
    every edge/.style = {draw, -{Stealth[scale=1.2]}}, 
    every edge quotes/.append style = {auto, inner sep=2pt, font=\footnotesize}
    ]
    \node (s0)  [state] {$s_0$};
    \node (s1)  [state, right=3cm of s0]   {$s_1$};
    \node (s2)  [state, right=3cm of s1]   {$s_2$};
    \node (s3)  [state, right=3cm of s2]   {$s_3$};
    \node (a0)  [decision, below right =of s0]   {$a_0$};
    \node (a1)  [decision, below right =of s1]   {$a_1$};
    \node (a2)  [decision, below right =of s2]   {$a_2$};
    \node (r0)  [reward, below =3cm of s1]   {$c(s_0,a_0,s_1)$};
    \node (r1)  [reward, below =3cm of s2]   {$c(s_1,a_1,s_2)$};
    \node (r2)  [reward, below =3cm of s3]   {$c(s_2,a_2,s_3)$};
    \node (r3)  [reward, dashed, below right =of s3]   {$C(s_3)$};

    \path   (s0)    edge [] (s1) 
                    edge [] (a0)
                    edge [bend left=-25] (r0)
            (s1)    edge [] (s2)
                    edge [] (a1)
                    edge [] (r0)
                    edge [bend left=-25] (r1)
            (s2)    edge [] (s3) 
                    edge [] (a2)
                    edge [] (r1)
                    edge [bend left=-25] (r2)

            (a0)    edge [] (s1)
                    edge [] (r0) 
            (a1)    edge [] (s2)
                    edge [] (r1) 

            (a2)    edge [] (r2) 
                    edge [] (s3)

            (s3)    edge [] (r2)
                    edge [dashed] (r3)
                    ;

\end{tikzpicture}

    }
    \caption{{\bf Graphical representation of a MDP with finite horizon $H=3$.} Circles denote states, diamonds represent actions and rectangles indicate costs associated with state-action-state transitions. The dashed rectangle depicts the terminal cost.}
    \label{fig:mdp}
\end{figure}
%
\subsubsection{MDP policies}
\label{subsubsec:policydefinition}
%
To control an MDP, it is essential to define a {\em policy}, which determines the actions to be taken at each decision step based on the past {\em history} of the process. We now introduce the concepts necessary to define, evaluate and choose the optimal policies.
%
\begin{defin}[MDP history]
Let an MDP $\langle S, A, H, P, c \rangle$ be given. 
For any $0< t<H$, any sequence $\langle s_0, a_0, ...,s_{t-1}, a_{t-1}, s_t\rangle$ of states in $S$ and actions in $A$ is called (length $t$) {\em history} of the MDP. A length $0$ history is a singleton $s_0\in S$.
The set of length $t$ histories is denoted ${\mathcal H}_t$.
\end{defin}
%
Let $\Delta(A)$ be the set of probability distributions over $A$.
%
\begin{defin}[Stochastic, history-dependent decision rules]
    \label{def:decision_rules}
    A stochastic, history-dependent decision rule at time $t$, $\pi_t: {\mathcal H}_t \to\Delta(A)$ assigns a probability distribution over actions to every length $t$ history.
\end{defin}
%
Interesting subclasses of decision rules are the deterministic and/or history-independent decision rules.
%
\begin{defin}[Deterministic decision rule]
    A decision rule $\pi_t$ is deterministic if it assigns a unique action to every history, \textit{i.e.} is a Dirac distribution over $A$.
\end{defin}
%
\begin{defin}[History-independent decision rule]
    A decision rule at time $t$, $\pi_t: {\mathcal H}_t \to\Delta(A)$ is history-independent if $\pi_t(h_t)=\pi_t(h'_t)$ for any pair $h_t=\langle s_0, a_0, ...,s_{t-1}, a_{t-1}, s_t\rangle, h'_t=\langle s'_0, a'_0, ...,s'_{t-1}, a'_{t-1}, s'_t\rangle\in{\mathcal H}_t$ such that $s_t=s'_t$. 
\end{defin}
%
MDP {\em policies} are time-indexed lists of decision rules that specify the behavior of a decision-maker in an MDP.
%
\begin{defin}[MDP policy]
    Let an MDP $\langle S, A, H, P, c \rangle$ be given. A policy for this MDP is a (potentially infinite) family of decision rules $\pi =\left(\pi_t\right)_{0\leq t < H}$.
    Let $\Pi_t$ denote the set of possible decision rules at time $t$ and $\Pi =\left(\Pi_t\right)_{0\leq t < H}$ denotes the full set of possible policies. 
\end{defin}
%
Again, specific subclasses of policies can be defined.
$\Pi^D$ is the subset of {\em deterministic policies}:  policies composed of deterministic decision rules only. 
 $\Pi^{M}$ is the subset of {\em Markovian policies}: policies composed of history-independent decision rules only. 
$\Pi^S$ is the subset of {\em stationary policies}: Markovian policies which do not depend on $t$, these are especially useful in the case of infinite-horizon MDP.
%
Given a policy $\pi\in \Pi$, it is straightforward to simulate a trajectory from an MDP using Algorithm~\ref{algo:controledMDP}.
%
\begin{algorithm}[tp]
    \caption{Simulation of a trajectory of a MDP controlled by policy $\pi$ starting from state $s_0$ up to the horizon $H$}
\label{algo:controledMDP}
\begin{algorithmic}[1]
    \STATE $t\leftarrow 0$
    \STATE $S_t \leftarrow s_0$
    \FOR{$t$ $\gets 0$ to $H-1$}
        \STATE $A_t \sim \pi(S_t)$ 
        \STATE $S_{t+1} \sim P(\cdot|s_t,a_t)$
        \STATE $C_t \gets c(S_t, A_t, S_{t+1})$
        \STATE $t \gets {t+1}$
    \ENDFOR
\end{algorithmic}
\end{algorithm}
%
\subsubsection{Policy Evaluation criteria}
\label{subsubsec:evalcriteria}
%
The objective in Markov Decision Processes (MDPs) is to find a policy that minimizes a cost criterion (or equivalently maximizes a reward criterion) over the entire decision horizon. The transition cost, $c(s_t, a_t, s_{t+1})$, represents the immediate penalty of taking action $a_t$ in state $s_t$ and transitioning to state $s_{t+1}$, while the terminal cost (in finite-horizon), $C(s_H)$, captures the penalty upon reaching the final state of the process.
Beyond immediate costs, we must define a criterion that captures the entire stochastic dynamics of the process, considering both the short-term effects of an action and its long-term consequences. 
The classical evaluation criterion is the (potentially discounted) {\em expected sum} of future costs. 
The optimization problem associated with a cost criterion involves finding a policy $\pi$ minimizing the expected sum of costs over the set of all possible state-action trajectories, following $\pi$.
We distinguish between 
finite-horizon and infinite-horizon optimization criteria. \\

\textbf{Infinite-horizon.} 
In this setting, the decision-maker makes decisions indefinitely. 
In this case, the discounted cost criterion represents the expectation of future transition costs, discounted by a factor $\gamma \in (0,1)$, which balances the immediate cost against future outcomes and ensures the finiteness of the sum as long as the running cost is bounded.
%
\begin{defin}[Discounted expected cost]\label{disexpcost}
    Starting from the initial state $s \in S$ and following the policy $\pi \in \Pi$ the discounted expected cost is 
\[V_{\gamma}(\pi, s) = \mathbb{E}^{\pi}_s\left[\sum_{t=0}^{\infty} \gamma^t c(S_t, A_t, S_{t+1})\right].\]
\end{defin}
%
The average cost criterion, introduced in \cite{puterman_markov_1994}, evaluates policies by comparing the limiting average (over time) of incurred expected costs.
%
\begin{defin}[Average expected cost]
     Starting from the initial state $s \in S$ and following the policy $\pi \in \Pi$ the average expected cost is 
\[V_{\text{av}}(\pi, s) = \lim_{H\to+\infty} \frac{1}{H}\mathbb{E}^{\pi}_s\left[\sum_{t=0}^{H-1}  c(S_t, A_t, S_{t+1})\right].\]
\end{defin}
%
\textbf{Finite-horizon.}
In this setting, the decision-maker is tasked with controlling the system over a fixed number of steps, $H$. The natural objective in this case is to minimize the {\em total expected cost}:
%
 \begin{defin}[Total expected cost]\label{totexpcost}
 Starting from the initial state $s \in S$ and following the policy $\pi \in \Pi$ the total cost is the sum of the costs accumulated over each step, up to the horizon $H$, plus the terminal cost associated with the final state. The expected cost can be expressed as
\[V_H(\pi, s) = \mathbb{E}^{\pi}_s\left[\sum_{t=0}^{H-1} c_t(S_t, A_t, S_{t+1}) + C(S_H)\right].\]
\end{defin}
%
Note that the above evaluation criteria apply to general history-dependent stochastic policies as well as all subclasses of policies.
%
\subsubsection{Computing optimal policies}
\label{sssec:optimalMDP}
%
When studying an MDP, one is interested in computing an {\em optimal policy}, that is a policy with minimal value in every starting states.
%
\begin{defin}[Optimal policy]
    Let MDP $\langle S, A, H, P, c \rangle$ be given. An optimal policy $\pi^*$ of the MDP, with respect to criterion $V\in\{V_\gamma, V_{\text{av}}, V_H\}$ is a policy such that:
    $$
    V(\pi^*,s)\leq V(\pi,s),\  \forall \pi\in\Pi,\ \forall s\in S.
    $$
\end{defin}
%
Optimal policies are not necessarily unique for a given problem. However, it is well-known (see \textit{e.g.} \cite{puterman_markov_1994}) that, in the finite state and action spaces case, 
an optimal policy exists, for any MDP and criterion $V\in\{V_\gamma, V_{\text{av}}, V_H\}$; 
an optimal stationary deterministic policy exists for $V\in\{V_\gamma, V_{\text{av}}\}$, in the infinite horizon case;  
an optimal (potentially non-stationary) Markovian deterministic policy exists for $V_H$, in the finite-horizon case.
%
Even though optimal policies need not be unique, they all share the same {\em value function} $V^*$, which is unique.
%
\begin{defin}[Value function]
Let MDP $\langle S, A, H, P, c \rangle$ be given. Let any optimal policy $\pi^*$ of the MDP, with respect to criterion $V\in\{V_\gamma, V_{\text{av}}, V_H\}$ be also given. 
The optimal value function (which is unique) is 
\[V^{\star}(s) = V(\pi^*, s) = \inf_{\pi\in \Pi}V(\pi,s),\ \forall s\in S.\]
\end{defin}
%
The optimal value functions and optimal policies of an MDP 
can be obtained by solving sets of non-linear equations.
For the infinite-horizon discounted criterion and for the finite-horizon total total cost criterion, the corresponding systems of equations are the so-called \emph{Bellman equations} (see \textit{e.g.} \cite{bellman_dynamic_1958}).
%
\begin{theorem}[Bellman equations]~\\
\label{th:mdp-bellman}
$\bullet$ infinite horizon: the optimal value function $V_\gamma^{\star}$ is the unique solution of a system of non-linear equations in variables $\{V_\gamma(s)\}_{s\in S}$, called Bellman equations
\[V_{\gamma}(s) = \min_{a \in A} \sum_{s' \in S}P(s'|s,a) \left(c(s,a,s') + \gamma V_{\gamma}(s')\right),\ \forall{s \in S}.\]
Any optimal policy is such that 
\[\pi^{\star}(s)  \in \argmin_{a \in A}\sum_{s' \in S}P(s'|s,a)\left(c(s,a,s')+\gamma V_{\gamma}^{\star}(s')\right), \ \forall{s \in S}.\]
$\bullet$ finite horizon: the optimal value function $V^{\star}$ is the unique solution to the following system of Bellman equations in variables $\{V_t(s)\}_{0\leq t\leq H, s\in S}$ 
\[V_t(s) = \min_{a \in A} \sum_{s' \in S}P(s'|s,a)\left(c(s,a,s') + V_{t+1}(s')\right),\ \forall{t \in \{0, \ldots, H-1\}}, \forall{s \in S}\]
where $V_H(s) = C(s)$ defines the terminal values. 
Any optimal policy $\pi^{\star}=(\pi^\star_t)_{t=0:H-1}$ is such that 
\[\pi^{\star}_t(s)  \in \argmin_{a \in A}\sum_{s' \in S}P(s'|s,a)\left(c(s,a,s')+V_{t+1}^{\star}(s')\right), \ \forall{s \in S}.\]
\end{theorem}
%
There exists a similar defining set of equations for the average criterion, but it is omitted here due to its more technical aspects. One may refer to \cite{puterman_markov_1994} for a definition. 

\subsubsection{Running medical example as a finite state and action MDP} 
\label{ss:mini_mdp}
%
Let us consider a simplified model of the patient’s disease progression using an MDP with a finite state space. 
 The patient's state $s = (m, \xi) \in S$ has two components: $m$ still represents the patient's condition, where $m=0$ represents remission, $m=1$ indicates slow relapse, $m=2$ indicates aggressive relapse, and $m=3$ corresponds to death and is an absorbing state, and $\xi$ corresponds to the logarithm of the blood marker value $\xi=\log\zeta$, assumed now to be an integer, belonging to $\llbracket 0~;~ 40 \rrbracket$. In the log scale, the evolution of $\xi$ is linear.
\color{black}

An action $a$ corresponds to the treatment. 
Without treatment ($a=0$), 
if the patient is in remission $(m=0)$, the marker level stays at $\xi'=\xi=0$, the patient remains in a healthy condition ($m'=0$) or transitions to a slow ($m'=1$) or an aggressive relapse ($m'=2$) with probabilities $p_{m'}^0$; 
 during a slow relapse the logarithm of the biomarker $\xi$ increases by $1$ unit at each time step $\xi'=\xi+1$, if the biomarker value reaches 40, the patient dies, \textit{i.e.}, $m'=3$; 
 during an aggressive relapse, the logarithm of the biomarker increases by $2$ units at each time step $\xi'=\xi+2$, if the biomarker value reaches 40, the patient dies, \textit{i.e.}, $m'=3$. \color{black} 
%
When treatment is applied $(a=1)$, 
     if the patient is in remission $(m=0)$, the marker level stays at $\xi'=\xi=0$, and no relapse is possible; 
    for slow ($m=1$) or aggressive ($m=2)$ relapses, the biomarker decreases by $1$ unit at each time step $\xi'=\xi-1$. 
    If the biomarker reaches $0$, the patient returns to remission, \textit{i.e.}, $m'=0$ and the marker value remains at $\xi'=0$. 

At each time step, a cost is incurred  based on the current state and the chosen action. Relapses incur a cost, with aggressive relapses ($m=2$) having higher costs than slow relapses ($m=1$). Treatment ($a=1$) also incurs a cost. A terminal cost is assigned at the end of the process if the patient has died ($m=3$). 
The follow-up process is scheduled over $H=50$ visits. 
The MDP associated to this example is described by
\begin{align*}
    S&=\{(0,0)\} \cup \{1,2\}\times\llbracket 0~;~39\rrbracket \cup \{(3,40)\}, \quad A=\{0,1\}, \quad H=50, \\
    P&=\left\{\begin{aligned}
        P(s'|(0,0),0)&=p_0^0\1_{\{s'=(0,0)\}}+p_1^0\1_{\{s'=(1,0)\}}+p_2^0\1_{\{s'=(2,0)\}} \\
        P(s'|(0,0),1)&=\1_{\{s'=(0,0)\}} \\
        P(s'|(1,\xi),0)&=\1_{\{s'=(1,\xi+1)\}}\1_{\xi<D-1} + \1_{\{s'=(3,D)\}}\1_{\xi=D-1} \\
        P(s'|(1,\xi),1)&=\1_{\{s'=(1,\xi-1)\}} \1_{\xi>1} + \1_{\{s'=(0,0)\}}\1_{\xi=0} \\
        P(s'|(2,\xi),0)&=\1_{\{s'=(2,\xi+2)\}}\1_{\xi<D-2} + \1_{\{s'=(3,D)\}}\1_{\xi\geq D-2} \\
        P(s'|(2,\xi),1)&=\1_{\{s'=(2,\xi-1)\}} \1_{\xi>1} + \1_{\{s'=(0,0)\}}\1_{\xi=1} \\
        P(s'|(3,D),0)&= P(s'|(3,D),1) = \1_{\{s'=(3,D)\}}, \\
    \end{aligned} \right. \\
    c(s,a)&=a +  \times \1_{m=1} + 2\times \1_{m=2}, \quad C(s)=200 \times \1_{m=3}. 
\end{align*}
%
Note that, for this example, the cost function is time-independent. 
%
%
 
We consider three reference policies for controlling this MDP: $\pi_0$ never administers treatment; $\pi_1$  administers treatment randomly and $\pi_2$ is a threshold policy that treats when the biomarker exceeds 3 (until relapse). Their behaviors are illustrated in Figure~\ref{fig:trajectories} through sample trajectories. 
The optimal policy $\pi^*$ is computed via dynamic programming (see Section~\ref{sub:DP} and the accompanying Jupyter notebook for details). 
It corresponds to a strategy in which treatment is administered as soon as the disease onset is detected. A visual comparison of policy performances is provided in Figure~\ref{fig:costs}.
\begin{figure}[tp]
    \centering
    \begin{subfigure}{0.45\textwidth}
        \centering
        \includegraphics[width=\textwidth]{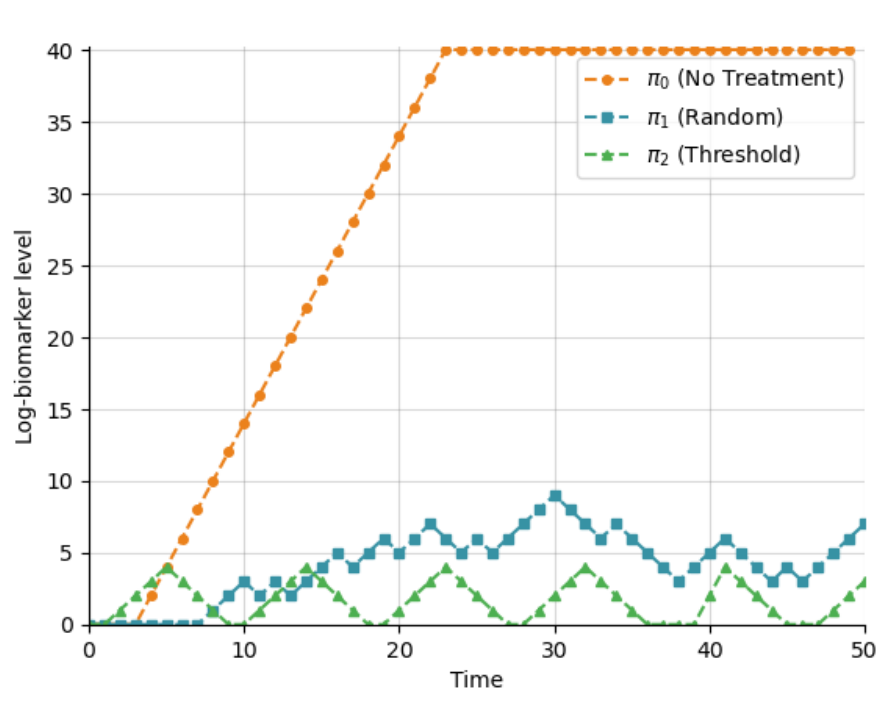}
        \caption{Sample trajectories generated under the different reference policies $\pi_0$, $\pi_1$, and $\pi_2$.}
        \label{fig:trajectories}
    \end{subfigure}
    \quad 
    \begin{subfigure}{0.45\textwidth}
        \centering
        \includegraphics[width=\textwidth]{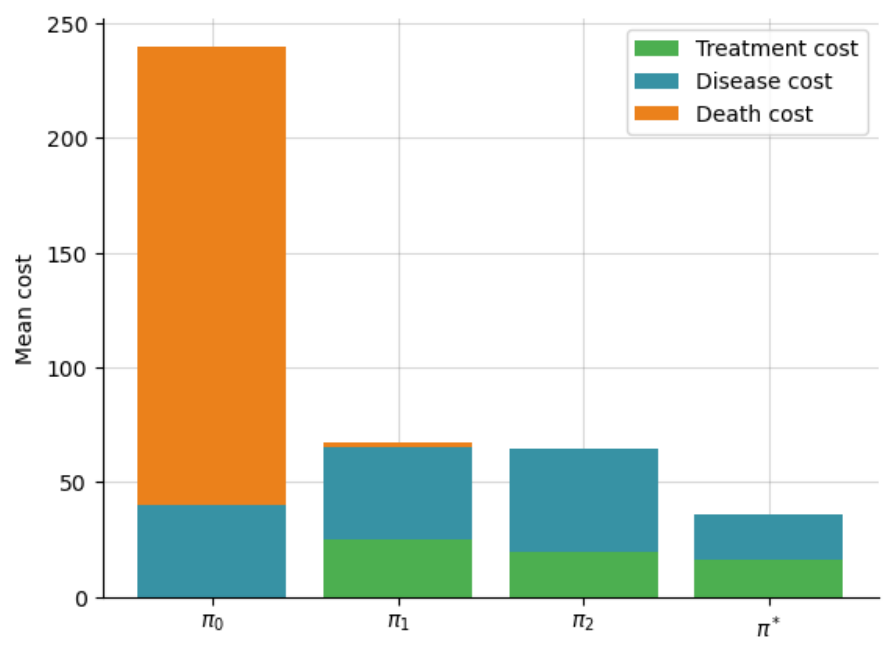}
        \caption{Total expected cost of each policy, estimated via $10^5$ Monte Carlo simulation.}
        \label{fig:costs}
    \end{subfigure}

    \caption{Comparison of policies in terms of trajectories and performance for the running medical example as a MDP.}
    \label{fig:policy_comparison}
\end{figure}
\color{black}
%
\subsubsection{Generalizations}
\label{ss:extension_mdp}
%
For simplicity, we defined MDPs with finite or countable state and action spaces and assumed that all actions are feasible in all states. These assumptions can be relaxed to model more complex scenarios. \\

\textbf{Constraints.}
%
In many applications, specific actions may be prohibited in certain states. For example, in the medical scenario of Section~\ref{ss:mini_mdp}, administering a treatment  ($a=1$) to a deceased patient ($m=3$) should not be allowed. There are two main approaches to handle such constraints. 
The first one is to assign a prohibitive cost ($+\infty$ for instance) to invalid state-action pairs, preventing them from being chosen by optimization algorithms. 
The second one is to add explicit constraints through the addition of an element $K$ to the MDP definition: $\langle S, A, H, P, c\rangle$ becomes $\langle S, A, H, K, P, c\rangle$, where $K \subset S\times A$ is called the \emph{constraints set} and contains all the feasible state-action pairs. Its sections are denoted by $\{K(s), s\in S\}$, where $K(s)\subset A$ is the collection of all allowed actions in state $s$. This approach is explicit and intuitive for finite state and action spaces, requiring only that  $K(s)$ be non empty for all $s\in S$. In the case of countable action spaces, it is usually required that $K(s)$ be non empty and finite for all $s$, see \emph{e.g.} \cite[Section 1.2]{hernandez-lerma_discrete-time_1996} or \cite[Section 2.2]{hinderer_dynamic_2016}. Decision rules $\pi_t$ are then restricted to select actions only from $K(s)$ for a given state or history. Such decision rules are called \emph{admissible}, and \emph{admissible policies} are collections of admissible decision rules. \\

\textbf{General state or action spaces.}
MDPs can also be extended to general (uncountable) state and action spaces. While the core definitions remain similar (\emph{e.g.}, transition probabilities replaced by transition kernels), and in most frameworks some version of Theorem \ref{th:mdp-bellman} still holds as well as existence of optimal policies in suitable classes of policies, these extensions introduce significant mathematical challenges.

In order to properly define controlled trajectories, $S_t$ and $A_t$ from Algorithm \ref{algo:controledMDP} must be random variables, \emph{i.e.} measurable with respect to some suitably chosen $\sigma$-algebras. This first means that the state $S$ and action $A$ spaces must be endowed with $\sigma$-algebras and thus have some minimal topological structure (typically Polish spaces). This also means that \emph{admissible} decision rules must also be measurable, and that the constraints set $K$ must contain the graph of measurable mapping from $S\times A$ onto $A$. Finally, the cost functions also need to be measurable. There exists a wide literature on easy sufficient conditions to ensure all these properties, for instance if $S$ is a subset of some $\mathbb{R}^d$ endowed with its Borel $\sigma$-algebra, and all $K(s)$ are intervals with continuous bounds in $s$, see \emph{e.g.} \cite[chapter 16]{hinderer_dynamic_2016}. 
%
Unbounded cost functions in general spaces may lead to undefined expectations. This issue is resolved by assuming cost functions and transition kernels are bounded by a non-negative \emph{bounding function} and using weighted norms, see \cite[Chapter 16]{hinderer_dynamic_2016}). 
%
Finally, existence of optimal policies requires some minimal regularity on the cost functions and transition kernel (typically lower or upper semi-continuity) and compactness of the set of policies (hence the use of stochastic decision rules). The interested reader is referred to \cite{hernandez-lerma_discrete-time_1996,bauerle_markov_2011,hinderer_dynamic_2016}
for further details.\\

\textbf{Semi-Markov MDPs.}
As detailed in Definition~\ref{def:mdp}, MDPs rely on the transition probability $P$ which verifies the Markov property.
Similarly to the PDMP framework, it is easy to encompass semi-Markov dynamics (where the transition to a new state may depend on the time spent in the current state) by augmenting the dimension of the state space to include time in the state. 
In the literature,  the term Semi-Markov Decision Process is also  used to describe frameworks where decisions do not happen at regular intervals of time, but instead where selected actions may lead to a random amount of time before another action can be taken (this is also sometimes referred to as continuous-time MDPs). These aspects will not be discussed here, the interested reader might refer to \cite{guo2009continuous} for more details, and to \cite{thesisKohar} for a formulation as an augmented MDP. 
%
\subsection{Partially Observed Markov Decision Processes}\label{sec:pomdp}
Partially-Observed MDPs (POMDPs) are an extension of MDPs that were first introduced by \cite{astrom_optimal_1965} within the context of control theory. In this work, we adopt a more operational perspective of POMDPs, proposed by \cite{kaelbling_planning_1998}, which focuses on practical algorithms for computing optimal strategies in settings with finite state, observation and action spaces. We start with a definition of POMDPs in Section \ref{subsec:pomdp}. 
Then, we give an example of POMDP in our medical example context in Section \ref{subsubsec:mini_pomdp}. Here we focus on finite-horizon POMDPs.
%
\subsubsection{POMDP definition}
\label{subsec:pomdp}
%
\begin{defin} [POMDP]
    A finite-horizon POMDP is a tuple $\langle S, A, H, P, \Omega, O, c,\omega_0 \rangle$.
    In addition to the components of a standard MDP introduced in Definition~\ref{def:mdp},  POMDPs introduce 
    \begin{itemize}
        \item a finite set $\Omega$ representing the possible observations of the system; 
        \item an {\em observation function}
   $ O : S \times A \times \Omega \rightarrow [0,1]$,
where $O(s',a,\omega)$ is the probability of observing $\omega\in\Omega$,  after taking action $a$, resulting in next state $s'$;
    \item and an initial observation $\omega_0$, describing the initial imperfect knowledge we have about the initial state.
\end{itemize}
\end{defin}
%
The construction of a POMDP trajectory is illustrated on Figure \ref{fig:pomdp}.
In a POMDP, the objective is also to compute a strategy which minimizes the expected sum of future costs. However, a key difference is that, at each time step $t$, the true state of the system $s_t$ is not directly observed. 
Thus POMDP histories only include actions and observations:
\begin{align*}
h_t = \langle a_0,\omega_1,\ldots,a_{t-1},\omega_t \rangle.
\end{align*}
%
\begin{figure}[tp]
    \centering
    \scalebox{.9}{
\begin{tikzpicture}[scale=0.8,
    node distance = 8mm and 8mm,
    state/.style={circle, draw, minimum size=1cm},
    decision/.style={draw, diamond, align=center, minimum size=1cm},
    reward/.style={draw, minimum size=1cm},
    every edge/.style = {draw, -{Stealth[scale=1.2]}}, 
    every edge quotes/.append style = {auto, inner sep=2pt, font=\footnotesize}
    ]
    \node (s0)  [state] {$s_0$};
    \node (s1)  [state, right= 1.9cm of s0]   {$s_1$};
    \node (s2)  [state, right= 1.9cm of s1]   {$s_2$};
    \node (s3)  [state, right= 1.9cm of s2]   {$s_3$};
    \node (o0)  [state, above= 1.6cm of s0]   {$\omega_0$};
    \node (o1)  [state, above= 1.6cm of s1]   {$\omega_1$};
    \node (o2)  [state, above= 1.6cm of s2]   {$\omega_2$};
    \node (a0)  [decision, below right =of o0]   {$a_0$};
    \node (a1)  [decision, below right =of o1]   {$a_1$};
    \node (a2)  [decision, below right =of o2]   {$a_2$};
    \node (r0)  [reward, below =1.6cm of s1]   {$c(s_0,a_0,s_1)$};
    \node (r1)  [reward, below =1.6cm of s2]   {$c(s_1,a_1,s_2)$};
    \node (r2)  [reward, below =1.6cm of s3]   {$c(s_2,a_2,s_3)$};
    \node (r3)  [reward, dashed, below right =of s3]   {$C(s_3)$};

    \path   (s0)    edge [] (s1) 
                    edge [] (o0)
                    edge [] (r0)
            (o0)    edge [] (a0) 
                    edge [bend right=5] (a1) 
                    edge [bend left=65] (a2) 
            (o1)    edge [] (a1) 
                    edge [bend right=5] (a2) 
            (o2)    edge [] (a2) 
            (s1)    edge [] (s2)
                    edge [] (r0)
                    edge [] (o1)
                    edge [] (r1)
            (s2)    edge [] (s3) 
                    edge [] (r1)
                    edge [] (o2)
                    edge [] (r2)

            (a0)    edge [] (r0) 
                    edge [] (s1)
                    edge [] (o1)
            (a1)    edge [] (r1) 
                    edge [] (s2)
                    edge [] (o2)
            (a2)    edge [] (r2) 
                    edge [] (s3)

            (s3)    edge [] (r2)
                    edge [dashed] (r3)
                    ;
\end{tikzpicture}
    }
    \caption{{\bf Graphical representation of a Partially Observable Markov Decision Process (POMDP) with finite horizon $H=3$.} Circles denote hidden states and available observations, diamonds represent actions and rectangles indicate costs associated with state-action-state transitions. The dashed rectangle depicts the terminal cost.}
    \label{fig:pomdp}
\end{figure}
%
Similarly to solving an MDP, solving a POMDP consists of computing a {\em policy} assigning an action to every possible state of current knowledge of the POMDP.
%
\begin{defin}[POMDP history-dependent decision rule]
    Consider a POMDP $\langle S,A,H,P,\Omega,O,c,\omega_0 \rangle$. 
    A POMDP history-dependent decision rule  $\pi_t:\mathcal{H}_t\to \Delta(A)$ assigns a probability distribution over actions in $A$ to every POMDP history $h_t\in\mathcal{H}_t$. 
\end{defin}
%
\begin{defin}[POMDP policies]
    Consider a POMDP $\langle S,A,H,P,\Omega,O,c,\omega_0 \rangle$. 
    A POMDP history-dependent policy is a list of history-dependent decision rules $\pi=\left(\pi_t\right)_{0\leq t < H}$. Denote by $\Pi^o$ the set of all history-dependent policies.
\end{defin}
%
Note that in practice, many algorithms seek to obtain  \emph{deterministic} history-dependent decision policies, \textit{i.e.} lists of deterministic history-dependent decision rules $\pi_t:\mathcal{H}_t\to A$ which assign an action $a\in A$ to every POMDP history $h_t\in\mathcal{H}_t$. Indeed, it is guaranteed that there exists a deterministic history-dependent policy which optimizes the sum of future expected costs/rewards. 

A POMDP $\langle S,A,H,P,\Omega,O,c,\omega_0 \rangle$ can also be seen as an MDP $\langle \overline{S},A,H,\overline{P},c\rangle$ by setting
\begin{align*}
    \overline{S} &= S\times\Omega,\\
    \overline{P}(\overline{s}'=(s',\omega')|\overline{s}=(s,\omega),a) &= P(s'|s,a)O(s',a,\omega'),
\end{align*}
for all $\overline{s}'=(s',\omega')$ and $\overline{s}=(s,\omega)$ in $\overline{S}$. In this setting the term \emph{partially observed} is more obvious. The main difference between this MDP and the POMDP setting is that in the POMDP, only policies in $\Pi^o$ are considered, and the initial value $S_0$ is randomly drawn from an initial \emph{belief} $b_0$ which is a probability distribution over $S$: $b_0(s)=\mathbb{P}(S_0=s|\omega_0$) for all $s\in S$.
One can see $\Pi^o$ as a subset of $\Pi$ for the MDP $\langle \overline{S},A,H,\overline{P},c\rangle$, and directly use definitions \ref{disexpcost} to \ref{totexpcost} to define the cost for a POMPD policy. For any criterion $V\in\{V_\gamma, V_{\text{av}}, V_H\}$, one has
\begin{align*}
    V(\pi,\omega_0)=\sum_{s\in S}b_0(s)V(\pi,s).
\end{align*}
%
\begin{defin}[Optimal policy]
    Consider the POMDP $\langle S, A, H, P, \Omega,O,c,\omega_0 \rangle$. An optimal policy $\pi^*$ of the POMDP, with respect to criterion $V\in\{V_\gamma, V_{\text{av}}, V_H\}$ given in definitions~\ref{disexpcost} to~\ref{totexpcost} is a history-dependent policy $\pi^* \in \Pi^o$ such that:
    $$
    V(\pi^*,s)\leq V(\pi,s),\ \forall \pi\in\Pi^o, \forall s\in S.
    $$
\end{defin}
%
Note that here, the expected cost $V\in\{V_\gamma, V_{\text{av}}, V_H\}$ depends on the unobserved states of the process, but must be optimized with respect only to the history.
%
%

A way to deal with partial observability in POMDP is to consider a corresponding {\em belief MDP}, which (continuous, multi-dimensional) state space is made of the set of possible belief states $b_t(\cdot)$ of the MDP where $b_t(s)=\mathbb{P}(S_t=s|h_t)$. This approach will not be used in the sequel and is not detailed here. Interested readers are referred to \cite{astrom_optimal_1965,kaelbling_planning_1998} for its construction and the proofs that the belief MDP policy value and POMDP policy value are equal, and that optimal policies are history-dependent, deterministic in general for the POMDP version and Markovian, deterministic for the belief MDP.  
\color{black}
%
\subsubsection{Running medical example as a POMDP}
\label{subsubsec:mini_pomdp}
%
Building on the MDP described in~\ref{ss:mini_mdp}, we now consider that the patient's overall condition is hidden, except when the patient has died ($m=3$), and that the biomarker (still in log scale) is observed either through discrete or continuous noise.\\ 

\textbf{Discrete Observation Space.} 
Consider that the biomarker is observed through discrete noise 
\begin{align*}
    y = \xi + \epsilon \text{ with } \epsilon \sim \mathcal{U}\{-2,-1,0,1,2\},
\end{align*} 
where $y$ denotes the noisy biomarker.
Let $z = \1_{(m=3)}$ be the death indicator. 
At a given time $t$, the observation of a patient's condition is $\omega = (y,z)$, and we assume that the process starts in state $s_0=(0,0)$. 
The POMDP associated with this example builds on the MDP from Section~\ref{ss:mini_mdp} and is defined by adding
\begin{align*}
    \Omega &= \llbracket -2, 42 \rrbracket \times \{0,1\},\quad \omega_0 = (0,0), \\
    O(s',a,\omega=(y,z)) &= \left\{ \begin{aligned}\frac{1}{5} \1_{y\in \llbracket \xi'-2; \xi' + 2 \rrbracket; z=0}  &\text{ if } m'\neq 3,\\ 
     \frac{1}{5} \1_{y\in \llbracket 38; 42 \rrbracket; z=1} &\text{ if } m'=3.
    \end{aligned} \right. 
\end{align*}
%
\textbf{Continuous Observation Space.} 
In practice, partial observation typically involves continuous noise, and a biomarker observation is more realistically represented as
\begin{align*}
    y = \xi + \epsilon \text{ with } \epsilon \sim \mathcal{U}[-2,2]
\end{align*} 
where $y$ still denotes the noisy biomarker. 
With a continuous observation space, the POMDP associated with this example builds on the MDP from Section~\ref{ss:mini_mdp} and is defined by adding
\begin{align*}
    \Omega &= [-2, 42] \times \{0,1\},\quad \omega_0 = (0,0), \\
    o(s',a,\omega=(y,z)) &= \left\{ \begin{aligned}\frac{1}{4} \1_{y\in [\xi'-2; \xi' + 2]; z=0}  &\text{ if } m'\neq 3,\\ 
     \frac{1}{4} \1_{y\in [38; 42]; z=1} &\text{ if } m'=3.
    \end{aligned} \right. 
\end{align*}
where $o$ is the density of the observation function $O$ with respect to the Lebesgue measure on $\mathbb{R}$.


Consider the same reference policies $\pi_0$, $\pi_1$ and $\pi_2$, for controlling this POMDP as in section~\ref{ss:mini_mdp}, except that the threshold policy $\pi_2$ is now applied to the noisy observed level instead of the real hidden one.
Their behaviors are illustrated in Figure~\ref{fig:traj_po} through sample trajectories. The optimal policy $\pi^*$ is approximated via a Monte-Carlo planning algorithm (see Section~\ref{sub:MCP} and the accompanying Jupyter notebook for details). Table~\ref{tab:policy_performance} reports the performance of the different policies and compares their performance in the fully observed and partially observed environments. Policies $\pi_0$ and $\pi_1$ that do not depend on the process value have unchanged performances. Policies $\pi_2$ and the optimal policy reach a higher cost in the more hidden environment, emphasizing that the optimization problem under partial observation is more demanding than the fully observed one. 
\begin{figure}[tp]
    \centering
    \includegraphics[width=.5\textwidth]{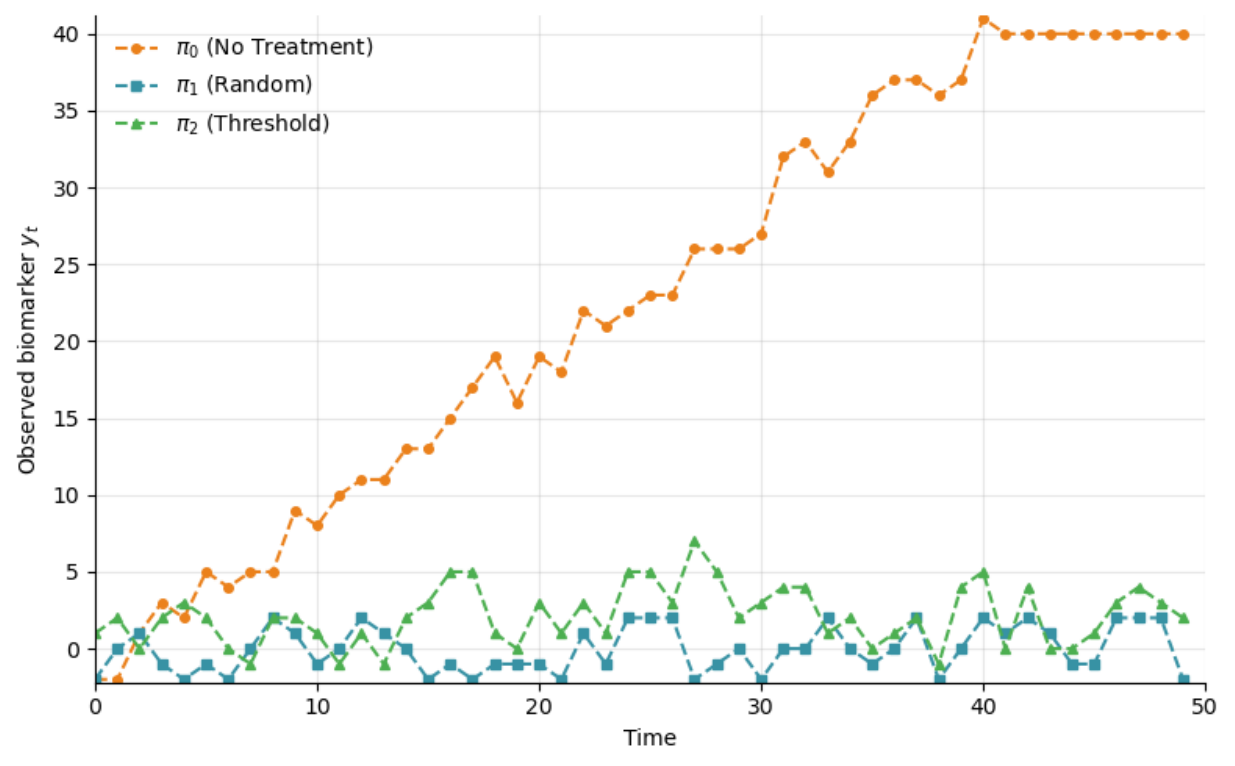}
    \caption{Sample trajectories under policies $\pi_0$, $\pi_1$, and $\pi_2$ for the running medical example as a POMDP.}
    \label{fig:traj_po}
\end{figure}
%
\begin{table}[tp]
    \centering
    \caption{Monte Carlo estimation of policy performance over $10^5$ simulations for the running medical example as a POMDP.}
    \label{tab:policy_performance}
    \begin{tabular}{lcc}
    \hline
    \textbf{Policy} & \textbf{MDP} & \textbf{POMDP} \\
    \hline
    $\pi_0$ & 239.85 $\pm$ 0.04 & 239.84 $\pm$ 0.04\\
    $\pi_1$ & 67.24 $\pm$ 0.21 & 67.15 $\pm$ 0.20\\
    $\pi_2$ & 64.86 $\pm$ 0.05 & 70.11 $\pm$ 0.08\\
    $\pi^*$ & 35.81 $\pm$ 0.03 & 42.06 $\pm$ 0.31\\
    \hline
    \end{tabular}
\end{table}

\color{black}
%
\subsection{Bayes-Adaptive MDP} 
\label{sec:bamdp}
%
In this Section, we describe a Bayesian modeling approach of MDPs and POMDPs as an extended model to handle transition uncertainty in MDPs. We start with motivating this framework, then present the MPD and POMDP versions illustrated by the running medical example.

%
\subsubsection{Motivation and literature}
\label{sub:motiv}
%
There exists a large body of work devoted to studying (PO)MDP where the transition functions and/or reward functions are ill-known. There are basically two families of approaches. First, Reinforcement-Learning algorithms (RL), which will be described in Section \ref{RL}, use interactions with the environment or with a simulator thereof, in order to learn optimal policies of the ill-known (PO)MDP. Another family of approaches, which we describe here, consider that a prior (e.g. multinomial distributions) is used to model a priori knowledge over the ill-known (PO)MDP model (e.g. transitions). Then, an augmented (PO)MDP model is constructed by incorporating additional variables in the state space, in order to model the fact that future observed transitions will inform us about the true (PO)MDP model and thus allow to adapt the optimal control policy.

In the second research avenue, Duff \textit{et al.} \cite{duff_optimal_2002} initially proposed a Bayesian approach to represent and solve finite horizon MDPs with an unknown transition model, where transition probabilities of the MDP are modeled via a priori Dirichlet distributions. An optimal policy to solve the MDP while taking into account the uncertainty about the model is proposed. In essence, Duff shows that this optimal planning problem can be modeled as an MDP over a finite state space, where the state includes {\em hyperparameters} and is exponentially larger than that of the initial MDP. The implementation aspects of this framework are studied in \cite{castro_using_2007} which provide empirically efficient algorithms to solve Bayes-Adaptive Markov Decision Processes (BAMDPs). An excellent review of algorithmic challenges and state of the art of the domain is given in \cite{arumugam_planning_2022}. Here we provide an overview of the BAMDP framework. We start with the definition of fully observed models in Section \ref{sub:bamdp}, and illustrate it with our running medical example in Section \ref{ss:mini_bamdp}. Then we turn to the partially observed case in Section \ref{sub:bapomdp} and its medical example in Section \ref{subsubsec:mini_bapomdp}.
\color{black}
%
\subsubsection{Bayes-Adaptive MDP}
\label{sub:bamdp}
%
In this Section we follow the didactic exposition of BAMDP and most of the notations of \cite{ross_bayesian_2011}. 
An archetypal BAMDP problem is when we face a (finite-horizon) MDP $\langle S, A, H, P, c \rangle$, where $S$ and $A$ are finite, the model is initially known, {\em except} for the transition probabilities $P(\cdot|s, a), \forall (s, a) \in S\times A$. 
BAMDPs are embedded in a Bayesian framework, where we assume that we have a {\em prior} distribution over $P(\cdot|s, a)$. As part of the acquisition of the new state $s'$, the decision-maker updates the {\em posterior} distribution with the information provided by the transition $(s, a, s')$ when it is encountered.  
Let $S=\{s^1,\ldots,s^n\}$ be the finite state-space and let us define 
$$p^{s,a}=(p^{s,a}_1,\ldots,p^{s,a}_n),$$ 
where $p^{s,a}_i = P(s^i|s, a)$. The vector $p^{s,a}$ can be seen as the parameter vector of a {\em multinomial distribution}.
In the Bayesian context, the {\em conjugate} of a multinomial distribution is a {\em Dirichlet} distribution, $Dirichlet(\theta)$, parameterized by a count vector $\theta = (\theta_1,\ldots,\theta_n)\in\mathbb{N}^n$. 
The multinomial/Dirichlet framework makes it easy to perform Bayesian updates of the knowledge over $p$ when transitions are experienced. 
Indeed, assume that from a state $s$ and action $a$, we observe $K$ transitions from the true distribution $p^{s,a}$, and let $k=(k_1,\ldots,k_n)$ represent the number of transitions to $s^1, \ldots, s^n$ respectively (with $K = \sum_i^nk_i$).
Then, a Bayesian update of a $Dirichlet(\theta_1,\ldots,\theta_n)$ prior on $p$ leads to the posterior distribution over $p'$, equal to
\begin{equation*}
p' = (p'_1,\ldots,p'_n) \sim Dirichlet(\theta_1+k_1,\ldots,\theta_n+k_n).
\end{equation*}
This form of update is particularly simple. Thus, most works, starting with \cite{duff_optimal_2002}, make the assumption of independent Dirichlet priors over independent multinomial distributions $\{p^{s,a}\}_{(s,a)\in S\times A}$. 

The fundamental principle of Bayes-Adaptive MDP is to incorporate model uncertainty into planning, by considering an {\em augmented MDP}, which states $\tilde{s}=(s,\theta)$ include the hyperparameters $\theta$ which can be considered as a memory of past experiences. As in classical MDP, solving a BAMDP requires identifying the optimal decision for all possible states and hyperparameter values. 
Formally, a BAMDP is defined as follows.
%
\begin{defin}[BAMDP]
    A Bayes-Adaptive MDP $\langle\tilde{S}, \tilde{A}, H, \tilde{P}, \tilde{c}\rangle$ is an MDP that can be defined from a MDP $\langle S, A, H, P, c\rangle$ with  ill-known transition matrix $P$ as
    \begin{itemize}
        \item $\tilde{S} = S \times \Theta$, where $\Theta = \{(\theta_{s'}^{s,a})_{s, s'\in S, a\in A}\} \subset \mathbb{N}^{|S|^2\times|A|}$ is the set of {\em hyperstates} $\tilde{s}=(s,\theta)$. The initial belief over $P$ is modeled by specifying an initial set of counts $\{\hat{\theta}_{s'}^{s,a}\}$ assumed to verify $\sum_{s'\in S} \hat{\theta}_{s'}^{s,a}>0, \forall s,a$, in order for the initial Dirichlet distributions to be well-defined.  
        \item $\tilde{A} = A$. The action space is unchanged.
        \item $\tilde{P}(s', \theta'|s, a, \theta)$ is the probability of transitioning from hyperstate $(s,\theta)$ to hyperstate $(s', \theta')$ under action $a$. This writes, by the chain rule:
        $$
        \tilde{P}(s', \theta'|s, a, \theta) = \mathbb{P}(S_{n+1}=s'|S_n=s,A_n=a,\theta_n=\theta)\mathbb{P}(\theta_{n+1}=\theta'|S_n=s,A_n=a,S_{n+1}=s'\theta_n=\theta),
        $$
        where $\mathbb{P}(S_{n+1}=s'|S_n=s,A_n=a,\theta_n=\theta)=P_\theta(s'|s,a)$ and $P_\theta$ is the transition matrix of the original MDP with parameters $p^{s,a}$ drawn from $Dirichlet(\theta)$, and $\mathbb{P}(\theta_{n+1}=\theta'|S_n=s,A_n=a,S_{n+1}=s'\theta_n=\theta)$ is the (deterministic) update rule described above for the hyperparameter:
        \begin{equation}\label{eq:update-hyper-BAMDP}
            \mathbb{P}(\theta_{n+1}=\theta'|S_n=s,A_n=a,S_{n+1}=s'\theta_n=\theta) = 1 \mbox{ iff } \theta'^{s,a}_{s'} = \theta_{s'}^{s,a}+1 \mbox{ and } \theta'^{\hat{s},\hat{a}}_{\hat{s}'} = \theta_{\hat{s}'}^{\hat{s},\hat{a}}, \forall (\hat{s}, \hat{a}, \hat{s}')\neq (s, a, s').
        \end{equation}
        \item $\tilde{c}(s, a, s', \theta) = c(s, a, s')$, since we assume that the cost function is known.
    \end{itemize}
\end{defin}
%
The BAMDP $\langle\tilde{S}, \tilde{A}, H, \tilde{P}, \tilde{c}\rangle$ is a classical MDP with finite state space if $H$ is finite.
In this case, the size of $\Theta$ is bounded by $H^{|S|^2|A|}$ and the size of $\tilde{S}$ by $|S|\times H^{|S|^2|A|}$. Note that, in a BAMDP, policies $\pi: \tilde{S}\to A$ choose actions as functions of the current state, $s$ {\em and} of the current knowledge about the unknown MDP transition function, summarized by the parameter vector $\theta$.  
Let us make clear that the Bayes-Adaptive framework handles parameters' uncertainty in two phases.
\begin{itemize}
    \item The first phase is completely off-line and does not require to observe any transition of the MDP. In this phase, we use the prior knowledge about transitions and the Multinomial/Dirichlet conjugacy in order to build a (BA)-MDP with augmented state space $\tilde{S}=S\times\Theta$, as shown in Figure~\ref{fig:bamdp_state_space}. 
    This MDP can be solved using any MDP solution algorithm (see Section \ref{sec:algo}).
    \item In the second phase, the decision-maker applies the optimal  BAMDP policy $\tilde{\pi}: S\times\Theta \to A$ provided by the first phase to the \emph{real-life} problem, in which only transitions $(s,a,s')$ drawn form the original unknown matrices $P$ are observed as shown in Figure~\ref{fig:bamdp}. Retrieving the hyperparameter values needed to select optimal actions is straightforward since $\theta$ and $\theta'$ are transition counts and the transition $(s,a,s',\theta)\rightarrow \theta'$ is actually deterministic, as seen in Equation (\ref{eq:update-hyper-BAMDP}).
\end{itemize}
%
\begin{figure}[tp]
    \centering
    \resizebox{0.9\textwidth}{!}{
    \begin{tikzpicture}[
    state/.style={align=center, minimum width=3.5cm},
    decision/.style={draw, diamond, align=center, minimum size=1cm},
    arrow/.style={-stealth, thick},
    dashedarrow/.style={dashed, -stealth, thick}
]

\node[state] (s0) {$\Tilde{S}_{t} = \left(1,  \
\begin{array}{c}
\begin{bmatrix} 
\theta^{1\alpha}_1 & \theta^{1\alpha}_2 \\ 
\theta^{2\alpha}_1 & \theta^{2\alpha}_2
\end{bmatrix},
\begin{bmatrix}
\theta^{1\beta}_1 & \theta^{1\beta}_2 \\ 
\theta^{2\beta}_1 & \theta^{2\beta}_2 
\end{bmatrix}
\end{array}
\right)$};

\node[decision] at ([xshift=-3cm, yshift=-3cm]s0) (a1) {$a = \alpha$};
\node[decision] at ([xshift=3cm, yshift=-3cm]s0) (a2) {$a = \beta$};

\draw[arrow] (s0) -- (a1);
\draw[arrow] (s0) -- (a2);

\node[state] at ([xshift=-3.5cm, yshift=-3cm]a1) (s1a1) {$\Tilde{S}_{t+1} = \left(1,  \
\begin{array}{c}
\begin{bmatrix} 
\theta^{1\alpha}_1+1 & \theta^{1\alpha}_2 \\ 
\theta^{2\alpha}_1 & \theta^{2\alpha}_2
\end{bmatrix}, \\[0.3em]
\begin{bmatrix}
\theta^{1\beta}_1 & \theta^{1\beta}_2 \\ 
\theta^{2\beta}_1 & \theta^{2\beta}_2 
\end{bmatrix}
\end{array}
\right)$};

\node[state] at ([xshift=0.9cm, yshift=-3cm]a1) (s2a1) {$\left(2,  \
\begin{array}{c}
\begin{bmatrix} 
\theta^{1\alpha}_1 & \theta^{1\alpha}_2+1 \\ 
\theta^{2\alpha}_1 & \theta^{2\alpha}_2
\end{bmatrix}, \\[0.3em]
\begin{bmatrix}
\theta^{1\beta}_1 & \theta^{1\beta}_2 \\ 
\theta^{2\beta}_1 & \theta^{2\beta}_2 
\end{bmatrix}
\end{array}
\right)$};

\node[state] at ([xshift=-0.9cm, yshift=-3cm]a2) (s1a2) {$\left(1,  \
\begin{array}{c}
\begin{bmatrix} 
\theta^{1\alpha}_1 & \theta^{1\alpha}_2 \\ 
\theta^{2\alpha}_1 & \theta^{2\alpha}_2
\end{bmatrix}, \\[0.3em]
\begin{bmatrix}
\theta^{1\beta}_1+1 & \theta^{1\beta}_2 \\ 
\theta^{2\beta}_1 & \theta^{2\beta}_2 
\end{bmatrix}
\end{array}
\right)$};


\node[state] at ([xshift=3.5cm, yshift=-3cm]a2) (s2a2) {
$\left(2,\ 
\begin{array}{c}
\begin{bmatrix} 
\theta^{1\alpha}_1 & \theta^{1\alpha}_2 \\ 
\theta^{2\alpha}_1 & \theta^{2\alpha}_2
\end{bmatrix}, \\[0.3em]
\begin{bmatrix}
\theta^{1\beta}_1 & \theta^{1\beta}_2+1 \\ 
\theta^{2\beta}_1 & \theta^{2\beta}_2 
\end{bmatrix}
\end{array}
\right)$
};

\draw[dashedarrow] (a1) -- node[midway, left,xshift=-2mm] {$\frac{\theta^{1\alpha}_1}{\theta^{1\alpha}_1+\theta^{1\alpha}_2}$}  (s1a1);
\draw[dashedarrow] (a1) -- node[midway, right] {$\frac{\theta^{1\alpha}_2}{\theta^{1\alpha}_1+\theta^{1\alpha}_2}$} (s2a1);
\draw[dashedarrow] (a2) -- node[midway, left] {$\frac{\theta^{1\beta}_1}{\theta^{1\beta}_1+\theta^{1\beta}_2}$} (s1a2);
\draw[dashedarrow] (a2) -- node[midway, right, xshift=2mm] {$\frac{\theta^{1\beta}_2}{\theta^{1\beta}_1+\theta^{1\beta}_2}$} (s2a2);

\end{tikzpicture}
    }
    \caption{{\bf Evolution of the state space of a BAMDP over a single time step.} This figure illustrates a specific example with two possible states $S=\{1,2\}$ and two possible actions $A=\{\alpha,\beta\}$.  Diamonds represent actions. The parameters $\theta_{s'}^{s,a}$ correspond to the transition counts for each $(s, a, s')$ triplet. The dashed arrows indicate transition probability to hyperstates, defined according to the count matrix and the Multinomial/Dirichlet conjugacy. Here, only the case of $s_t = 1$ is shown, the case of $s_t = 2$ is similar.}
    \label{fig:bamdp_state_space}
\end{figure}
%
\begin{figure}[tp]
    \centering
    \scalebox{.9}{
\begin{tikzpicture}[scale=0.8,
    node distance = 10mm and 15mm,
    state/.style={circle, draw, minimum size=1cm},
    decision/.style={draw, diamond, align=center, minimum size=1cm},
    reward/.style={draw, minimum size=1cm},
    every edge/.style = {draw, -{Stealth[scale=1.2]}}, 
    every edge quotes/.append style = {auto, inner sep=2pt, font=\footnotesize}
    ]
    \node (s0)  [state] {$s_0$};
    \node (s1)  [state, right=3cm of s0]   {$s_1$};
    \node (s2)  [state, right=3cm of s1]   {$s_2$};
    \node (s3)  [state, right=3cm of s2]   {$s_3$};
    \node (theta0)  [state, above=of s0]   {$\theta_0$};
    \node (theta1)  [state, above=of s1]   {$\theta_1$};
    \node (theta2)  [state, above=of s2]   {$\theta_2$};
    \node (theta3)  [state, above=of s3]   {$\theta_3$};
    \node (a0)  [decision, below right =of s0]   {$a_0$};
    \node (a1)  [decision, below right =of s1]   {$a_1$};
    \node (a2)  [decision, below right =of s2]   {$a_2$};
    \node (r0)  [reward, below =3cm of s1]   {$c_0(s_0,a_0,s_1)$};
    \node (r1)  [reward, below =3cm of s2]   {$c_1(s_1,a_1,s_2)$};
    \node (r2)  [reward, below =3cm of s3]   {$c_2(s_2,a_2,s_3)$};
    \node (r3)  [reward, dashed, below right =of s3]   {$C(s_3)$};

    \path   (s0)    edge [] (s1) 
                    edge [] (a0)
                    edge [bend left=-25] (r0)
                    edge [] (theta1)
            (s1)    edge [] (s2)
                    edge [] (a1)
                    edge [] (r0)
                    edge [] (theta1)
                    edge [] (theta2)
                    
                    edge [bend left=-25] (r1)
            (s2)    edge [] (s3) 
                    edge [] (a2)
                    edge [] (r1)
                    edge [] (theta2)
                    edge [] (theta3)
                    edge [bend left=-25] (r2)
            (theta0)    edge [] (theta1)
                        edge [] (s1)
                        edge [] (a0)
            (theta1)    edge [] (theta2)
                        edge [] (s2)
                        edge [] (a1)
            (theta2)    edge [] (theta3)
                        edge [] (s3)
                        edge [] (a2)
            (a0)    edge [] (s1)
                    edge [] (r0)
                    edge [] (theta1)
            (a1)    edge [] (s2)
                    edge [] (r1) 
                    edge [] (theta2)
            (a2)    edge [] (r2) 
                    edge [] (s3)
                    edge [] (theta3)
            (s3)    edge [] (r2)
                    edge [] (theta3)
                    edge [dashed] (r3)
                    ;

    \draw[dashed, rounded corners] 
        ([shift={(-0.5,0.5)}]theta0.north west) -- 
        ([shift={(0.5,0.5)}]theta0.north east) --
        ([shift={(0.5,-0.5)}]s0.south east) --
        ([shift={(-0.5,-0.5)}]s0.south west) -- cycle;
    \node[above=0.1cm] at ([shift={(0,0.5)}]theta0.north) {$\Tilde{s}_0$};

    \draw[dashed, rounded corners] 
        ([shift={(-0.5,0.5)}]theta1.north west) -- 
        ([shift={(0.5,0.5)}]theta1.north east) --
        ([shift={(0.5,-0.5)}]s1.south east) --
        ([shift={(-0.5,-0.5)}]s1.south west) -- cycle;
    \node[above=0.1cm] at ([shift={(0,0.5)}]theta1.north) {$\Tilde{s}_1$};
    
    \draw[dashed, rounded corners] 
        ([shift={(-0.5,0.5)}]theta2.north west) -- 
        ([shift={(0.5,0.5)}]theta2.north east) --
        ([shift={(0.5,-0.5)}]s2.south east) --
        ([shift={(-0.5,-0.5)}]s2.south west) -- cycle;
    \node[above=0.1cm] at ([shift={(0,0.5)}]theta2.north) {$\Tilde{s}_2$};
    
    \draw[dashed, rounded corners] 
        ([shift={(-0.5,0.5)}]theta3.north west) -- 
        ([shift={(0.5,0.5)}]theta3.north east) --
        ([shift={(0.5,-0.5)}]s3.south east) --
        ([shift={(-0.5,-0.5)}]s3.south west) -- cycle;
    \node[above=0.1cm] at ([shift={(0,0.5)}]theta3.north) {$\Tilde{s}_3$};

\end{tikzpicture}
    }
    \caption{{\bf Graphical representation of BAMDP with finite horizon $H=3$.} Circles denote states and hyperparameters, dashed rounded rectangles denote hyperstates, diamonds represent actions and rectangles indicate costs associated with state-action-state transitions. The dashed rectangle depicts the terminal cost.}
    \label{fig:bamdp}
\end{figure}

BAMDPs can be extended to the case where the cost function is also unknown. Interested readers may refer to \cite{duff_optimal_2002} for an introduction. In \cite{arumugam_planning_2022}, the authors provide an overview of the literature on algorithmic approaches to solve BAMDP, as well as a complexity analysis of the full problem. 
Among the algorithms,Reinforcement Learning approaches to solve BAMDP, based on Monte-Carlo planning are possible options.
%
\subsubsection{Running medical example as a BAMDP}
\label{ss:mini_bamdp}
%
Going back to the MDP medical example presented in Section~\ref{ss:mini_mdp}, we now suppose that the parameters $p_{0}^0$, $p_{1}^0$ and $p_{2}^0$ are unknown, \textit{i.e.} that the probability transitions from a health condition state $(0,0)$ to any of the states $(0,0)$, $(1,0)$ or $(2,0)$ in the absence of treatment are unknown. Recall that all other transitions are deterministic. 
To model the uncertainty on these transitions corresponding to $P(\cdot|s=(0,0),a=0)$, we use the standard Multinomial/Dirichlet approach 
$$(p_{0}^0,p_{1}^0,p_{2}^0) \sim Dirichlet(\theta_{0}^0,\theta_{1}^0,\theta_{2}^0).$$

With each new observation, we update the a posteriori distribution of the parameters of the transition matrix $\theta' = \theta + \Delta$, where $\Delta=e_{m'}$ is the vector in $\mathbb{N}^3$ that has a value of $1$ at the position corresponding to the observed new mode $m'$, while other entries are $0$. 
Let us illustrate the update process with an example. 
Suppose that at time $t$, the patient is in state $s = (0, 0)$ and no treatment ($a=0$) is applied.
At this stage, the hyperparameter vector $\theta$ is, for instance,
$
\theta =(5,1,0) 
$.
At the next step, the patient transitions into a slow relapse state. 
This observation corresponds to the state-action-state sequence $s = (0, 0), a=0, s'=(1, 0)$.
In this case, the parameter $\theta_1^0$ of the Dirichlet distribution is updated, and $\Delta$ is equal to
$
e_2 = ( 0,1,0)
$.
The hyperparameter vector $\theta$ is updated to 
       $ \theta' = \begin{pmatrix}
    5 & 1 & 0 
    \end{pmatrix} + \begin{pmatrix}
    0 & 1 & 0 
    \end{pmatrix}  = \begin{pmatrix}
    5 & 2 & 0 
    \end{pmatrix}$.
The BAMDP associated with this example builds on the MDP from Section~\ref{ss:mini_mdp}. It is defined by adding
\begin{align*}
    \tilde{S}&=S\times \llbracket 0~;~ H \rrbracket^3, \quad \tilde{A} = A, \quad \tilde{c}  = c, \\
    \tilde{P}(s'=(m',\xi'),\theta'|s,a,\theta)&= 
    \left\{\begin{aligned}
P(s'|s,a)\1_{\theta'=\theta},& \text{ if } (s,a)\neq (0,0,0),\\
\frac{\theta_{m'}^0}{\theta_{0}^0+\theta_{1}^0+\theta_{2}^0}\1_{\theta'=\theta+e_{m'}},& \text{ if } (s,a)= (0,0,0).
    \end{aligned} \right.
\end{align*}

Let us consider the threshold policy $\pi_2$, which administers treatment whenever the biomarker exceeds the value $3$. Figure~\ref{fig:prior_evol} illustrates the evolution of the belief over the probability transitions along sampled trajectories \emph{I.e.} the posterior distribution induced by the current prior. Importantly, this evolution is stochastic: it varies not only over time within a single trajectory, but also across different trajectory realizations under the same policy.
\begin{figure}[tp]
    \centering
    \includegraphics[width=0.7\textwidth]{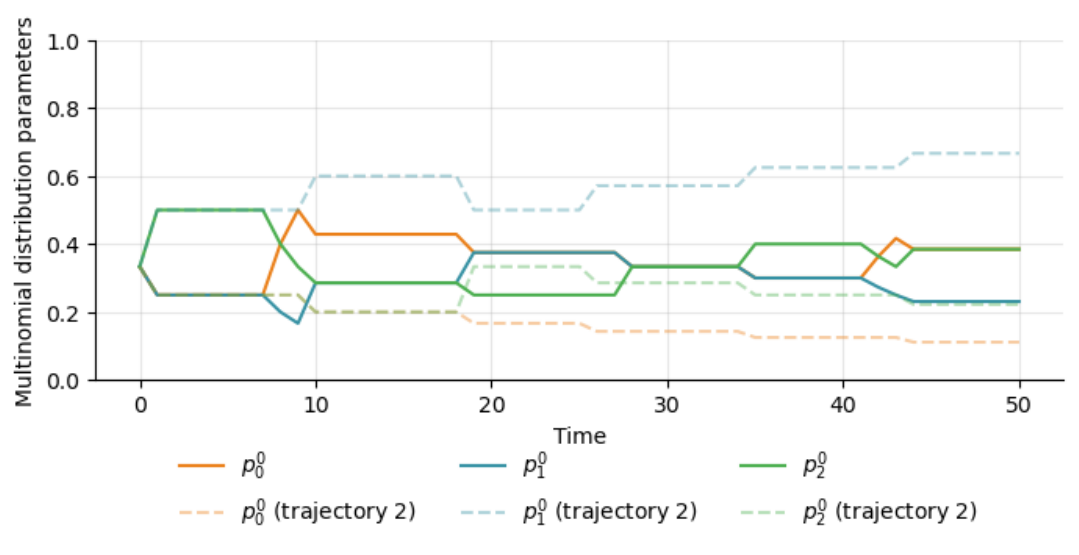}
    \caption{Evolution of the posterior belief over the probability transitions along trajectories generated under policy $\pi_2$. Solid lines correspond to a first sampled trajectory, while dashed lines represent a second independent sampled trajectory under the same policy.}
    \label{fig:prior_evol}
\end{figure}
    Resolution of the BAMDP is provided in the accompanying notebook. Results are illustrated through trajectory simulations.
    The optimal policy $\pi^*$ is approximated via \textit{model-based} Monte Carlo planning methods(see Section~\ref{RL} and the accompanying Jupyter notebook for details). As in the MDP version in Section~\ref{ss:mini_mdp}, it corresponds to a strategy in which treatment is administrated as soon as a relapse is detected. 
\color{black}
%
\subsubsection{Bayes-Adaptive POMDP}
\label{sub:bapomdp}
%
In \cite{ross_bayesian_2011}, the author proposed an extension of BAMDP to the partially-observed case. In this {\em Bayes-Adaptive POMDP} model (BAPOMDP), the state of the process is hidden and the only available information is given through observations $\omega\in \Omega$.
In addition, the observation model is parameterized, similarly to the transition model, and the values of the parameters are updated through interactions with the system.
%
\begin{defin}[BAPOMDP]
    A BAPOMDP is an extension of a BAMDP given by $\langle \bar{S}, A, H, \bar{P}, \Omega, O, c, \bar{b}_0\rangle$, where
\begin{itemize}
    \item $\bar{S} = \tilde{S} \times \Psi = S \times \Theta \times \Psi$, where
    $
    \Psi = \left\lbrace \left(\psi^{a}_{s',\omega}\right)_{a\in A, s'\in S, \omega\in\Omega}\right\rbrace \subset \mathbb{N}^{|S|\times|A|\times|\Omega|}
    $
    is an observation count function. The coordinate $\psi^{a}_{s',\omega}$ counts the number of times $\omega$ has been obtained when action $a$ has led to $s'$. Note that it is not observed since $s'$ is not observed.
    \item $\bar P$ is defined using Bayes rule and independence properties as
    \begin{align*}
        \lefteqn{\bar{P}(s',\theta',\psi',\omega|s,\theta,\psi,a)}\\
        &= \mathbb{P}(S_{t+1}=s'|S_t=s,A_t=a,\theta_t=\theta) \times \mathbb{P}(\theta_{t+1}=\theta'| S_t=s, A_t=a, S_{t+1}=s',\theta_t=\theta)\\
        &\quad \times \mathbb{P}(\psi_{t+1}=\psi'|\psi_t=\psi, A_t=a, S_{t+1}=s', \omega_{t+1}=\omega)\times \mathbb{P}(\omega_{t+1}=\omega|A_t=a, S_{t+1}=s', \psi_{t}=\psi), \\
        &= \tilde{P}(s',\theta^\prime|s,a,\theta) \times \mathbb{P}(\psi_{t+1}=\psi'|\psi_t=\psi, A_t=a, S_{t+1}=s', \omega_{t+1}=\omega)\times \mathbb{P}(\omega_{t+1}=\omega|A_t=a, S_{t+1}=s', \psi_{t}=\psi).
    \end{align*}
    As for the Dirichlet state transition probability, the observation probability is defined as
    \begin{equation*}
        \mathbb{P}(\omega_{t+1}=\omega|A_t=a, S_{t+1}=s', \psi_{t}=\psi) = \frac{\psi^a_{s',\omega}}{\sum_{\omega'\in\Omega}\psi^a_{s',\omega'}},
    \end{equation*}
    and the transition to $\psi'$ is also deterministic
    \begin{equation*}
        \mathbb{P}(\psi_{t+1}=\psi'|\psi_t=\psi, A_t=a, S_{t+1}=s', \omega_{t+1}=\omega) = 1 \mbox{ iff }\psi'^a_{s',\omega} = \psi^a_{s',\omega}+1 \mbox{ and } \psi'^{\hat{a}}_{\hat{s},\hat{\omega}} = \psi^{\hat{a}}_{\hat{s},\hat{\omega}}, \forall (\hat{a}, \hat{s}, \hat{\omega})\neq (a, s', \omega).
    \end{equation*}
    \item Finally, the initial belief state is defined over $\bar{S}$ as : $\bar{b}_0(s, \theta_0, \psi_0) = b_0(s)$ and $\bar{b}_0(s, \theta, \psi) =0$ if $(\theta, \psi)\neq(\theta_0, \psi_0)$. The initial counts $(\theta_0, \psi_0)$ are assumed to be known.
\end{itemize}
\end{defin}
%
As for BAMDPs, the BAPOMDP model  $\langle \bar{S}, A, H, \bar{P}, \Omega, O, c, \bar{b}_0\rangle$ is a new POMDP over an extended finite state space (provided that $H$ is finite). In \cite{ross_bayesian_2011} the authors provide exact (for very small problems) and approximate solution algorithms, for solving BAPOMDP, see Section~\ref{sec:algo}. 
%
\subsubsection{Running medical example as a BAPOMDP}
\label{subsubsec:mini_bapomdp}
%
Going back to the BAMDP example presented in section~\ref{ss:mini_bamdp}, we now consider a noisy measurement $y$ of the biomarker as introduced in section~\ref{subsubsec:mini_pomdp}: $y=\xi+\epsilon$ with a multinomial noise $\epsilon \in \{-2, -1, 0, 1, 2\}$ with unknown parameters $p_i=p_{s'=(m',\zeta'),\omega=y}\1_{\zeta'-y=i} = \mathbb{P}(\epsilon = i)$ independent from the action. We assume that  $p=(p_{-2},p_{-1}, p_0, p_1, p_2)$ follows a Dirichlet distribution. 
At the initial stage, the parameters of the observation model are given by 
$$ p \sim Dirichlet(\psi_{_2},\psi_{-1}, \psi_0, \psi_1, \psi_2),$$ 
with $(\psi_i)_{-2:2} \in \Psi\subset \mathbb{N}^{5}$. 
The update of the parameters for $\psi$ is similar to that of the parameters for $\theta$ presented in section~\ref{ss:mini_bamdp}. With each new observation $(s',a,\omega)$, we update the posterior distribution of the parameters of the transition matrix, so that $\psi' = \psi + \Delta$ is a counting matrix, where $\Delta=e_{i}$ is the vector in $\mathbb{N}^5$ that has a value of $1$ at the position corresponding to $i=\zeta'-y$ and $0$ elsewhere, for $s'=(m',\zeta')$ and $\omega=y$.  
This BAPOMDP is therefore equivalent to a POMDP on a larger state-space. As in the classic POMDP formalism (see Section~\ref{subsec:pomdp}), 
this uncertainty can be represented using regular POMDP formalism presented in Section~\ref{subsec:pomdp}, by maintaining (and updating) a belief over these count vectors. 

This model is now much more complicated and its resolution is well beyond the scope of this work. 
\color{black}
%
\subsection{Algorithms for computing exact or approximate optimal  policies in MDPs} 
\label{sec:algo}
%
The exact computation of optimal policies is often challenging due to combinatorial explosion, the curse of dimensionality and/or the continuous nature of many real-world problems. To address these resolution challenges, various methods have been developed to approximate optimal policies. This section explores these methods, beginning with exact solution techniques and transitioning to advanced approximation methods, such as simulation based methods or Reinforcement Learning (RL). 
 
Exact solution techniques involve targeting the optimal policy with convergence guaranties, often through smart exploration of the whole set of admissible policies. 
 While these methods are theoretically appealing, they quickly become computationally infeasible for large-scale problems due to the curse of dimensionality.
To overcome these limitations, approximation techniques provide a practical alternative. Many of these approaches rely on simulations to strike a balance between precision and computational efficiency.
Monte Carlo Planning (MCP) is a notable example, leveraging random sampling to evaluate and refine policies. Techniques like Monte Carlo Tree Search (MCTS) excel in sequential decision-making problems by simulating state-action trajectories and updating estimates based on averaged outcomes.
Beyond Monte Carlo methods, approaches such as approximate dynamic programming and policy iteration with sampling have been proposed. These methods offer additional practical solutions for complex control problems but fall outside the scope of this paper. The interested reader is referred to \cite{chang2013simulation} for further details on this topic.

Reinforcement Learning provides a powerful framework for approximating value functions and deriving near-optimal policies, especially in cases where exact methods are impractical. RL methods focus on learning optimal policies through iterative interactions with the environment, making them particularly effective in scenarios with partially unknown or highly complex models. 

This section offers a high-level overview of these approaches, covering their theoretical underpinnings and practical applications without delving into algorithmic details. Specifically, Section~\ref{sub:DP} introduces dynamic programming-based algorithms, which provide exact solutions; Section~\ref{sub:MCP} discusses Monte Carlo Planning algorithms, which are designed to handle problems of higher complexity, and finally Section~\ref{RL} presents RL approaches, tailored for frameworks with partially unknown models.
%
\subsubsection{Dynamic programming algorithms}
\label{sub:DP}
%
MDPs with finite state and action spaces and explicit models may be solved exactly via dynamic programming. 
This iterative method is based on Bellman equations, stated in Theorem \ref{th:mdp-bellman}. It calculates the optimal value function and policy by defining and solving simpler embedded subproblems. 
This section extensively uses notations and definitions from Section~\ref{subsec:genericMDP} and describes dynamic programming for finite-horizon and then infinite horizon problems. We start with backward iteration, then introduce value iteration and policy iteration and conclude with a discussion on limitations and extensions of dynamic programming.\\

%
\paragraph{{\bf Backward iteration}}
%
The backwards induction algorithm given in Algorithm~\ref{algo:Backward_Iter} consists of systematically solving decision-making problems over a fixed number of time steps by working backwards from the final stage. 
Unlike the infinite-horizon case, where the goal is to optimize decisions over an indefinite period, in the finite-horizon case, decisions are optimized over a fixed number of stages. 
The process begins with the definition of the terminal value function, where no further decisions can be made. It thus corresponds to the terminal cost function. 
For each preceding time step, the value function is updated using Bellman recursion from Theorem \ref{th:mdp-bellman}, which calculates the expected cumulative cost based on the current state, available actions, immediate costs, and the expected value of future states. 
This backward induction process continues until we reach the initial time step \(t=0\), at which point the value function for each state has been determined for every stage.
Once the value function has been calculated for all stages, a Markovian optimal policy can be derived by selecting, at each time step, the action that minimizes the expected cost for each state.  
This approach ensures that the resulting policy is optimal for the finite-horizon problem, as it takes into account both immediate and future costs at each decision point.
In finite-horizon MDPs, the optimal decision rule generally varies with the time step $t$, making the optimal policy non-stationary. 
%
\begin{algorithm}[tp]
\caption{Backward Iteration Algorithm}
\label{algo:Backward_Iter}
\begin{algorithmic}[1]
    \FOR{$s \in S$}
        \STATE $V_H(s) = C(s)$
    \ENDFOR
    \FOR{$t$ from $H-1$ to $0$}
        \FOR{$s \in S$}
                \STATE $V_t(s) \leftarrow \min_{a \in A} \sum_{s' \in S}P(s' | s,a)(c_t(s,a, s') + V_{t+1}(s'))$
                \STATE $\pi^{\star}_t(s) \leftarrow \argmin_{a \in A} \sum_{s' \in S}P(s' | s,a)(c_t(s,a, s') + V_{t+1}(s'))$
        \ENDFOR
    \ENDFOR
\end{algorithmic}
\end{algorithm}
%
The accompanying jupyter notebook provides an implementation of the backward iteration algorithm for the MDP version of our running medical example given in Section \ref{ss:mini_mdp}. 
\\

%
\paragraph{{\bf Value iteration}}
%
The value iteration algorithm given in Algorithm~\ref{algo:Value_Iter} is a dynamic programming method used to solve MDPs in infinite-horizon problems. 
The algorithm updates the values $V(s)$ at each iteration, for every states $s \in S$, converging to the optimal value function $V^{\star}$. 
It begins with arbitrary initial values $V_0(s), \forall s\in S$ and improves them by repeatedly applying the Bellman optimality equation given in Theorem~\ref{th:mdp-bellman}. 
The stopping criterion is based on the maximum change in the value function. 
Once the algorithm stops it returns a value function $V_n$, from which an optimal policy $\pi^{\star}$ can be derived (even though $V_n$ may not equal $V^*$). 
Note that the optimal policy obtained this way is stationary. 
The convergence of the algorithm to the optimal policy is guaranteed for discounted factors $0<\gamma<1$ and sufficiently small $\epsilon$, see \cite{bellman_dynamic_1958}.\\

\begin{algorithm}[tp]
\caption{Value Iteration Algorithm}
\label{algo:Value_Iter}
\begin{algorithmic}[1]
    \REQUIRE Discount factor $\gamma \in (0, 1)$, convergence tolerance $\epsilon$
    \FOR{$s \in S$}
        \STATE  $V_0(s)\leftarrow 0$ 
    \ENDFOR
    \STATE $n\leftarrow 0$
    \STATE $\delta \leftarrow +\infty$
    \WHILE{$\delta > \epsilon(1-\gamma)\/\gamma$}
        \FOR{$s \in S$}
                \STATE $V_{n+1}(s) \leftarrow \min_{a \in A} \sum_{s' \in S}P(s' | s,a)(c(s,a, s') + \gamma V_n(s'))$
        \ENDFOR
        \STATE $\delta \leftarrow \max_{s\in S} |V_{n+1}(s)-V_n(s)|$
        \STATE $n\leftarrow n+1$
    \ENDWHILE
    \FOR{$s \in S$}
        \STATE $\pi^{\star}(s) \leftarrow \argmin_{a \in A}\sum_{s' \in S}P(s' | s,a) (c(s,a, s') + \gamma V_{n+1}(s'))$
    \ENDFOR
    
\end{algorithmic}
\end{algorithm}
%
\paragraph{{\bf Policy iteration}}
%
The policy iteration algorithm given in Algorithm~\ref{algo:Policy_Iter} is closely related to the value iteration algorithm. While value iteration focuses on iterating over value functions, policy iteration directly updates the policies themselves.
After initialization with an arbitrary initial policy, the algorithm operates by iteratively improving the policy as follows.
\begin{enumerate}
    \item \emph{Policy Evaluation}: this step computes the value of each state under the current policy by solving the associated system of linear equations or iteratively estimating the value function.
    \item \emph{Policy Improvement}: this step updates the policy by deriving a \emph{greedy} policy with respect to the current value function. This involves selecting actions that maximize the expected return for each state.
\end{enumerate}
These steps are repeated, with the newly updated policy being evaluated and improved in turn. The process continues until the policy converges to the optimal policy, which remains unchanged in subsequent iterations.
Policy iteration is guaranteed to compute the optimal policy in a finite number of steps, see \cite{howard_dynamic_1962}. The number of iterations required is polynomial in the number of states and actions of the MDP, making it an efficient algorithm for solving MDPs. 
%
\begin{algorithm}[tp]
\caption{Policy Iteration Algorithm}
\label{algo:Policy_Iter}
\begin{algorithmic}[1]
    \REQUIRE Discount factor $\gamma \in (0, 1)$, initial policy $\pi_0$
    \FOR{$s \in S$}
        \STATE $\pi(s)\leftarrow \pi_0(s)$
    \ENDFOR  
    \REPEAT
        \STATE  $V\gets$ solution of the following linear program in variables $(V(s), s\in S)$ \hfill \COMMENT{Policy Evaluation}
        \[
        V(s) = \sum_{s' \in S} P(s'|s, \pi(s)) \left(c(s, \pi(s), s') + \gamma V(s')\right), \quad \forall s \in S
        \]
        \STATE Update greedy policy $\pi$ \hfill\COMMENT{Policy Improvement} 
        \[
        \pi(s) \leftarrow \argmin_{a \in A} \sum_{s' \in S} P(s'|s, a) \left( c(s, a, s') + \gamma V(s') \right), \quad \forall s \in S
        \]
    \UNTIL{$\pi$ does not change}
\end{algorithmic}
\end{algorithm}
~\\

%
\paragraph{{\bf Limitations and extensions}}
%
The three algorithms discussed above rely on dynamic programming equations, which represent one of the most efficient methods for exploring the policy space. While dynamic programming is guaranteed to yield an optimal policy in theory, it often becomes computationally infeasible for real-world applications. This limitation arises from the curse of dimensionality, a phenomenon identified by \cite{bellman_dynamic_1958}, where the computational complexity grows exponentially with the dimensionality of the state and action spaces, especially if those are continuous.

One approach to mitigate this complexity is to employ piecewise discrete  policies, enabling the application of dynamic programming techniques (see, for example, \cite{bertsekas_convergence_1975}). Discretization simplifies the problem by drastically reducing the dimensionality of the state space, making computations more manageable. However, this reduction comes at the cost of precision, often resulting in suboptimal policies (see \cite{feng_dynamic_2004}).
Additionally, the process of discretizing a state space is itself challenging. It can be time-consuming, difficult to implement effectively, and prone to introducing inaccuracies, especially as the dimensionality of the problem increases. 

Given these limitations, alternative approaches such as approximation methods and RL algorithms -- explored in the following sections -- are often preferred. These methods provide practical solutions for handling high-dimensional discrete or continuous spaces more effectively, frequently outperforming discretized dynamic programming in terms of both policy quality and computational efficiency.

For simplicity, the remainder of this discussion will focus on finite-horizon problems to streamline the analysis and facilitate the exploration of these alternative methods.
%
\subsubsection{Monte-Carlo simulations-based algorithms} 
\label{MCP}
\label{sub:MCP}
%
For finite-horizon MDPs with a known starting state, solving the MDP involves exploring a {\em min-expectation} tree. The root of this tree corresponds to the initial state, $s_0$. Each action $a_0\in A$ generates a successor action node. Each action node $a_0$ further branches into successor state nodes for all  $s_1$ such that $P(s_1|s_0,a_0)>0$. This process continues across time steps, with the tree's leaves representing the final states $s_H$. 
Solving the MDP using exact dynamic programming involves fully constructing this tree and applying backward induction. Backward induction labels each state node $s_t$ with its optimal value  $V_t^*(s_t)$, which represents the return obtained by following the optimal policy $\pi^*$ from $s_t$ until the horizon $H$.  
However, backward induction is computationally intensive because it requires exploring the entire tree. Even storing the complete tree is infeasible in many cases, though depth-first exploration with pruning of dominated actions can help mitigate this storage issue.

To address these computational challenges, Monte Carlo Planning methods, and in particular the family of Monte Carlo Tree Search (MCTS) algorithms, offer an efficient alternative. These heuristic-free approaches, such as the Upper Confidence bound applied to Trees (UCT) algorithm \cite{kocsis_bandit_2006}, focus on exploring the tree through simulations. MCTS algorithms prioritize exploration of the most promising parts of the tree—areas that show high potential or require further evaluation to refine their estimated values. By doing so, they avoid exhaustive exploration, making them well-suited for large or complex MDPs (see, \textit{e.g.}, \cite{coulom_efficient_2006}). 

There exist many variants of MCTS \cite{browne_survey_2012, swiechowski_monte_2023}, however most of them iteratively update a search tree (which is not exactly the same as the dynamic programming tree) until a computation budget is exhausted, at which point an action, computed for the root node, is returned. 
A tuple $\langle s, a, \nu, \rho \rangle$ is attached to each node $v$.
Component $v.s=s\in S$ is a state, $v.a=a\in A$ is the action that led to state $v.s$, $v.\nu=\nu\in\mathbb{N}$ is the number of visits to node $v$ following the current edge (\textit{i.e.} the number of times $v.a$ was applied in $v'.s$ and led to $v.s$, if $v'=Par(v)$ is the parent node of $v$), $v.\rho=\rho\in\mathbb{R}$ is an estimate of the value function $V^*(v.s)$ at state $v.s$. 
Initially, the tree consists of a single root node $v_0$ with attached tuple $\langle v_0.s=s_0, v_0.a=\emptyset, v_0.\nu=1, v_0.\rho=0 \rangle$.
At any step, the algorithm loops over the following four successive substeps to update the  current tree, see Figure \ref{mcts}. 
\begin{enumerate}
\item \textbf{Selection.} 
Starting from the root node, child nodes are selected recursively.
\begin{itemize}
    \item If $v$ is the current node, we check that it has a child node $v^a$ for every $a\in A$, \textit{i.e.} if all available actions have been tried in state $v.s$. Else, Selection phase is over (condition I). 
    \item Then, an action $a$ is selected, using a {\em tree policy} and a successor state $s'\sim P(\cdot|v.s,a)$ is simulated. 
If there exist $v'$ such that $v'.s = s'$ and $v'.a=a$, then $v'$ is selected and the Selection phase goes on from $v'$. Else, Selection phase is over (condition II).
\end{itemize}
\item \textbf{Expansion.}
When Selection phase is over, the tree is expanded.
A new child node of $v$ (selected node at the end of the Selection phase) is added. It is obtained  
\begin{itemize}
    \item either (if condition I was met) by playing a yet unplayed action $a\in A$ and simulating $s'\sim P(\cdot|v.s,a)$,
    \item or (if condition II was met) corresponding to the newly simulated state $s'$.
\end{itemize}
Let $v'$ be this node and $\langle v'.s=s', v'.a=a, v'.\nu=1,v'.\rho=0\rangle$ the attached tuple.
\item \textbf{Simulation.} 
When a node $v'$ has been added in the expansion phase, a single trajectory is simulated, starting from $v'.s$, using a {\em default policy} (for example, successive actions may be selected uniformly at random). Let $\mathcal{C}$ be the total (discounted in case of a discounted MDP) cost collected along the simulated trajectory. Then, $v'.\rho \leftarrow\mathcal{C}$. 
\item \textbf{Backpropagation.} 
The $Q$-values (see definitions \ref{def:Qdiscount} and \ref{def:Qtotal}) of the ancestor nodes of $v'$ are updated recursively:  
For every ancestor nodes $v$ of $v'$, selected in backward order,
$$
v.\rho \leftarrow \min_{a\in A} \frac{\sum_{\tilde v\in Succ(v), \tilde v.a=a} \tilde v.\nu\times (c(v.s, a, \tilde v.s)+\gamma \tilde v.\rho)}{\sum_{\tilde v\in Succ(v), \tilde v.a=a} \tilde v.\nu}.
$$
\end{enumerate}
%
\begin{figure}[tp]
    \scalebox{0.6}{
    \begin{tikzpicture}
level distance=10mm
\tikzstyle{algnode} = [rectangle, rounded corners, minimum width=3cm, minimum height=0.8cm,text centered, draw=black, fill=red!20]

\begin{scope}
\node [algnode] (Se) at (0,1) {\Large Selection};
\node [algnode] (E) at (6,1)  {\Large Expansion};
\node [algnode] (Si) at (12,1) {\Large Simulation};
\node [algnode] (B) at (18,) {\Large Backpropagation};
\draw[-{Stealth[length=3mm, width=2mm]}, line width=2pt] (Se) -- (E);
\draw[-{Stealth[length=3mm, width=2mm]}, line width=2pt] (E) -- (Si);
\draw[-{Stealth[length=3mm, width=2mm]}, line width=2pt] (Si) -- (B);

\draw[rounded corners, line width=2pt] (B) -- (18,2.5) -- (0,2.5) -- (Se);

\draw [-{Stealth[length=3mm, width=2mm]}] (0,2) -- (Se);

\begin{scope}[]
\node[circle,line width=2pt,draw]{0}
    child{node[circle,line width=2pt,draw]{1}  edge from parent[line width=2pt,-stealth]
        child{node[circle,line width=1pt,draw]{4} edge from parent[line width=1pt,-]}
        child{node[circle,draw]{5} 
            child{node[circle,draw, line width=1pt]{8} edge from parent[line width=1pt,-]}
            child[missing]
        }
    }
    child{node[circle,draw]{2}
        child[missing]
        child{node[circle,draw]{6}}
    }
    child{node[circle,draw]{3}
        child[missing]
        child{node[circle,draw]{7}}
    };
\end{scope}
\hspace{6cm}
\begin{scope}[]
\node[circle,draw]{0}
    child{node[circle,draw]{1}
        child{node[circle,draw]{4}}
        child{node[circle,draw]{5} 
            child{node[circle,draw]{8}}
            child{node[circle,line width=2pt,draw]{9} edge from parent[line width=2pt,-]}
        }
    }
    child{node[circle,draw]{2}
        child[missing]
        child{node[circle,draw]{6}}
    }
    child{node[circle,draw]{3}
        child[missing]
        child{node[circle,draw]{7}}
    };
\end{scope}
\hspace{6cm}
\begin{scope}[]
\node[circle,draw]{0}
    child{node[circle,draw]{1}
        child{node[circle,draw]{4}}
        child{node[circle,draw]{5} 
            child{node[circle,draw]{8}}
            child{node[circle,line width=2pt,draw]{9}
                child{node{{\bf default policy}} edge from parent[line width=2pt,dashed,-stealth]}
            }
        }
    }
    child{node[circle,draw]{2}
        child[missing]
        child{node[circle,draw]{6}}
    }
    child{node[circle,draw]{3}
        child[missing]
        child{node[circle,draw]{7}}
    };
\end{scope}
\hspace{6cm}
\begin{scope}[]
\node[circle,line width=2pt,draw]{0}
    child{node[circle,line width=2pt,draw]{1}  edge from parent[line width=2pt,stealth-]
        child{node[circle,line width=1pt,draw]{4} edge from parent[line width=1pt,-]}
        child{node[circle,draw]{5} 
            child{node[circle,line width=1pt,draw]{8} edge from parent[line width=1pt,-]}
            child{node[circle,draw]{9}}
        }
    }
    child{node[circle,draw]{2}
        child[missing]
        child{node[circle,draw]{6}}
    }
    child{node[circle,draw]{3}
        child[missing]
        child{node[circle,draw]{7}}
    };
\end{scope}
\end{scope}
\end{tikzpicture}
    }
    \caption{{\bf The four steps of Monte-Carlo Tree Search.} A child node is \emph{selected} (here node 5), then the tree is \emph{expanded} with a new child node (here node 9) and a trajectory is \emph{simulated} from this new node with default policy. Finally, the tree is \emph{backpropagated} to update the $Q$-values.}
    \label{mcts}
\end{figure}
%
Once the Backpropagation step is over, the algorithm goes back to the Selection step, until the computation budget is over.
At this point, an action $a^*$ is returned for the initial state $s_0$:
$$
a^* \leftarrow \argmin_{a\in A} \sum_{v'\in Succ(v_0), v'.a=a} \frac{\sum_{v'\in Succ(v), v'.a=a} v'.\nu\times (c(v.s, a, v'.s)+\gamma v'.\rho)}{\sum_{v'\in Succ(v), v'.a=a} v'.\nu}.
$$

An MCTS algorithm instance is determined by the choice of a {\em tree policy} and a {\em default policy}.
Typically, the tree policy is derived from a UCT criterion, balancing {\em exploitation} of the best action so far and {\em exploration} of under-used actions.
The default policy is often uniformly random.
However many choices of tree policies and default policies are possible, leading to different algorithms. The reader may refer to \cite{swiechowski_monte_2023} for a recent overview.

Note that, contrarily to exact dynamic programming, Monte-Carlo planning does not require complete knowledge of the MDP model, $(P, c)$, but only the faculty to {\em simulate} $P(\cdot|s,a)$ and to record experiences $(s,a,s',c)$. Even though the backpropagation step uses costs $c(v.s,a,v.s')$, these costs have been encountered during the expansion step and may be added to the tree data structure.
This feature is common with RL approaches, described in the next section.
%
\subsubsection{Reinforcement-learning algorithms}
\label{RL}
%
Reinforcement Learning (RL) is a computational framework that enables decision-makers to learn optimal mappings from situations to actions in order to minimize a numerical cost. In RL, the environment represents everything beyond the decision-maker's direct control.
In the context of MDPs, the environment corresponds to the mathematical structure that defines state transitions and costs, as described in Section~\ref{subsec:genericMDP}. Unlike traditional learning methods, RL does not require complete knowledge of this structure. Instead, it allows decision-makers to learn directly through interactions with the environment.

The RL process involves a cycle of interaction: the decision-maker observes the current state, selects an action, and receives feedback in the form of a cost, which reflects the quality of its decision (see Figure~\ref{fig:rl_schema}). This action also influences the next state of the environment. Through repeated iterations, the decision-maker refines its behavior, ultimately developing a policy that optimizes long-term costs, as outlined in Section~\ref{subsec:genericMDP}.
%
\begin{figure}[tp]
    \centering
    \scalebox{.9}{
    \begin{tikzpicture}[node distance=5cm, auto]
    \node[circle, minimum size=2cm, path picture={\node at (path picture bounding box.center) {\includegraphics[width=2cm]{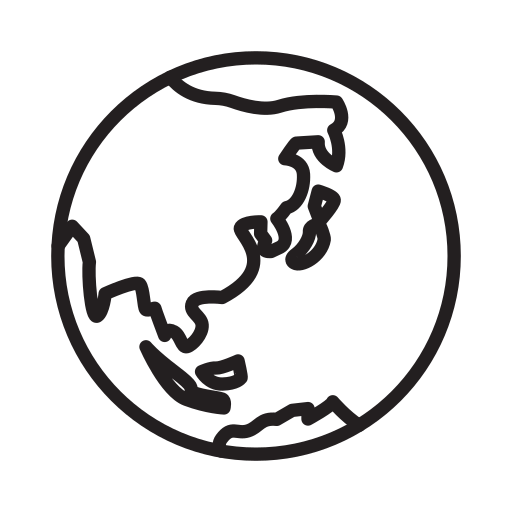}};}] (environmentPic) {};
    \node[below] at (environmentPic.south) (environment) {Environment};

    \node[circle, minimum size=2cm, path picture={\node at (path picture bounding box.center) {\includegraphics[width=2cm]{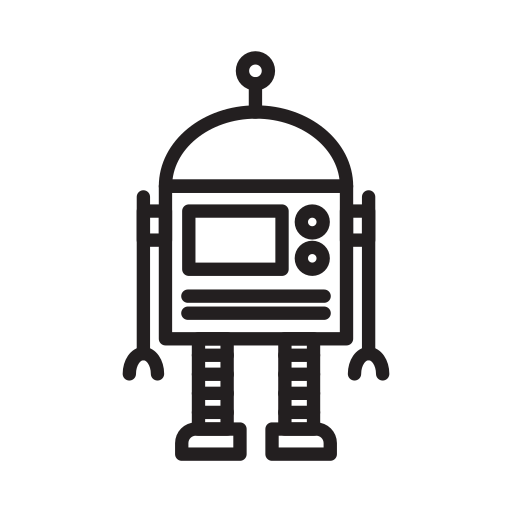}};}, right of = environmentPic] (agentPic) {};
    \node[below] at (agentPic.south) (agent) {Agent};
    
    \draw[->] (environment.south) -- ++(0,-0.5) -- node[midway, below] {observation} ++(5,0) -- (agent.south);

    \draw[->] (agentPic.west) -- (environmentPic.east) node[midway, above] {action};

    \draw[->] (environmentPic.north) -- ++(0,0.5) -- node[midway, above] {cost} ++(5,0) -- (agentPic.north);    
  
\end{tikzpicture}

    }
    \caption{{\bf Reinforcement Learning schematics.} The learning process involves a decision-maker observing a state, taking an action that affects the state, and receiving feedback through costs.}
    \label{fig:rl_schema}
\end{figure}

One of the fundamental challenges in RL is balancing the trade-off between exploration and exploitation. To minimize costs, an RL decision-maker exploits actions that have proven effective in the past while exploring new actions to uncover potentially better strategies. Striking this balance is essential to prevent the decision-maker from becoming stuck in suboptimal behaviors and to enable continuous policy refinement through exploration. 

Learning optimal actions in RL is approached from multiple perspectives, each characterized by distinct methodologies and challenges. This section begins with \emph{model-free methods}, which do not rely on a model of the environment's dynamics. Within this category, 
    \emph{Value-based} approaches focus on estimating the value function to derive optimal policies;
    \emph{Policy-based} methods directly search for the optimal policy within the policy space.
    
Following model-free approaches, we delve into \emph{model-based methods}, which involve constructing an explicit model of the environment to predict future states and associated costs. 
By structuring the discussion in this manner, we aim to provide a comprehensive overview of the key strategies used in RL, emphasizing the strengths and limitations of each approach. This section draws heavily on the seminal work by \cite{sutton_reinforcement_2018}, which offers an extensive survey of RL methodologies and their applications. 

Recent advancements in deep learning (\textit{e.g.} \cite{lecun_deep_2015, goodfellow_deep_2016}) have significantly enhanced RL capabilities, particularly in addressing challenges posed by high-dimensional state and action spaces. These developments have given rise to deep Reinforcement Learning (deep RL) algorithms, which are discussed throughout this section. These algorithms show considerable promise for advancing the field. 
For practitioners, practical implementation can begin by encoding their environment using the (python) \emph{Gymnasium} library introduced by \cite{brockman_openai_2016}. This library integrates seamlessly with the library \emph{RLlib} \cite{liang_rllib_2018}, which offers a suite of pre-implemented deep RL algorithms ready for use.\\

%
\paragraph{{\bf Q-function}} 
%
Most classical reinforcement-learning algorithms (\textit{e.g.} $Q$-learning introduced in \cite{watkins_Q-learning_1992}) make use of the following notion, related to the optimal-value function of an MDP: the {\em Q-function}.
%
\begin{defin}[$Q$-function, discounted criterion]\label{def:Qdiscount}
Let $\langle S, A, H, P, c \rangle$ be a given infinite-horizon MDP and $\pi\in\Pi$ be an arbitrary policy.
The $Q$-function, for the infinite-horizon discounted criterion, represents the expected cost of selecting action $a \in A$ in the current state $s \in S$ and then following the policy $\pi$ forever
\begin{equation*}
     Q_\gamma^{\pi}(s,a) =  \mathbb{E}_s^{\pi} \left[c(s,a,S_1) + \sum_{t=1}^{\infty} \gamma^t c(S_t, A_t, S_{t+1}) \right] =  \sum_{s'\in S} P(s'|s,a)\left(c(s,a,s') + \gamma V_\gamma(\pi,s')\right).
\end{equation*}
Then, if $V^*_\gamma$ is the optimal value function of the MDP for the discounted criterion, $Q^*_\gamma$, the optimal Q-function of the MDP is defined as
\begin{equation*}
    Q_\gamma^*(s,a) = \sum_{s'\in S} P(s'|s,a)\left(c(s,a,s') + \gamma V_\gamma^*(s')\right),\ \forall (s,a)\in S\times A.
\end{equation*}
\end{defin}
%
In the finite-horizon discounted case, the $Q$-function is defined as follows.
%
\begin{defin}[Q-function, total cost criterion]\label{def:Qtotal} For any $0\leq t < H$, the Q-function at time $t$ represents the expected cost of selecting action $a \in A$ in the current state $s \in S$ and then following policy $\pi$ up to horizon $H$.
    \begin{equation*}
        Q_{\{t,H\}}^{\pi}(s, a)  =  \mathbb{E}^{\pi}_s\left[ c(s,a,S_1) + \sum_{p=t+1}^{H-1} c_p(S_p, A_p, S_{p+1}) + C(S_H)\right]
         =  \sum_{s'\in S} P(s'|s,a)\left(c(s,a,s')+V_{t+1}(\pi,s')\right).
    \end{equation*}
Then, if $\{V*_t\}_{t=0..H}$ form an optimal value function, the optimal Q-function is:
    \begin{equation*}
        Q_{\{t,H\}}^*(s, a)= \sum_{s'\in S} P(s'|s,a)\left(c(s,a,s')+V_{t+1}^*(s')\right).
    \end{equation*}
\end{defin}
\color{black}
%
\paragraph{{\bf Model-free approaches: value-based algorithms}} 
%
Model-free value-based RL is a key area within artificial intelligence that focuses on learning optimal policies by estimating value functions, without the need for an explicit model of the environment. This approach is especially valuable in complex environments where creating an accurate model is either difficult or impractical. 

We begin with the simplest cases, where the state and action spaces are small enough for the value functions to be approximated using arrays. These methods are known as tabular algorithms, see \cite{sutton_reinforcement_2018}. 
One of the foundational algorithms in this category, $Q$-learning was introduced by \cite{watkins_learning_1989}. It is an \emph{off-policy} method that allows a decision-maker to learn a policy different from the one it follows during training. This distinction enables the decision-maker to learn the optimal policy while adhering to a more exploratory behavior policy. In $Q$-learning, the decision-maker interacts with the environment and updates its $Q$-values based on received costs and the minimum expected sum of future costs, as defined by the Bellman equation. The iterative process is continued until a stopping criterion is reached. \cite{watkins_Q-learning_1992} show that provided the action state is finite, the $Q$-values converge to the optimal action values almost surely. Through this process, the decision-maker identifies the best possible action for any given state, as outlined in Algorithm~\ref{alg:$Q$-learning}.

$Q$-learning is widely recognized for its simplicity and effectiveness in the case of finite state and action spaces. However, \cite{watkins_Q-learning_1992} also show that it can be computationally expensive and slow to converge, particularly in large state spaces. 
%
\begin{algorithm}[tp]
\caption{$Q$-learning}
\label{alg:$Q$-learning}
\begin{algorithmic}[1]
    \REQUIRE discount factor $\gamma>0$, tolerance $\epsilon>0$, learning rate $\alpha\in (0,1)$, arbitrary initial values for $Q(s,a), \forall s\in S, a\in A$
    \FOR{each episode}
        \STATE Initialize $s$ randomly
        \FOR{each step of episode}
            \STATE Choose $a\leftarrow \argmax_{a\in A} Q(s,a)$ with probability $\epsilon$ or uniformly in $A$ with probability $(1-\epsilon)$ ($\epsilon$-greedy choice)
            \STATE Apply action $a$, observe $c$, $s'$
            \STATE $Q(s,a) \leftarrow Q(s,a) + \alpha \left[c + \gamma \min_{a'\in A} Q(s',a') - Q(s,a)\right]$
            \STATE $s \leftarrow s'$
        \ENDFOR
    \ENDFOR
\end{algorithmic}
\end{algorithm}
%
In \cite{rummery_-line_1994}, the authors introduce the State-Action-Reward-State-Action (SARSA) algorithm, which is also model-free and value-based. The name SARSA reflects the sequence of events the algorithm uses to update its knowledge. Unlike $Q$-learning, SARSA updates its $Q$-values based on the action actually taken by the decision-maker, rather than the optimal action, making it an \emph{on-policy} algorithm in contrast to $Q$-learning's off-policy approach. \cite{singh_convergence_2000} show that the SARSA algorithm is guaranteed to converge if every state-action pair is visited infinitely often, but its convergence can be slow, especially in environments with large state spaces.
Another alternative to $Q$-learning is Temporal-Difference (TD) learning, introduced by \cite{sutton_learning_1988}, which combines Monte Carlo methods with dynamic programming.

Over the last few decades, extensions to these algorithms have been developed to address their limitations. For instance, an on-policy $Q$-learning extension was created to achieve better computational complexity bounds \cite{strehl_pac_2006}, perform in multi-agent environments \cite{shoham_if_2007}, and address the issue of overestimated action values \cite{hasselt_double_2010}. 

While tabular algorithms and their extensions work well in finite and relatively small state spaces, many real-world problems involve larger or even continuous state spaces. In cases where a tabular representation of the $Q$-function is impractical, \cite{ tsitsiklis_convergence_2003} proposed to approximate the $Q$-function using function approximation techniques.
Here, the value function is not represented as a table but as a parameterized function with a weight vector $\mathbf{w} \in \mathbb{R}^d$. 
The approximate value of a state $s$ given weight vector $\mathbf{w}$ is denoted by $\hat{v}(s, \mathbf{w})$. The approximation 
$\hat{v}$ can take various forms such as a linear function, a decision tree, or others. \cite{sutton_reinforcement_2018} extended the theoretical results for methods using function approximation to cases involving partial observability. 

Deep RL extends function approximation techniques by using a multi-layer artificial neural network build the approximate the value function  $\hat{v}$. This approach provides a powerful method for function approximation, albeit without theoretical guarantees \cite{sutton_reinforcement_2018}.
Many Deep RL algorithms are build upon foundational tabular methods. A notable example is the Deep $Q$-Network (DQN), introduced by \cite{mnih_playing_2013} and illustrated in Figure~\ref{fig:dqn_schema}. Subsequent advancements, such as \emph{Double DQN} \cite{hasselt_deep_2016}, \emph{Dueling DQN} \cite{wang_dueling_2016}, and \emph{Rainbow} \cite{hessel_rainbow_2018}, have further enhanced the efficiency and performance of DQN-based approaches.

As depicted in Figure~\ref{fig:dqn_schema}, when dealing with partially observable environments, the neural network receives observations rather than full state information. In such cases, the optimal policy depends on the entire history of observations, as discussed in Section~\ref{sec:pomdp}. Addressing this challenge, \cite{matthew_hausknecht_deep_nodate} introduced \emph{Recurrent Replay Distributed DQN} (R2D2), which uses recurrent neural networks (RNNs) instead of convolutional networks. RNNs are capable of retaining information from past observations, making them particularly well-suited for POMDPs.
%
\begin{figure}[tp]
    \centering
    \scalebox{1}{
    \begin{tikzpicture}[node distance=2cm, auto]
    \node[draw,rectangle,minimum width=2cm, minimum height=1cm,  rounded corners] (dqn) {Neural Network DQN};

    \node[draw,rectangle,minimum width=2cm, minimum height=1cm,rounded corners, left = 1cm of dqn] (env) {Environment};
  
    \node[minimum width=2cm, minimum height=1cm, right = 1cm of dqn] (q_valeur) {$
    \left\{\begin{array}{ll}
            Q(s_t,a_1) \\
            \cdots \\
            Q(s_t,a_6)
        \end{array}\right\}$};
  
    \draw[->] ([yshift=0.1cm]env.east) -- node[above] {$s_t$} ([yshift=0.1cm]dqn.west);

    \draw[->] (dqn.east) -- (q_valeur.west); 
    \draw[->] (q_valeur.south) -- ++(0,-0.5) -- node[below] {$a^{\star} = \argmax_{a_i} Q(s_t,a_i)$} ++(-8,0) -- (env.south); 

    \node[draw, rectangle, minimum width=2cm, minimum height=1cm,  rounded corners, above of = q_valeur, align=center] (loss) {Loss calculation \\ $L = [(c(s_t,a^{\star},s_{t+1})+\gamma\min_{a'}Q(s_{t+1},a'))-Q(s_t,a^{\star})]^2$};

    \draw[->] ([xshift=-3.7cm]loss.south) -- node[right] {Gradient loss} ([xshift=0.40cm]dqn.north); 
    \draw[->] ([xshift=-0.9cm]env.north) -- ++(0,1.5) -- node[above, anchor=south east] {$(s_{t+1}, c(s_t, a^{\star}, s_{t+1}))$}  (loss.west); 


    \begin{scope}[on background layer]
        \node[fit=(loss)(dqn)(q_valeur), inner xsep=2mm, inner ysep=2mm, fill = bleu_ql, opacity = 0.1] (framework) {};
    \end{scope}

    \node[minimum width=2cm, minimum height=1cm,  rounded corners, above of = framework, bleu_ql, align=center] (agent) {Agent};

\end{tikzpicture}
    }
    \caption{{\bf Deep $Q$-Network algorithms.} DQN algorithms send the current state of the environment to a neural network, which calculates the $Q$-values for each possible action in the MDP. The algorithm then selects the action that minimizes these $Q$-values. Once the action has been chosen, the environment moves to the next state and the associated cost is calculated. This cost is used to update the loss function, which measures the difference between the predicted $Q$-values and the actual values. Finally, the network weights are adjusted by backpropagation to improve future predictions.}
    \label{fig:dqn_schema}
\end{figure}
%
These algorithms collectively demonstrate the flexibility and potential of model-free value-based RL methods in tackling complex decision-making challenges across various domains.\\

%
\paragraph{{\bf Model-free approaches: policy-based algorithms}}
%
While value-based RL algorithms generally learn a $Q$-function, from which a (nearly-) optimal policy can be obtained, policy-based algorithms directly learn a policy from experiments.
The dynamic programming policy-based algorithm is {\em policy iteration}, which alternates phases of {\em evaluation} of the current (deterministic) policy and {\em improvement} of the current policy (see Section \ref{sub:DP}). 

When the model of the MDP is unknown (RL context) and/or when state and action spaces are too large or even continuous, policy iteration is no longer applicable.
In this context, {\em Policy Gradient} approaches (see \textit{e.g.} \cite{deisenroth_survey_2013, silver_deterministic_2014}) suggest to use continuously parameterized policies and replace policy improvement steps with {\em gradient descent} steps, to improve the current policy (defined by its continuous parameters). 
With parameterized policies, the value of the policy in an arbitrary step can be considered itself as a function of the parameters of the policy. As such, the gradient of the value function can be estimated using stochastic policy gradient  algorithms (the value function in state $s$ is an expectation of the sum of costs obtained along possible trajectories).

More precisely, model-free {\em policy gradient} algorithms exploit the {\em gradient theorem} from \cite{sutton_policy_1999} which states
\begin{equation*}
    \nabla_\theta V(\pi_\theta, s_0) = \mathbb{E} \left[ \nabla_\theta\log\pi_\theta(A|S)\tilde{Q}(S,A)\right],
\end{equation*}
where the expectation is taken with respect to $S\sim\mu^{\pi_\theta}_{s_0}$ and $A\sim\pi_\theta$, $\mu^{\pi_\theta}_{s_0}$ is the discounted occupation measure of the Markov chain followed by starting in $s_0$ and applying policy $\pi_\theta$, $\pi_\theta(\cdot|\cdot)$ is any differentiable parameterized representation of a stochastic policy (with parameter vector $\theta$) and $\tilde{Q}$ is an estimation of the {\em true} $Q$-function (which can be computed, \textit{e.g.} by simulation).
This policy gradient expression makes it easy to estimate the gradient of the current policy by simulation, since the gradient of $\pi_\theta$ may be known explicitly, and simulations can be used to estimate $\tilde{Q}$ as well as the global expectation. 
Note also that policy-gradient algorithms naturally handle continuous state and action spaces. 

One limitation of these approaches is the need to arbitrarily choose a parameterized policy, something that may not always be ideal. Moreover, estimating the gradient through simulations often requires substantial computational resources, even if the model is known. The challenge becomes more significant when the model is unknown, as the approach then depends on datasets of observed MDP trajectories. 

On the other hand, policy iteration approaches are highly versatile and capable of addressing remarkably complex MDP problems. 
This is especially true with recent advancements handling neural network representations of policies.
Proximal Policy Optimization (PPO, \cite{schulman_proximal_2017}) is a popular example of a deep RL algorithm with policy iteration. 
In PPO, instead of approximating the value function as in deep $Q$-networks, the policy itself is parameterized by a neural network. 
This network directly associates states (or observations in the case of partial observability) with actions, optimizing the policy to minimize expected costs.
In PPO, the algorithm alternates between collecting data from the environment using the current policy and updating the policy on the basis of the data collected. 
As with DQN, PPO can be extended to handle partial observability by incorporating recurrent neural networks instead of feedforward networks. 
This allows the policy to retain past observations and improve decision-making in environments with hidden states, such as POMDPs.

Unlike DQN, which generates a deterministic policy by selecting the action with the highest $Q$ value, PPO generates a stochastic policy.
A stochastic policy is one in which the decision-maker selects actions according to a probability distribution over the available actions, rather than always choosing a fixed action for a given state. 
This stochastic nature allows the decision-maker to explore more diverse behaviors, which makes it particularly useful in environments where there is uncertainty or partial observability, as it can prevent the policy from getting stuck in sub-optimal solutions by occasionally trying different actions.\\

%
\paragraph{{\bf Model-based algorithms}}
%
This class of algorithms operates in an iterative loop: the decision-maker begins by collecting data to update the current model of the environment's dynamics. Using the refined model, the decision-maker learns and optimizes its policy within a simulated environment. Once a new policy is derived, the decision-maker applies it in the real world, gathering additional data. This new data is then used to further refine the model. 
The cycle of data collection, model updating, and policy optimization repeats continuously, enabling the decision-maker to improve the performance over time. This iterative process is depicted in Figure~\ref{model_based}.
%
\begin{figure}[tp]
    \centering
    \scalebox{1.2}{
    \begin{tikzpicture}[node distance=1cm, auto]
    \node[draw,font=\scriptsize,rectangle,minimum width=2cm, minimum height=0.7cm,  rounded corners] (env) {Environment};

    \node[draw,font=\scriptsize,rectangle,minimum width=2cm, minimum height=0.7cm,  rounded corners, below = of env] (data) {Experiences $(s,a,s',c)$};

    \node[draw,font=\scriptsize,rectangle,minimum width=1cm, minimum height=0.7cm,  rounded corners, below right=of data] (model) {Model};

    \node[draw,font=\scriptsize,rectangle,minimum width=1cm, minimum height=0.7cm,  rounded corners, below left=of data] (policy) {Policy};
  

    \draw [->]  (data.east) -- node[right,font=\scriptsize] {Learning} (model.north);

    \draw [->]  (env.south) -- node[right,font=\scriptsize] {Generating} (data.north);

    \draw [->]  (policy.north) -- node[left,font=\scriptsize] {Acting} (env.west);

    \draw [->]  (model) -- node[above,font=\scriptsize] {Improving} (policy.east);
  




\end{tikzpicture}
    }
    \caption{{\bf Model-based Reinforcement Learning.} Data is collected to update the environment's model knowledge, which is then used to learn and optimize the policy in a simulated environment. This updated policy is applied, yielding new data.}
    \label{model_based}
\end{figure}

Atkeson \textit{et.al} \cite{atkeson_comparison_1997} show that this provides better results than model-free approaches with less data.
For the sake of simplicity, we only focus on Bayesian Reinforcement Learning (BRL) in this survey. 
However, readers interested in model-based methods are encouraged to explore the comprehensive survey on model-based RL provided in \cite{luo_survey_2022}, as well as the state-of-the-art DreamerV3 algorithm proposed by \cite{hafner_mastering_2024}.

Introduced in \cite{mohammad_ghavamzadeh_bayesian_2015}, BRL 
leverages methods from Bayesian inference to model uncertainty in the environment's dynamics. 
The Bayesian framework, described in section~\ref{sub:bamdp}, allows for principled incorporation of prior knowledge and can help in balancing the exploration-exploitation tradeoff effectively. 

A Bayesian approach to dynamic programming was introduced by \cite{strens_bayesian_2000}, where a model is sampled from the posterior distribution and dynamic programming is applied to the sampled model. 
To improve accuracy, the model is periodically resampled to account for new information. 
While effective, this method can be computationally slow, especially in complex environments.
In \cite{duff_optimal_2002}, a method is proposed that utilizes Finite State Controllers (FSCs) to represent Bayesian policies in a more compact form. 
FSCs are graph-based structures where the nodes represent memory states and the edges correspond to transitions between states, defined by the tuple \((s, a, s')\). 
The algorithm then searches the space of possible FSCs to identify the optimal controller, effectively reducing the complexity of finding the best policy.
The Bayesian Exploration Exploitation Tradeoff in LEarning (BEETLE) algorithm \cite{poupart_analytic_2006} builds upon the \textsc{Perseus} algorithm from \cite{spaan_perseus_2005}, originally developed for solving continuous POMDPs.
The key innovation of BEETLE is to reformulate the BAMDP as a POMDP that operates over the space of states and transition functions. 
This approach leverages the strengths of POMDP techniques to address the complexities inherent in BAMDPs. 

Due to the high complexity associated with model-based BRL, most approaches have been limited to very small state spaces, typically comprising only $10$ to $20$ states. 
As the number of states increases, a large amount of data is required to accurately learn a robust model. 
In addition, planning methods become difficult to implement in larger state spaces, as they require consideration of the full range of possible posterior distributions, which significantly complicates the decision-making process.

Online tree search methods have become a promising solution for solving BAMDPs by evaluating potential future action sequences from the current state, building a decision tree to guide actions. Unlike offline methods, online approaches focus on finding effective policies for the current belief state, making them more feasible for complex problems. A notable example is Bayesian Sparse Sampling, introduced in \cite{wang_bayesian_2005}, which non-uniformly explores the tree of future states and actions, reducing computational complexity by concentrating on the most promising paths. Another well-known approach is MCTS algorithms \cite{browne_survey_2012} (see Section \ref{sub:MCP}), which balance exploration and exploitation through random sampling.

Building on MCTS, \cite{guez_efficient_2012} developed Bayesian Adaptive Monte Carlo Planning (BAMCP) to extend tree search for BAMDPs, incorporating Bayesian updates to inform action exploration. Despite improvements over simpler methods, BAMCP remains computationally demanding, especially in large state and action spaces, and lacks function approximation, limiting its use in very large or continuous state spaces.

In BRL, significant progress has been made with algorithms designed to handle continuous state spaces and complex environments. BayesianExplore (BE) \cite{killingberg_bayesian_2023} is a standout algorithm, leveraging posterior sampling of value functions for efficient exploration. Another key advancement is Deep Bayesian Reinforcement Learning (Deep BRL),  which has been particularly effective in dealing with continuous state spaces. Methods such as BRL-Deep Sparse Sampling (DSS) from \cite{grover_sample_2020} combine deep learning with sparse sampling to approximate value functions in high-dimensional continuous spaces where other methods often falter.

%
\subsection{Challenges addressed and challenges remaining}\label{sec:cclMDP}
%
MDPs have a long and successful history as a modeling framework for sequential decision-making under uncertainty, and a rich body of efficient algorithms exists for computing useful decision policies. Recent advances in the field have addressed several key challenges:
\begin{itemize}
\item \textbf{Partial observability.} When the state of the process cannot be directly observed, the POMDP framework provides a principled extension, at the cost of increased solution complexity.
\item \textbf{Model uncertainty.} When the MDP model itself is not fully known, two complementary strategies have emerged: enriching the model, as in the BAMDP framework, or bypassing explicit modeling altogether through simulation-based or reinforcement learning approaches.
\item \textbf{Scalability.} High-dimensional or continuous state and action spaces pose a fundamental computational challenge, which (deep) reinforcement learning methods have recently made tractable.
\end{itemize}
One challenge that remains comparatively under-addressed is the limited expressivity of the MDP dynamics model. In complex systems, parametrizing the transition functions can be quite technical without interpretability of the parameters. The Markovian assumption is restrictive in a number of real-world settings. It has been partially relaxed through frameworks such as Semi-Markov Decision Processes, which allow for non-exponential sojourn times and come with dedicated solution algorithms \cite{sutton_smdp_1999}.
A deeper form of this relaxation is offered by controlled PDMPs, whose dynamics are inherently richer but also harder to handle, due to their continuous-time and continuous-state nature. Bridging the gap between controlled PDMPs and MDPs is therefore a significant opportunity for both fields: it would bring increased expressivity to the MDP framework, while making the latest algorithmic advances in MDPs, including deep RL, available to tackle controlled PDMP problems.
\color{black}
%
\section{Controlled PDMP as member of the MDP family} 
\label{sec:cpdmp}
%
In this section, we demonstrate how impulse control problems for PDMPs can be reformulated within the MDP framework. {We also outline how this shift in perspective opens up new avenues for addressing our main challenges that have so far remained unresolved within the PDMP framework: hidden information, model knowledge uncertainty and practical resolution.} 
For example, as discussed in Section~\ref{ss:impulse}, formulating a rigorous approach to impulse control problems for PDMPs involving hidden jumps and noisy observations presents significant challenges, particularly in defining admissible strategies (see, for instance, \cite{almudevar_dynamic_2001} or \cite[Section 1.1]{costa_continuous_2013}). The MDP framework offers a robust structure for overcoming these challenges and for properly formalizing such problems. 
Here, we show how a fully observed impulse control problem for PDMP can be formulated as an MDP (Section~\ref{subsec:PDMP-MDP}), then turn to the partially observed version (Section~\ref{subsec:POPDMP-POMDP}). 
%
\subsection{Controlled PDMP as MDP}
\label{subsec:PDMP-MDP}
%
We start by encompassing the notion of impulse for PDMP in the notion of action for MDPs in Section~\ref{subsec:strategies}, then formally state a sub class of impulse control problems for PDMPs as MDPs in Section \ref{sssec:PDMDP}, and finish in Section \ref{ss:ex_PDMP_MDP} with a new version of our running medical example fitting this extended framework.
%
\subsubsection{Impulse control strategies and MDP policies} \label{subsec:strategies}
%
As detailed in Sections~\ref{ss:impulse} and \ref{subsec:genericMDP}, the notions of \emph{impulses} for PDMP and \emph{actions} for MDP slightly differ.
On the one hand, in impulse control, the decision-maker typically has two tasks
\begin{enumerate}
\item choose the next impulse date,
\item at an impulse date, select the point from which to restart the process.
\end{enumerate}
On the other hand, in the MDP framework, the decision-maker takes decisions at each stage, and decisions act on the dynamics of the process by changing the probability of observing new states through the controlled transition kernel $P$.

Both frameworks can be unified in full generality, as done for instance in \cite{dufour_optimal_2016}. This approach is mainly used to obtain theoretical properties of controlled PDMPs and is ill-adapted to practical resolution. The first main drawback is a possibly complex action space (typically as complex as the PDMP state space), ill-adapted to numerical resolution. The second main drawback is that to keep the free choice of continuous impulse dates with a explicit transition kernels, most authors add the jump times of the PDMP as decision dates. This directly prevents any extension to contexts where jumps are hidden. Such contexts are nevertheless crucial for applications. It is still possible to keep free impulse choice without adding the jumps of the PDMPs as MDP decision dates, but this requires a complex joint management of the PDMP time and MDP stages, and non explicit (but still simulatable) transition kernels.

In this Section, we choose to restrict the choice for the impulse strategies of the PDMP decision-maker. 
Our main aim here is to address the challenges raised by our running medical example, while keeping the exposition as simple as possible. In view of the numerical resolution challenge, especially with hidden information, we choose to work with  small finite action spaces MDPs. We also only consider a finite horizon in the sequel, for the same reasons. 
\color{black}

The first restriction concerns the control set. In the impulse control framework for PDMPs, the decision-maker chooses the new starting point of the process after an impulse, which means they may choose to change either the mode or the Euclidean variable of the process or both. In the sequel, we only allow the decision-maker to change the mode of the process. {This restriction is made to keep a small finite action set, and is in line with our running example where actions correspond to treatments.} 
In this view, we introduce a \emph{mode-augmented} PDMP framework in which we consider two types of modes, one corresponding to the decision-maker-chosen modes, the other to the modes the decision-maker cannot act on. Formally, we decompose the discrete state space $\mathsf{M}$ of a PDMP into a finite product set $L\times M$, where modes in $L$  can be decision-maker chosen while modes in $M$ cannot. 
%
\begin{defin}[Mode-augmented controlled PDMP]
    \label{def:ta-pdmp}
    A mode-augmented controlled Piecewise Deterministic Markov Process is a PDMP defined by the tuple $\langle E, \Phi, \lambda, Q\rangle$, where
    \begin{itemize}
        \item The state space $E$ has the specific form 
\begin{align*}
E = \{(\ell,m,\mathsf{x}), \ell\in L, m\in M, \mathsf{x} \in E_m^\ell\},
\end{align*}
for some finite product mode set $\mathsf{M}=L\times M$, and where $E_m^\ell=E_{\mathsf{m}=(\ell,m)}$ is some Borel subset of $\mathbb{R}^{d_\mathsf{m}}$.
        \item The flow $\Phi$ just satisfies the conditions from definition \ref{def:pdmp}, and for all $x=(\ell,m,\mathsf{x})\in E$, set
\begin{align*}
\Phi(x,t)=(\ell,m,\Phi_m^\ell(\mathsf{x})).
\end{align*}
        \item The jump intensity satisfies the conditions from definition \ref{def:pdmp} and for all $x=(\mathsf{m}=(\ell,m),\mathsf{x})\in E$, set
\begin{align*}
\lambda_\mathsf{m}(x)=\lambda_m^\ell(\mathsf{x}).
\end{align*}
        \item The jump kernel $Q$ satisfies the conditions from definition \ref{def:pdmp} with the additional constraint that $Q$ cannot change the value of $\ell$, as $\ell$ is to be chosen by the decision-maker. For all $x=(\ell,m,\mathsf{x})\in E$, set
\begin{align*}
Q(\cdot|x)=Q_{m}^{\ell}(\cdot|\mathsf{x}).
\end{align*}
    \end{itemize}
\end{defin}
%
From this rigorously defined mode-augmented PDMP, we can further drop the regime $\ell\in L$ from the state ${x}\in E$ and include it in the action instead to better fit the standard MDP notation.
Denote 
\begin{align*}
E_{{m}}=\bigcup_{\ell\in L}E_{{m}}^\ell,\quad 
E_M=\{({{m}},\mathsf{x}), {{m}}\in M, {\mathsf{x}} \in E_{{m}}\}. 
\end{align*}
We also adapt notations for the skeleton chain kernels $P_{t_{n+1}-t_n}$ defined in Section~\ref{sssec:chains} as follows
\begin{align}\label{def:Pskeleton_mdp}
P^\ell_{t_{n+1}-t_n}(B | m,\mathsf{x})
= P_{t_{n+1}-t_n}(\{\ell\}\times B | \ell,m,\mathsf{x})
= \mathbb{P}(X_{t_{n+1}}\in \{\ell\}\times B | X_{t_n}=(\ell,m,\mathsf{x})),
\end{align}
for any Borel set $B\subset E_M$ and $x=(\ell,m,\mathsf{x})\in E$. 

The second restriction concerns the impulse dates. 
In the MDP framework, decision dates are often deterministic and predefined. In contrast, within the impulse control framework for PDMPs, determining the timing of the next impulse is typically part of the decision-making process and is allowed to depend on the current state of the process. This introduces a key distinction: for a fixed real-time duration $\mathsf{H}$ of the impulse control task, different strategies may result in varying numbers of \textit{impulse dates}, meaning different numbers of decisions.

To reconcile this difference, we impose a maximum number $H$ of impulses, for example by restricting impulse dates to a finite time grid. For simplicity, we assume that decision dates are limited to a subset $\mathbb{T}$ of a uniform time grid ${0, \delta, \dots, H\delta}$ for some $\delta > 0$. As mentioned above, other more general approaches are possible, including a continuous space or state-dependent grids for the decision dates. Our choice is made to keep finite action spaces in view of the resolution challenge, and because it makes sense in our medical context or more generally in applicative contexts where regular inspections are scheduled.
\color{black}

Keeping track of both the control horizon $\mathsf{H}$ and the maximum number of decisions $H$ induces additional technicalities. We thus introduce an artificial decision, denoted $\cd$, which is triggered when the real-time horizon $\mathsf{H}$ is reached. This decision will be applied repeatedly as many times as needed to align with the decision horizon $H$, transitioning the process to an absorbing cemetery state $\Delta$. 
This necessitates the introduction of constraints on the set of actions, as decision-makers are not permitted to control the process beyond the time horizon. We thus incorporate constraints into the decision framework using constraint sets $K$ as mentioned in Section \ref{ss:extension_mdp}. 

In summary, an action $a$ for the MDP corresponding to the controlled PDMP will consist in both a decision-maker-chosen mode $\ell\in L$ and a time until the next decision $r\in \mathbb{T}$, to encompass both aspects of impulses for PDMPs. 
Other formalisms closer to the original impulse problem for PDMPs are possible. They involve more complex MDPs, with possibly continuous action spaces. In our approach, the strongest restriction is that the delay between consecutive impulses belongs to a fixed grid $\mathbb{T}$, and thus is not allowed to depend on the current state of the process, whereas in the original PDMP framework, decision dates are stopping-times, which allows much more flexibility. This flexibility may be retrieved by defining larger action and constraints spaces. The MDP framework can also increase the flexibility of the impulse control problem for PDMPs as MDP actions are allowed to have a random effect, whereas a PDMP impulse fixes the new starting point of the process.
%
\subsubsection{Formalism of the equivalent MDP}
\label{sssec:PDMDP}
%
We can now formally consider impulse control for a PDMP in an MDP framework. In order to avoid confusion between the continuous time index $t\in \mathbb{R}_+$ of the PDMP and the discrete time index $t\in\mathbb{N}$ of the MDP stages, we rename the latter with indices $n\in\mathbb{N}$.
%
\begin{defin}[Controlled PDMP as an MDP]
    \label{def:pdmp-to-mdp}
Let $\langle E, \Phi, \lambda, Q\rangle$ be a mode-augmented controlled PDMP. The corresponding MDP is defined by
 the MDP $\langle {S}, A, H, K, P, c,C \rangle$ with the following characteristics.
\begin{itemize}
    \item The state space is 
 ${S} = E_M\times\{0,\delta,\ldots ,(H-1)\delta\} \cup \{\Delta\}$, 
    where $\Delta$ is a cemetery state, and a time variable $t_n\in \{0,\delta,\ldots ,H\delta\}$ is added to keep track of the continuous-time horizon.
    \item The horizon $H$ is finite and equals $\mathsf{H}/\delta$ (which is supposed to be an integer).
    \item The action space is 
$  
    A = (L\times\mathbb T)\cup \{\cd\}
$, 
where $\cd$ is an artificial action to deal with the continuous-time horizon of the PDMP.
    \item The constraints set $ K\subset {S} \times  A$, is such that its sections $ K(s)=\{a\in  A; (s,a)\in K\}$ satisfy 
\begin{equation*}    
    K(\Delta)=\{\cd\},\quad K(s=(x,t_n))=L\times(\mathbb{T} \cap \{\delta,\dots,H\delta-t_n\}),
\end{equation*}
     to ensure that, in the cemetery state, the only action allowed is $\cd$, and that no decision can send the process beyond the horizon $\mathsf{H}=H\delta$. We assume that $\delta\in\mathbb{T}$ so that the constraints sets are never empty.
    \item The transition kernel $P$ is the transition kernel $P_{t_{n+1}-t_n}$ associated with the skeleton chain  of the mode-augmented PDMP, where the grid $(t_0,\dots,t_H)$ corresponds to the decision dates, combined with the deterministic increase of time. More precisely, if $s=(m_{t_n},\mathsf{x}_{t_n},t_n)$, and $a=(\ell, r)$, then $t_{n+1}=t_n+r$ and for any Borel subset $B$ of $E_M$, one has
\begin{align*}
    P(s'=(x'\in B,t_{n+1}) |s=(x,t_n),a=(\ell,r)) = P_{r}^\ell(B|x)\1_{t_{n+1}=t_n+r},
\end{align*}
where $P^\ell$ is defined in Eq. (\ref{def:Pskeleton_mdp}).
Additionally, the artificial action $\cd$ sends the process to the cemetery state: for all $s\in E$,
$
    P(\{\Delta\}|s,\cd) = 1
$. 
    \item The cost function $c$ should in principle be defined by
\begin{align*}
c(s,a,s') = c(s=(x,t_n),a=(\ell,r)) =\mathbb{E} \left[\int_0^r c_R(X_t)dt \big|\  X_0=(\ell,x) \right] + \tilde{c}_I(s,a),
\end{align*}
where $\tilde{c}_I(s,a)=c_I(X_{t_n}=(\ell_{t_n},x),X_{t_n^+}=(\ell,x))$ if $a=(\ell,r)$, $s=(x,t_n)$, and the decision is taken to move $X_{t_n}=(\ell_{t_n},x)$ to $X_{t_n^+}=(\ell,x)$. In practice, such a cost function is generally intractable, and alternative cost functions close to $c$, depending only on the values of $s,a,s'$ are considered instead. We provide an example in Section \ref{ss:ex_PDMP_MDP}. The terminal cost $C$ is equal to $c_T$.
\end{itemize}
\end{defin}
%
One of the main advantages of the MDP framework, is that it is now very easy to rigorously define decision rules $\pi_n$ as measurable functions from $\mathcal{S}$ onto $A$ such that $\pi_n(s)\in K_n(s)$ for all $s$ and $n$, and to define admissible policies, as exposed in Section \ref{subsubsec:policydefinition}. 
The controlled trajectory of the MDP following policy $\pi =\left(\pi_n\right)_{0\leq n < H}\in \Pi$ is defined recursively in Algorithm \ref{algo:controledPDMPtoMDP}. 
%
\begin{algorithm}[tp]
    \caption{Simulation of a trajectory of a controlled PDMP seen as an MDP controlled by policy $\pi=\left(\pi_t\right)_{0\leq t < H}$ starting from state $s_0$ up to the horizon $\mathsf{H}=\delta H$}
\label{algo:controledPDMPtoMDP}
\begin{algorithmic}[1]
    \STATE $S_0 \leftarrow s_0=(x_0,0)$
    \FOR{$n$ from $0$ to $H-1$}
        \STATE $A_n=(\ell_n,r_n) = \pi_n(S_n)$         
        \STATE $S_{n+1} \sim P(\cdot|S_n=(X_n,T_n),A_n
        )=(P^{\ell_n}_{r_n}(\cdot|X_n), T_n+r_n)$ \label{algo:PDMPtoMDPstep}
    \ENDFOR
\end{algorithmic}
\end{algorithm}
%
Line \ref{algo:PDMPtoMDPstep} corresponds to Algorithm \ref{algo:PDMP-SSA} for the mode-augmented PDMP starting from state $x=(\ell_n,X_n)$ at time $0$ and up to time horizon $r_n$. 
The cemetery state $\Delta$ ensures that all trajectories have the same length $H$, even if they do not have the same number of actual decisions ($a\neq \cd$).

Then the total expected cost of policy $\pi\in\Pi$ starting at $s_0\in S$ 
\begin{equation*}
 V_H({\pi},s_0)=\mathbb{E}^{\pi}_{s_0}\left[\sum_{n=0}^{H-1} c(S_n,A_n,S_{n+1})+C(S_H)\right],
\end{equation*}
defined in Section~\ref{subsec:genericMDP} is equivalent to $\mathcal{V}_H(\mathcal{S},x)$ defined in Section~\ref{ss:impulse},
 and the impulse control problem \emph{corresponds} to the optimization problem
\begin{equation*}
{V}(s_0)=\inf_{\pi\in \Pi}V_H(\pi,s_0).
\end{equation*}
To obtain a formal equality $\mathcal{V}(x_0) = V(s_0=(x_0,0))$ is possible but requires a perfect match between the cost functions and the sets of admissible policies. In particular, admissible policies for the controlled PDMP must satisfy the same restrictions as those introduced in Section \ref{subsec:strategies}. This will not be further discussed here as we chose the MDP framework specifically to avoid delving into the formal definition of admissible policies or cost functions for PDMPs.
%
\subsubsection{Running medical example as a controlled PDMP formalized as an MDP}
\label{ss:ex_PDMP_MDP}
%
We now consider a more interesting version of our running medical example as a controlled mode-augmented PDMP and then specify the corresponding MDP.\\

\textbf{Controlled PDMP} 
We consider again a model for the dynamics of the marker of a patient as a PDMP as described in Section~\ref{subsubsec:medicPDsMP}, with very different decision from those proposed in Section~\ref{ex:impulse_control_pdmp}. 
We now consider that impulses correspond to the choice of the next visit date to the medical center and of the treatment to be applied until this next visit. 
To do so, we consider a new mode and time augmented PDMP. Its state space is defined from the modes sets $L=\{0,1\}$ and $M=\{0,1,2,3\}$ and the state spaces
\begin{equation*}
E_0^\ell= \{\zeta_0\}\times [0,\mathsf{H}]^2,\ 
E_1^\ell= E_2^\ell=[\zeta_0, D)\times [0,\mathsf{H}],\ 
E_3^\ell=\{D\}\times [0,\mathsf{H}],
\end{equation*}
for all $\ell\in L$. 
In this setting, 
$\ell \in L$ corresponds to the treatment applied, where $\ell=0$ indicates no treatment, and $\ell=1$ represents the standard chemotherapy treatment; 
$m\in M$ still corresponds to the overall condition of the patient ($m=0$: sound, $m=1$: main disease, $m=2$: therapeutic escape, $m=3$: death); 
the counter $k\geq0$ corresponds to the number of treatments the patient has already undergone; 
$\zeta\in [\zeta_0,D]$ still corresponds to the marker value, $\zeta_0$ being the nominal value for a sound patient,  and we further assume that if the marker reaches level $D$, the patient dies; 
$u\in  [0,\mathsf{H}]$ (in mode $m=0$) corresponds to the time since the last change of overall condition $m$.
%
Note that we deal with both a mode-augmented and time-augmented 
PDMP. 

Decision dates correspond to visits to the medical center, during which the marker is measured, and a new treatment is selected and applied until the next visit. 
At each visit, the practitioner must decide on two aspects: the interval until the next visit, $r \in \mathbb{T}=\{15, 30, 60\}$ days ($\delta=15$), and the treatment to be administered $\ell \in L$. Thus, at each decision point, the practitioner selects a decision $a = (\ell,r)$. 
The patient's follow-up period is denoted by $\mathsf{H}$, marking the real-time horizon of the decision process. 
The local characteristics of the PDMP are defined as follow, depending on the treatment $\ell$.

If no treatment ($\ell=0$) is applied, in the disease conditions ($m=1$ or $m=2$), the marker level grows exponentially and can reach the death level $D$ (boundary jump to $m=3$) in finite time. The death level is absorbing.
In the therapeutic escape condition $m=2$,  no other change of state is possible.
In the main disease condition $m=1$, a random jump to therapeutic escape is also possible.
In the remission condition $m=0$, the marker level remains at its nominal value $\zeta_0$ and the patient may randomly jump to any of the two disease conditions $m\in\{1,2\}$. Formally, the characteristics of the PDMP for $\ell=0$ are thus
\begin{equation*} 
\left\{
    \begin{aligned}
        \Phi^0_{0}(\zeta_0,t) & =\zeta=\zeta_0, \\
        \Phi^0_{1}(\zeta,t) & =\zeta e^{v_1^0t}, \\
        \Phi^0_{2}(\zeta,t) & =\zeta e^{v_2^0t},  \\  
        \Phi^0_{3}(\zeta,t) & =\zeta=D, 
    \end{aligned}
\right. \hspace{.2cm}
\left\{
    \begin{aligned}
        \lambda^0_{0}(\zeta,u) & =f_{\to1}^0(u)+f_{\to2}^0(u), \\
        \lambda^0_{1}(\zeta)& =f_{\to2}^0(\zeta), \\
        \lambda^0_{2}(\zeta) & =0, \\ 
        \lambda^0_{3}(\zeta) & =0,
    \end{aligned}
\right.\hspace{.2cm}
\left\{
    \begin{aligned}
    Q^0_{0}(B|\zeta_0,u)&=\frac{f_{\to m^\prime}^0(u)}{f_{\to1}^0(u)+f_{\to2}^0(u)}\1_{B}(m^\prime,\zeta_0), \\
    Q^0_{1}(B|\zeta)&=\1_{\{\zeta<D\}}\1_{B}(2,\zeta)+\1_{\zeta=D}\1_{B}(3,D),  \\
    Q^0_{2}(B|D)&=\1_{B}(3,D),
    \end{aligned}
\right. 
\end{equation*}
for any Borel subset $B$ of $E_M$, where $f_{\to1}^0$ and $f_{\to2}^0$ are intensity functions corresponding respectively to a jump from mode $0$ to mode $1$ and $2$. For instance, they can be chosen as Weibull distributions as in Section \ref{subsubsec:medicPDsMP}.

Under treatment ($\ell=1$), in remission ($m=0$) the marker level remains at its nominal value and the patient may only randomly jump to the therapeutic escape condition $m=2$.
In the disease condition $m=1$, the marker level decreases exponentially and can reach the nominal level $\zeta_0$ ($m=0$) in finite time or randomly jump to the therapeutic escape condition $m=2$.
In the therapeutic escape condition $m=2$, the marker level grows exponentially and can reach the death threshold $D$ ($m=3$) in finite time, no other change of state is possible. Again, the death state is absorbing. Formally, the characteristics of the PDMP for $\ell=1$ are thus
\begin{equation*} 
\left\{
    \begin{aligned}
        \Phi^1_{0}(\zeta_0,t)&=\zeta=\zeta_0,\\
        \Phi^1_{1}(\zeta,t)&=\zeta e^{-v_1^1t}, \\
        \Phi^1_{2}(\zeta,t)&=\zeta e^{v_2^1t}, \\   
        \Phi^1_{3}(\zeta,t)&=\zeta =D, \\  
    \end{aligned}
\right. \hspace{.2cm}
\left\{
    \begin{aligned}
        \lambda^1_{0}&=f_{\to2}^1(u), \\
        \lambda^1_{1}&=f_{\to2}^1(\zeta), \\
        \lambda^1_{2}&=0,\\    
        \lambda^1_{3}&=0,\\
    \end{aligned}
\right.
\hspace{.2cm}
\left\{
    \begin{aligned}
        Q^1_{0}(B| \zeta_0,u)&=\1_{B}(2,\zeta_0), \\
        Q^1_{1}(B|\zeta=\zeta_0)&=\1_{B}(0,\zeta_0,0), \\
        Q^1_{1}(B|\zeta\neq\zeta_0)&=\1_{B}(2,\zeta), \\
        Q^1_{2}(B|D)&=\1_{B}(3,D),   
    \end{aligned}
\right. 
\end{equation*}
for any Borel subset $B$ of $E_M$, where again $f_{\to2}^1$ is the intensity function corresponding to a jump from mode $0$ to mode $2$. Figure~\ref{fig:egimpulsecontrolledPDMP_mdp_full} illustrates a trajectory from this PDMP with a given strategy including treatment and next decisions dates.\\

\begin{figure}[tp]
\centering
\includegraphics[width=.5\textwidth]{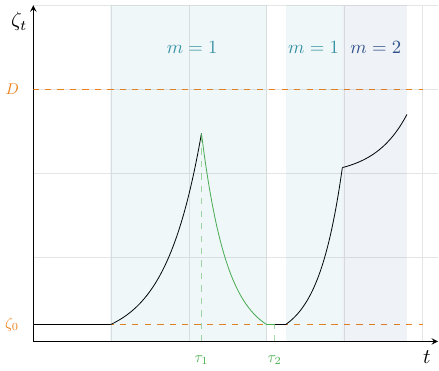}
\caption{{\bf Medical example as a mode-augmented PDMP.} Example of a trajectory from a mode-augmented PDMP, with a strategy including only visits every 15 days. Background colors indicate patient condition mode, while colors of the trajectory indicate treatment modes (black: $\ell=0$, green: $\ell=1$). At time $t=0$ the process starts with action $a=(0,15)$; at time $\tau_1$ the decision-maker starts treatment, \textit{i.e.} $a=(1,15)$, at time $\tau_2$ the decision-maker stops the treatment, \textit{i.e.} $a=(0,15)$.}
\label{fig:egimpulsecontrolledPDMP_mdp_full}
\end{figure}

\textbf{Formulation as an MDP}
We assume that visits take place at discrete dates $t_n$, where the time lapse between two visits can be $15$, $30$ or $60$ days. 
At most $H = \frac{\mathsf{H}}{15}$ visits can occur. 
Moreover, decision-related constraints appear. 
The last visit must take place at the end $H$ of the follow-up (We assume $\mathsf{H}$ is a multiple of $15$). 
The associated MDP is $\langle S, A, H, K, P, c,C \rangle$. Its characteristics are defined as follows.
\begin{itemize}
    \item The state space $S$ corresponds to the state of a patient $x \in E_M$, a counter of time since the beginning $t$ and the cemetery state $\Delta$: $S = E_M\times\{0,\delta,\ldots ,(H-1)\delta\} \cup \{\Delta\}$,
    \item The action space $A$ is expressed either as a decision pair $a = (\ell,r)$, where $r \in \mathbb{T}=\{15,30,60\}$ is the delay until the next visit and $\ell \in L=\{0,1\}$ is the treatment choice or as the artificial action $a = \cd$ : $A=L\times\mathbb{T}\cup\{\cd\}$.
    \item The horizon $H$ is equal to $\mathsf{H}/15$  decision dates.
    \item The constraints $K(s) \subset A$ is the space of admissible decisions in state $s$, 
\begin{equation*}
\begin{split}
\mathbb{K}(s) = \left\{
\begin{aligned}
    & \cd && \text{if } s=(m,\mathsf{x},t) \text{ and } m = 3 \text{ or } s=\Delta,\\
    & L\times(\mathbb{T} \cap \{\delta,\dots,H\delta-t\}) && \text{ otherwise. }
\end{aligned}
\right.
\end{split}
\end{equation*}
    \item The transition function $P$ can be written as a combination of the PDMP flow, jump intensity and Markov kernel. Detailed analytic formulas are omitted here, but notice that the MDP transition function is set according to the PDMP parameters. 
    \item The cost function $c$ as introduced in Section~\ref{sssec:PDMDP} is approximated by 
    $$c(s,a,s')= (\zeta-\zeta')\times r + C_V +C_\ell,$$ 
    where $a=(r,\ell)$, $\zeta$ is the second coordinate of $s$ and $\zeta'$ that of $s'$ (if both $s$ and $s'$ differ from $\Delta$), \textit{i.e.} the running cost is approximated by the histogram estimator of the area under the curve, and the impulse cost corresponds to a visit cost $C_V$ and a cost of treatment $C_\ell$. The terminal cost is $C_D>0$ if the patient has died, $0$ otherwise.
\end{itemize}
Figure~\ref{fig:egimpulsecontrolledPDMP_mdp} illustrates a trajectory of the PDMP controlled as an MDP, with perfect observations on visit dates.
%
\begin{figure}[tp]
\centering
\includegraphics[width=.5\textwidth]{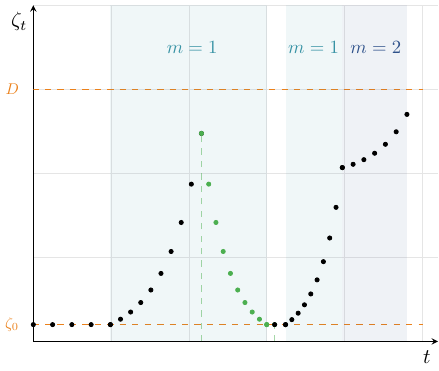}
\caption{{\bf Trajectory of a PDMP controlled as an MDP.} Dots indicate the perfect observations obtained only at visit dates. Here again, background colors indicate the patient condition mode, while colors of the dots indicate the treatment mode. In this example, visits are made every 15 days (\textit{i.e.} $r=15$) and the treatment changes twice.}
\label{fig:egimpulsecontrolledPDMP_mdp}
\end{figure}
%
\subsection{Partially observed controlled PDMP as POMDP}\label{subsec:POPDMP-POMDP}
%
While properly defining an impulse control problem for a  hidden PDMP is very challenging, the framework of POMDPs allows to deal with partially hidden MDPs. Transforming an impulse control problem for partially observed PDMPs into a POMDP is done in a similar manner as in the perfect observation case, and again greatly simplifies the problem statement. 
Here we assume that observations are only available at the decision dates $t_n$, and are given by $Y_n=F(X_{t_n}) + \epsilon_n$ with $(\epsilon_n)$
real-valued independent and identically distributed random variables with density $f$ independent from the PDMP and the actions, and $F$ a real-valued link function (to simplify the exposition. Both the noise and link function could be multi-dimensional). Only partial information about the process $X$ is obtained through $Y$, and we hence consider $X_{t_n}$ as hidden. We also assume that the time since the beginning of the follow-up is observed through a variable $w_n=t_n$. The objective is now to control the PDMP through policies that may only rely on the observations available up to each decision date. 
We first state this problem as a POMDP in Section \ref{sec:POPDMPdef}, then extend our running medical example to this more realistic context in Section \ref{ssec:medic-pdmp-pomdp}.
%
\subsubsection{Formulation as a POMDP}
\label{sec:POPDMPdef}
%
We keep the restrictions on the impulses from Section \ref{subsec:strategies}, and now give a POMDP formulation of the partially observed impulse control problem for a PDMP.
%
\begin{defin}[Partially observed controlled PDMP as a POMDP]
    \label{def:pdmp-to-pomdp}
Let $\langle E, \Phi, \lambda, Q\rangle$ be a mode-augmented controlled PDMP under partial observations. The corresponding POMDP is defined by
 $\langle S, A, H, P, \Omega,  O, K, c,C, \omega_0 \rangle$ with the following characteristics.
\begin{itemize}
\item The state space $S$, the action space $A$, the transition kernel $P$, the cost functions $c$ and $C$ and the horizon $H$ are the same as those of Definition \ref{def:pdmp-to-mdp}.
\item The observation space is $\Omega=(\mathbb{R} \times[0,\mathsf{H}] )\cup\{\Delta\}$. We will denote $\omega=(y,t) \in \Omega$.
\item The observation function $O$ corresponds to the observation function of the PDMP, which may or may not depend on the decision $a$. Here, it has density 
\begin{align*}
    o(s'=(x',t'),a,\omega=(y,t)) = f(y-F(x')) \1_{t=t'},
\end{align*}
with respect to the Lebesgue measure on $\mathbb{R}$.
\item The constraints set $ K\subset \Omega \times  A$, satisfy the same conditions as in Definition \ref {def:pdmp-to-mdp} and may only depend on the observations.
\item {The initial observation is $\omega_0=x_0+\epsilon_0$.} 
\end{itemize}
\end{defin}
%
The impulse control for the partially observed PDMP is equivalent to its POMDP counterpart, where both the state space $S$ and the observation space $\Omega$ are continuous. 
%
\subsubsection{Running medical example as a partially observed controlled PDMP formalized as a POMDP}
\label{ssec:medic-pdmp-pomdp}
%
We consider the same example as in Section~\ref{ss:ex_PDMP_MDP}, but now assume that the patient condition $m$ and time since the last change of condition $u$ (in mode $m=0$) are hidden, while the marker is only observed through noise at each visit date, \textit{i.e} at each decision date $t_n$ we observe $Y_n=F(X_{t_n}) + \epsilon_n = F(\zeta_{t_n}) + \epsilon_n$ where the $(\epsilon_n)_n$ are independent Gaussian random variables independent from the PDMP and decisions.
However, we assume that the  death (mode $m=3$) is perfectly observed. 
In this context, the partially observed  impulse control problem for our PDMP can be formulated as a POMDP with the following characteristics : 
the observation space is $\Omega=\mathbb{R}\times[0,H]\cup \Delta$; 
the observation function does not depend on decision $a$, and has density $o(s'=(x',t'),a,\omega=(y,t))=f(y-F(x'))\1_{t=t'}$ with repect to the Lebesgue measure on $\mathbb{R}$; 
the constraints set is the same as in Example~\ref{ss:ex_PDMP_MDP} as constraints did not depend on the hidden part of the process; 
we further assume that the follow-up of the patient starts when they reach their first remission phase, \textit{i.e.} $b_0=\delta_{(0,\zeta_0,0)}$.
Figure~\ref{fig:egimpulsecontrolledPDMP_pomdp} illustrates an observed trajectory (through noise) of the original PDMP controlled via optimization of the equivalent POMDP value function.
%
\begin{figure}[tp]
\centering
\includegraphics[width=.5\textwidth]{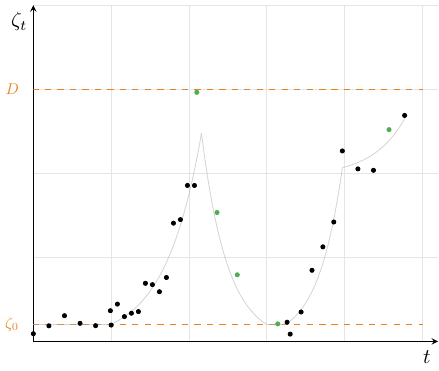}
\caption{{\bf Example of observed trajectory in a POMDP framework.} Patient condition modes are now fully hidden, dots indicate (noisy) marker measurements and colors of the dots indicate treatment choices. In this example, the POMPD is controlled by optimization of the value function, and visit dates are allowed to change at each new observation. }
\label{fig:egimpulsecontrolledPDMP_pomdp}
\end{figure}
%

Solving this POMDP is beyond the scope of this paper. Interested readers are referred to \cite{cleynen_numerical_2021} and \cite{de_saporta_monte-carlo_2024} for the resolution of similar problems through discretizations or Monte Carlo planning respectively.
\color{black}
\section{Concluding remarks and remaining challenges}
\label{ssec:open}
%
We conclude this survey by going back to our main challenges and stating some open questions regarding the control of PDMPs that we believe can be tackled through the MDP formulation as well as interesting new questions arising in MDP from this PDMP perspective. 
As shown in the previous section, the formalism of MDPs makes it possible to address very challenging problems for controlled PDMPs. Conversely, deriving MDPs from PDMPs opens new questions for continuous-state space MDPs and their practical resolution. We briefly present some of these challenges in Sections  \ref{ssec:openPDMP} and \ref{ssec:openMDP} respectively. 
%
\subsection{Open questions in controlled PDMPs that might be tackled using MDPs}
\label{ssec:openPDMP}
%
The very well established field of MDPs, together with its wide range of available resolution algorithms, represents a unique opportunity to tackle open questions that prevail in the framework of impulse control for PDMPs. A few examples have already been described above, and we provide here a short and non-exhaustive list of problems that we believe might be formalized and solved in an MDP context. 

As mentioned in Sections~\ref{ssec:impulse} and~\ref{subsec:strategies}, the very task of formalizing the definition of an impulse controlled PDMP in the partially observed framework is rather technical (even in the fully observed case). Embedding it in the MDP framework provides a very natural environment for this formalization, since by construction strategies are admissible policies, and value functions are natural elements of MDP resolutions. 

The MDP framework provides a wide range of algorithms for the efficient resolution of controlled PDMPs. As mentioned in Section~\ref{subsec:PDMP-MDP} and ~\ref{subsec:POPDMP-POMDP}, a (partially observed) controlled PDMP can be cast as an MDP with a continuous state-space. From there, several possibilities can be considered: discretizing the initial process to reduce it to a finite state space, for which exact resolution through dynamic programming is possible, or using approximate resolution algorithms designed for MDPs with finite but very large state spaces, or continuous state spaces. Those two strategies have been developed for variants of the medical example, see for instance \cite{cleynen_change-point_2018} for an exact resolution of the discretized process in the framework of optimal stopping, \cite{cleynen_numerical_2021} for an exact resolution of the discretized process in the general impulse control framework, and \cite{de_saporta_monte-carlo_2024} for an approximate resolution based on the adaptation of POMCP for PDMPs. 

Noting that (partially observed) PDMPs are both easy to simulate and can be cast as POMDPs, one can take advantage of the huge amount of resources (both in the literature and software) available in the field of (deep)-Reinforcement Learning (see, \textit{e.g.} \cite{mnih_playing_2013, sutton_reinforcement_2018}). The field of RL is mature and there exist several RL application programming interfaces (\textit{e.g.} Gymnasium \cite{brockman_openai_2016}), connecting POMDP simulation models, to libraries of (deep) RL algorithms (\textit{e.g.} RLlib \cite{liang_rllib_2018}). 
The embedding of PDMPs in the MDP paradigm hence allows to address the scalability issues of standard resolution strategies for PDMPs, by providing efficient resolution algorithms, including exact and approximate approaches, that provide excellent performances.

One may also want to consider impulse control for PDMPs where the impulsion does not affect the dynamics of the process (for instance deciding how to allocate an observation budget to learn about a process without affecting trajectories). Here again, while the PDMP framework allows to model continuous-time trajectories with few parameters that would all be interpretable, the MDP framework should provide the tools necessary to the formalism and resolution of the control problem of continuous-time trajectories that are only observed at discrete times.

Finally, the MDP framework also provides the opportunity to solve impulse control problems where the PDMP model is only partially known, which is different from state partial observability. PDMPs are typically described by three parameters: the flow $\Phi$, the jump intensity $\lambda$, and the jump kernel $Q$. 
Solving MDP problems with ill-known parameters is known as {\em controlling while learning}, since we have to balance optimizing the MDP with a currently \emph{wrong} model with \emph{improving our knowledge about the model}.
We have shown in Section~\ref{sub:bamdp} how Bayes-Adaptive MDP allow to solve MDP for which the transition kernel is unknown or partially known, by modeling the MDP with ill-known parameters as a new MDP with additional state variables (the hyperparameters). BAMDP are a \emph{control while learning} approach particularly adapted to solve impulse control problems where the jump kernel $Q$ is only (partially) known. One can expect that an adaptation of BAMDPs could allow the consideration of other impulse control problems with partially known models, for example, those in which certain parameters of the flow or jump intensity are unknown. 
Once we have built the BAMDP model, any MDP-solving method (including RL methods can be applied to solve the augmented problem). 
%
\subsection{Interesting questions in MDPs arising from converted PDMPs}
\label{ssec:openMDP}
%
In the MDP/RL community, there is extensive literature on solving problems with finite or continuous time and state spaces (see Section~\ref{sec:algo}). Recent advancements in deep RL have also led to practical solutions for these problems. However, solving MDPs with continuous state spaces or in continuous time remains a significant challenge, and practical algorithms often fail to perform well in these scenarios.

This difficulty is particularly obvious when the transition function, $P(\cdot|s,a)$,  is a continuous probability distribution. Monte Carlo algorithms, for instance, generally struggle in such settings because they cannot effectively explore the resulting tree, which has infinitely many successor nodes at each level. Similarly, Reinforcement Learning approaches require approximations of the $Q$-values. While deep RL techniques model $Q$-values using neural networks and are applicable in these cases, they demand immense computational resources and, when relying on real-world data rather than simulations, require vast quantities of data. 

Using MDPs derived from piecewise-deterministic Markov processes (PDMPs) offers a way to address complex, realistic problems with models that are both easily parameterized and straightforward to simulate. Additionally, the min-expectation tree resulting from applying Monte Carlo Tree Search (MCTS) to PDMPs maintains a finite, often small branching factor (see \textit{e.g.} \cite{de_saporta_monte-carlo_2024}). 

In other words, controlled PDMPs serve as an intriguing benchmark for evaluating MDP solution algorithms. Their complexity lies between that of finite MDPs, which are relatively easy to solve, and MDPs with continuous time, state, or decision spaces, where existing solution methods typically lack quality guarantees. 
Thus, we believe this paper makes a meaningful contribution to the MDP community by addressing these intermediate cases and providing insights into solving such challenging problems.
%
\bibliographystyle{plain}
\bibliography{references.bib}{}
\end{document}